\documentclass[%
superscriptaddress,
nofootinbib,
notitlepage,
amsmath,amssymb,
aps, physrev,
]{revtex4-2}

\usepackage{makecell}
\usepackage[utf8]{inputenc}
\usepackage[english]{babel}
\usepackage[T1]{fontenc}
\usepackage{amsmath}
\usepackage{hyperref}
\usepackage[capitalize]{cleveref}
\usepackage{tikz}
\usepackage{lipsum}
\usepackage{amssymb}
\usepackage{amsthm}
\usepackage{enumitem}
\usepackage[labelfont=bf]{caption}

\newtheorem{theorem}{Theorem}

\newtheorem{lemma}{Lemma}

\newtheorem{remark}{Remark}
\newtheorem{problem}{Problem}
\newtheorem{assumption}{Assumption}

\crefname{theorem}{Theorem}{Theorems}
\Crefname{theorem}{Theorem}{Theorems}

\crefname{definition}{Definition}{Definitions}
\Crefname{definition}{Definition}{Definitions}

\crefname{lemma}{Lemma}{Lemmas}
\Crefname{lemma}{Lemma}{Lemmas}

\crefname{corollary}{Corollary}{Corollaries}
\Crefname{corollary}{Corollary}{Corollaries}

\crefname{proposition}{Proposition}{Propositions}
\Crefname{proposition}{Proposition}{Propositions}

\crefname{remark}{Remark}{Remarks}
\Crefname{remark}{Remark}{Remarks}

\crefname{problem}{Problem}{Problems}
\Crefname{problem}{Problem}{Problems}

\crefname{assumption}{Assumption}{Assumptions}
\Crefname{assumption}{Assumption}{Assumptions}

\crefname{algocf}{Algorithm}{Algorithms}
\Crefname{algocf}{Algorithm}{Algorithms}
\crefname{algocfline}{Algorithm}{Algorithms}
\Crefname{algocfline}{Algorithm}{Algorithms}

\usepackage{subcaption} 
\usepackage[ruled,vlined,linesnumbered,lined]{algorithm2e}
\usepackage[braket, qm]{qcircuit}
\usepackage{graphicx}
\usepackage{tikz,lipsum,lmodern}

\usepackage{xcolor}
\definecolor{mycyan}{RGB}{0,153,153}
\definecolor{myred}{RGB}{205,46,46}
\definecolor{purplegray}{RGB}{164,67,255}
\definecolor{darkgreen}{RGB}{0,100,0}
\definecolor{mydarkgreen}{RGB}{0,100,0}
\usepackage[normalem]{ulem}
\usepackage{booktabs,tabularx}

\newcommand{\ZAB}{\textcolor{red}{\mathbf{1}}}
\newcommand{\ZBC}{\textcolor{mydarkgreen}{\mathbf{1}}}
\newcommand{\ZCA}{\textcolor{blue}{\mathbf{1}}}

\begin{document}

\title{Interpolation-based coordinate descent method for parameterized quantum circuits}

\author{Zhijian Lai}
\email{lai\_zhijian@pku.edu.cn}
\homepage{https://galvinlai.github.io}
\affiliation{Beijing International Center for Mathematical Research, Peking University, Beijing, China}

\author{Jiang Hu}
\email{hujiangopt@gmail.com}
\homepage{https://hujiangpku.github.io}
\affiliation{Yau Mathematical Sciences Center, Tsinghua University, Beijing, China}

\author{Taehee Ko}
\email{kthmomo@kias.re.kr}
\homepage{https://sites.google.com/view/taeheeko}
\affiliation{School of Computational Sciences, Korea Institute for Advanced Study, Seoul, South Korea}

\author{Jiayuan Wu}
\email{jyuanw@wharton.upenn.edu}
\affiliation{Wharton Department of Statistics and Data Science, University of Pennsylvania, Philadelphia, PA, USA}

\author{Dong An}
\email{dongan@pku.edu.cn}
\homepage{https://dong-an.github.io}
\affiliation{Beijing International Center for Mathematical Research, Peking University, Beijing, China}

\date{\today}
\begin{abstract}
Parameterized quantum circuits (PQCs) are ubiquitous in the design of hybrid quantum-classical algorithms. In this work, we propose an interpolation-based coordinate descent (ICD) method to address the parameter optimization problem in PQCs. The ICD method provides a unified framework for existing structure optimization techniques such as Rotosolve, sequential minimal optimization, ExcitationSolve, and others. ICD employs interpolation to approximate the PQC cost function, effectively recovering its underlying trigonometric structure, and then performs an argmin update on a single parameter in each iteration. In contrast to previous studies on structure optimization, we determine the optimal interpolation nodes to mitigate statistical errors arising from quantum measurements. Moreover, in the common case of $r$ equidistant frequencies, we show that the optimal interpolation nodes are equidistant nodes with spacing $2\pi/(2r+1)$ (under constant variance assumption), and that our ICD method simultaneously minimizes the mean squared error, the condition number of the interpolation matrix, and the average variance of the approximated cost function. We perform numerical simulations and test on the MaxCut problem, the transverse field Ising model, and the XXZ model. Numerical results imply that our ICD method is more efficient than the commonly used gradient descent and random coordinate descent method. 
\end{abstract}

\maketitle

\section{Introduction}

Parameterized quantum circuits (PQCs) are central to a wide range of hybrid quantum-classical algorithms, including variational quantum algorithms (VQAs) and quantum machine learning (QML) models.
VQAs have found applications across diverse fields: the variational quantum eigensolver (VQE) has been used to determine ground state energies of molecular systems and to simulate quantum dynamics \cite{peruzzo2014variational, kandala2017hardware,grimsley2019adaptive,Shang2023Schrodinger,Yuan2019theory,McArdle2019variational,Endo2020variational}, while the quantum approximate optimization algorithm (QAOA) shows considerable promise for tackling combinatorial optimization problems \cite{farhi2014quantum,zhou2020quantum,blekos2024review}.
Meanwhile, QML has been developed for a wide range of tasks, including classification, regression, and generative modeling \cite{schuld2019quantum,schuld2020circuit, perdomo2018opportunities,killoran2019continuous,benedetti2019parameterized}.
A prominent subclass of QML models is quantum neural networks (QNNs), which are hybrid architectures that encode classical input data into quantum states or gate spaces and utilize PQCs to learn target functions.
Numerous studies \cite{perez2020data,yu2022power,havlivcek2019supervised,abbas2021power} have shown that QNNs possess strong expressive power, capable of approximating arbitrary functions.

In all these approaches, the quantum circuit is parametrized by a set of classical variables. After executing the circuit and measuring its output on quantum hardware, one evaluates a cost function that reflects the current performance. A classical optimizer then updates the parameters iteratively to minimize this cost and improve the result. 
Specifically, in this work,  we consider a $q$-qubit system with $N:=2^q$. 
Without loss of generality, finding the optimal parameters of a PQC ultimately reduces to solving the following unconstrained optimization problem:
\begin{equation}\label{eq:cost-function}
    \min _ {\boldsymbol{\theta} \in \mathbb{R}^m} f(\boldsymbol{\theta})=\langle \psi_0|U(\boldsymbol{\theta})^{\dagger} M U(\boldsymbol{\theta})|\psi_0\rangle.
\end{equation}
Here,  $U(\boldsymbol{\theta}) \in \mathbb{C}^{N \times N}$ is a PQC that depends on a set of classical parameters $\boldsymbol{\theta}=\left[ \theta_1, \ldots, \theta_m\right]^{\dagger} \in \mathbb{R}^m$. Typically, the circuit $U(\boldsymbol{\theta})$ is applied to a fixed and easily prepared initial state $|\psi_0\rangle \in \mathbb{C}^{N}$, yielding the output state $U(\boldsymbol{\theta})|\psi_0\rangle $ in a quantum device. In the context of quantum mechanics, $f(\boldsymbol{\theta})$ is precisely the expectation value of the Hermitian observable $M \in \mathbb{C}^{N \times N}$, measured with respect to that output state. As in many studies \cite{sweke2020stochastic,wierichs2022general,mari2021estimating,ding2024random}, we consider the typical PQC structure as 
\begin{equation}\label{eq-Utheta}
    U(\boldsymbol{\theta})=V_m U_m\left(\theta_m\right) \cdots V_1 U_1\left(\theta_1\right),
\end{equation}
where $V_j$ are fixed arbitrary gates, and $U_j\left(\theta_j\right)$ are rotation-like gates, defined as
\begin{equation}\label{eq-eiHtheta}
    U_j\left(\theta_j\right)=e^{i H_j \theta_j }, \quad j = 1,\ldots,m,
\end{equation}
for some Hermitian generators $H_j \in \mathbb{C}^{N \times N}$. Notice that each $U_j$ is a single-parameter gate and fully captures the dependence on univariate $\theta_j \in \mathbb{R}$. 


\subsection{Optimization methods}

For optimizing the parameters in PQCs, the main cost lies in the evaluation of the cost function, namely, the number of calling of $\boldsymbol{\theta}\mapsto f(\boldsymbol{\theta})$ under different $\boldsymbol{\theta}$. This process is the only part of PQCs that relies on a quantum device. Effective optimization techniques can achieve faster reductions of the cost values with fewer function evaluations, thereby improving the efficiency of the whole PQCs. 
Thus, this paper primarily focuses on the classical algorithmic approach to solve \cref{eq:cost-function}. 

Here, we consider three classes of optimization techniques: derivative-free methods, gradient-based methods, and structure optimization methods.

\paragraph{Derivative-free methods.}
Derivative-free methods, such as COBYLA \cite{powell1994direct}, Nelder-Mead \cite{nelder1965simplex}, Powell \cite{powell1964efficient} and SPSA \cite{spall2000adaptive}, update parameters by directly searching or applying random perturbations in parameter space, thereby obviating the need for explicit gradient information. Although derivative-free methods are simple to implement in practice, empirical studies \cite{pellow2021comparison,lockwood2022empirical} have shown that gradient-based methods outperform them when only sampling noise is present.

\paragraph{Gradient-based methods}
Gradient-based methods obtain parameter gradients via the parameter shift rule (PSR) or finite-difference (FD) approximations, then employ advanced optimizers, such as BFGS \cite{nocedal2006numerical}, L-BFGS \cite{byrd1995limited}, Adam \cite{kingma2014adam}, AMSGrad \cite{reddi2019convergence} and quantum natural gradient \cite{stokes2020quantum}, to accelerate convergence. These methods offer convergence guarantees and excel in moderately noisy environments \cite{pellow2021comparison,lockwood2022empirical}. 
The well-known PSR technique \cite{crooks2019gradients,mari2021estimating,wierichs2022general,kyriienko2021generalized,hai2023lagrange,markovich2024parameter,hoch2025variational} gives the exact estimation of derivatives by evaluating the cost function in \cref{eq:cost-function} at a finite number of shifted parameter positions and combining those results linearly. This unbiased derivative estimation approach provides a solid foundation for various gradient-based techniques. See \cref{app-psr} for a review of PSR. Typically, derivatives are computed using PSR rather than FD. A comparison of PSR's advantages over FD also can be found in \cref{app-psr}. 

For the sake of comparison with our proposed algorithm, we focus here on the two canonical gradient-based methods: stochastic gradient descent (SGD)\footnote{Since our gradients can only be unbiased estimators rather than exact values, the simplest form of gradient descent we can use is stochastic gradient descent.} \cite{sweke2020stochastic} and random coordinate descent (RCD) \cite{ding2024random}. Specifically, the SGD requires the full gradient $\nabla f(\boldsymbol{\theta})$ at each iteration, followed by an update to all parameters in the direction of $-\nabla f(\boldsymbol{\theta})$, scaled by a constant learning rate. RCD, on the other hand, randomly selects a single coordinate $j$ at each iteration, computes the partial derivative $\partial_{j} f\left(\boldsymbol{\theta}\right)$, and updates only that coordinate by $-\partial_{j} f\left(\boldsymbol{\theta}\right)$, scaled by a constant. 

\paragraph{Structure optimization methods.}

In recent years, structure optimization strategies have attracted increasing interest in the training of PQCs, with one of the most prominent methods being Rotosolve \cite{ostaszewski2021structure}, which has been implemented in several open-source frameworks, including PennyLane \cite{pennylaneRotosolve} and TensorFlow Quantum \cite{tfqRotosolve}.
Rotosolve models the cost function associated with each parameter as a simple sinusoidal function, $\theta_j \mapsto f (\boldsymbol{\theta}) = A \sin (\theta_j + B) + C$, where $A$, $B$, and $C$ are unknown coefficients. It identifies the three coefficients using three function evaluations, then achieving a global minimization update for that parameter. This idea closely parallels earlier works such as Algorithm 1 in \cite{vidal2018calculus} and Jacobi+Anderson \cite{parrish2019jacobi}, albeit expressed in different terminologies. Around the same time, \cite{nakanishi2020sequential} referred to this approach as sequential minimal optimization (SMO) and demonstrated its equivalence to Rotosolve.
More recently, \cite{jager2024fast} introduced ExcitationSolve, an optimizer that can be viewed as a generalization of Rotosolve to excitation operator $(H^3 = H)$ based ansatz, particularly well suited to physically motivated UCC-type circuits. 
In a word, the core idea of Rotosolve has been independently articulated and named by multiple research groups to address various types of PQC optimization tasks.
Numerical experiments in these studies \cite{parrish2019jacobi,nakanishi2020sequential,ostaszewski2021structure,jager2024fast} have shown that, compared to both derivative-free and gradient-based methods, structure optimization approaches can more efficiently find optimal parameters under limited quantum resources (see the Baseline algorithms row in \cref{tab:comparison}).

In this work, we observe that all existing structure optimization methods share a common underlying principle: selecting a single parameter to update, reconstructing the cost function using interpolation, and then performing a global optimization on a classical computer. We refer to this general framework as the Interpolation-based Coordinate Descent (ICD) method. ICD can be regarded as the general extension of structure optimization techniques. In \cref{tab:comparison}, we summarize known structure optimization methods and show how their characteristics can be unified under the ICD framework.

\begin{table}[htbp]
\centering
\renewcommand{\arraystretch}{1.3}
\setlength{\tabcolsep}{6pt}
\begin{tabularx}{\textwidth}{
  @{} >{\raggedright\arraybackslash}X
  @{\extracolsep{1pt}}
  >{\centering\arraybackslash}X
  >{\centering\arraybackslash}X
  >{\centering\arraybackslash}X
  >{\centering\arraybackslash}X
  >{\centering\arraybackslash}X
  >{\centering\arraybackslash}X
  @{}
}
\toprule
Methods
& Algorithm 1 in \cite{vidal2018calculus} 
& Jacobi+Anderson \cite{parrish2019jacobi} 
& SMO \cite{nakanishi2020sequential} 
& Rotosolve \cite{ostaszewski2021structure} 
& ExcitationSolve \cite{jager2024fast} 
& ICD (this work) \\
\midrule
Gate generator type 
  & Eigenvalues of $H$ are integers $\{k_j\}_{j}$ 
  & $H^2=I$ & $H^2=I$ & $H^2=I$ & $H^3=H$ 
  & General Hermitian $H$ \\
- Frequency set
  & {\footnotesize $D=\{| k_i-k_j|>0\}$} 
  & $\{2\}$ & $\{2\}$ & $\{2\}$ & $\{1,2\}$ 
  & \cref{eq-freq-set}\\
- Equispaced? 
  & Yes or No & Yes & Yes & Yes & Yes 
  & Yes or No  \\
- Number of Fourier coeffs.
  & $2 |D|+1$ & 3 & 3 & 3 & 5 
  & $2r+1$ \\
\midrule
Interpolation nodes spacing 
  & Arbitrary or $\tfrac{2\pi}{2|D|+1}$ 
  & $\tfrac{\pi}{3}$ & $\tfrac{\pi}{2}$ & $\tfrac{\pi}{2}$ & $\tfrac{2\pi}{5}$ 
  & $\tfrac{2\pi}{2r+1}$ or solve \cref{pro-min-mse} \\
- Optimal?
  & Maybe not & No & No & No & Yes 
  & Yes \\
\midrule
Subproblem solution 
  & Arbitrary & Closed form & Closed form & Closed form & Eigenvalue method 
  & Arbitrary or Eigenvalue method \\
Reuse previous iteration? 
  & No & No & Yes & No & Yes 
  & No (\cref{alg-standard-ICD}) or Yes (\cref{alg-reduced-ICD})\\
Multivariable version?
  & No & Yes & Yes & No & Yes 
  & No \\
- Number of Fourier coeffs.
  & - & $3^K$ & $3^K$ & - & $5^K$ 
  & - \\
\midrule
Baseline algorithms
  & No experiments & Powell; L-BFGS & Nelder-Mead, Powell, SPSA; BFGS, CG 
  & SPSA; Adam & COBYLA, SPSA; Adam, BFGS, SGD 
  & RCD, SGD \\
\bottomrule
\end{tabularx}
\caption{
Comparison of various structure optimization methods under the ICD framework.
\textit{Frequency set} $\{2\}$ corresponds to the $\tfrac{1}{2}$-scaled version of $\{1\}$, due to the use of $e^{i \theta H / 2}$ in those studies. \textit{Optimal?} indicates whether the interpolation nodes correspond to the optimal spacings described in \cref{thm-first-veiw,thm-second-veiw,thm-third-veiw} later. \textit{Subproblem solution} describes the method used to solve the single parameter argmin update. \textit{Reuse previous iteration?} denotes whether the first interpolation node reuses results from the previous iteration --- if yes, it corresponds to our reduced ICD of \cref{alg-reduced-ICD} in \cref{app-pratical-ICD}; if not, it is the standard ICD of \cref{alg-standard-ICD}. \textit{Multivariable version?} indicates whether the interpolation approach is extended to recover the multivariate function over $K \geq 2$ parameters, though this comes at a cost of exponential scaling in $K$. \textit{Baseline algorithms} lists the methods used for comparison in their experimental evaluation.
}
\label{tab:comparison}
\end{table}

\subsection{Overview of ICD framework} 

Now, we briefly outline the main procedure and key features of the ICD method.
Similar to RCD, our ICD method randomly selects and updates one parameter at each iteration. 
However, ICD update strategy is based on the following observation: according to \cite{wierichs2022general}, the dependence of cost function in \cref{eq:cost-function} on single parameter, say $\theta_j$, can be expressed as a finite Fourier series, represented as a linear combination of sine and cosine functions as 
\begin{equation}\label{eq-653}
\theta_j \mapsto  f\left(\boldsymbol{\theta}\right)=\frac{1}{\sqrt{2}} a_0+\sum_{k=1}^{r_j}\left[a_k \cos \left(\Omega_k^j \theta_j\right)+b_k \sin \left(\Omega_k^j \theta_j\right)\right], \quad \text{fix other $m-1$ parameters},
\end{equation}
where $a_0, a_k$ and $b_k$ are some unknown coefficients, and constants $r_j$ and $\{\Omega_k^j\}_{k=1}^{r_j}$ are fully determined by $H_j$ corresponding to $\theta_j$. We will give detail discussion for above expression in \cref{sec_obs2}. 
In fact, as mentioned in several studies \cite{nemkov2023fourier,fontana2022efficient,okumura2023fourier,stkechly2023connecting}, the PQC cost function is essentially a truncated multivariate Fourier series. In its complex form, it can be expressed as
$f(\boldsymbol{\theta})=\sum_{\boldsymbol{k} \in \mathbb{Z}^m,\left|k_j\right| \leq r_j} c_{\boldsymbol{k}} e^{i\boldsymbol{k} \cdot \boldsymbol{\theta}}$.
Thus, when considering only a single variable and applying Euler formula, it can be rewritten in the form of \cref{eq-653}. However, existing literature has not explored how to leverage this property from an optimization perspective.

We now employ the interpolation method to recover the true Fourier coefficients $a_0$, $a_k$, and $b_k$ with the greatest possible accuracy.
Once these estimated coefficients are obtained, existing solvers on classical computers can be employed to minimize the function value with respect to the selected $\theta_j$. When $\Omega_k^j = k$, the so-called eigenvalue method stated in \cref{app-eig-method} can solve this one-dimensional minimization problem globally and exactly. The overall ICD process is shown in \cref{fig:process_oicd}.

\begin{figure}[htbp]
  \begin{subfigure}{\textwidth}
    \centering
    \includegraphics[width=0.9\linewidth]{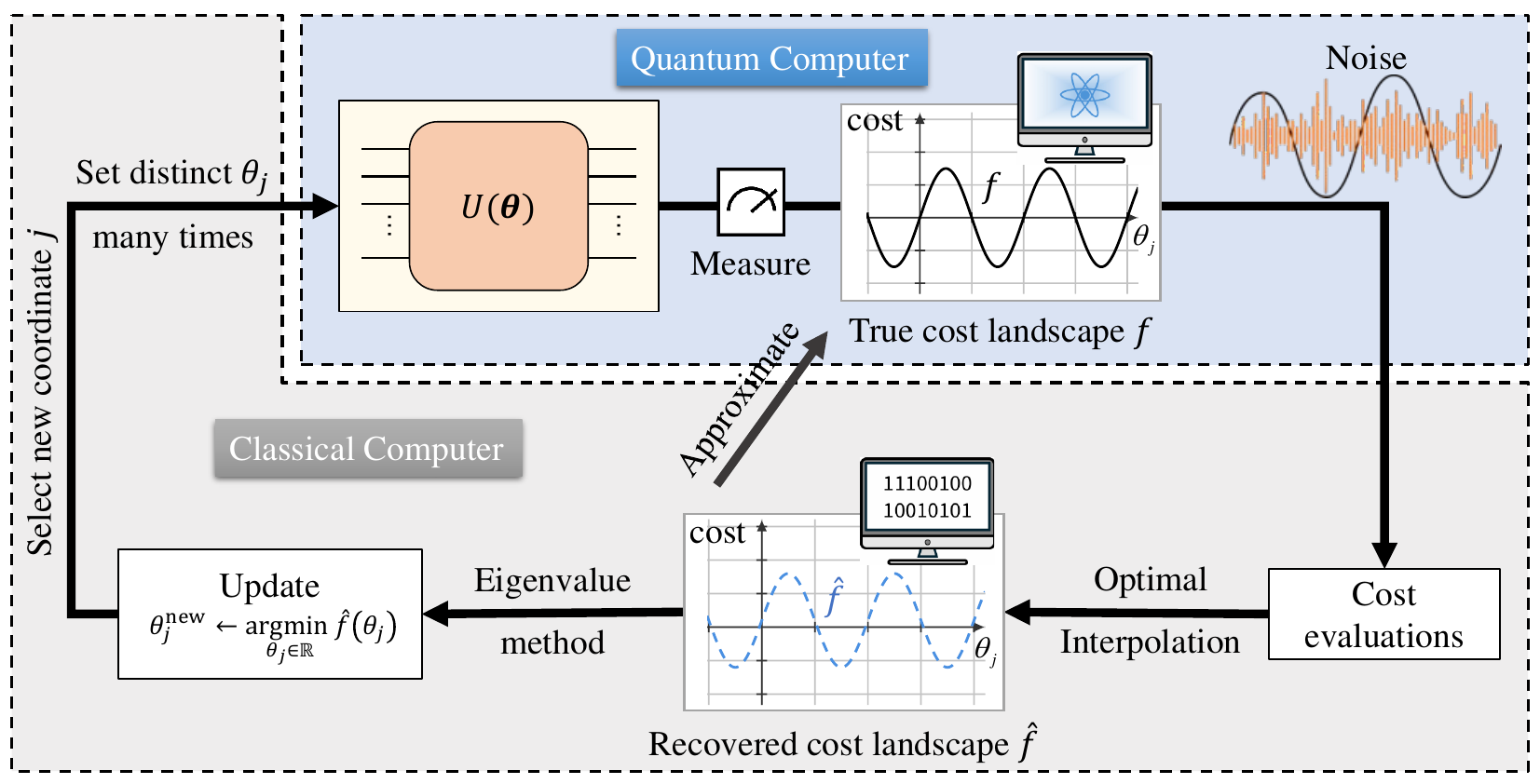} 
    \caption{}
  \end{subfigure}
  
  \vspace{0.5cm}
  
  \begin{subfigure}{\textwidth}
    \includegraphics[width=0.85\linewidth]{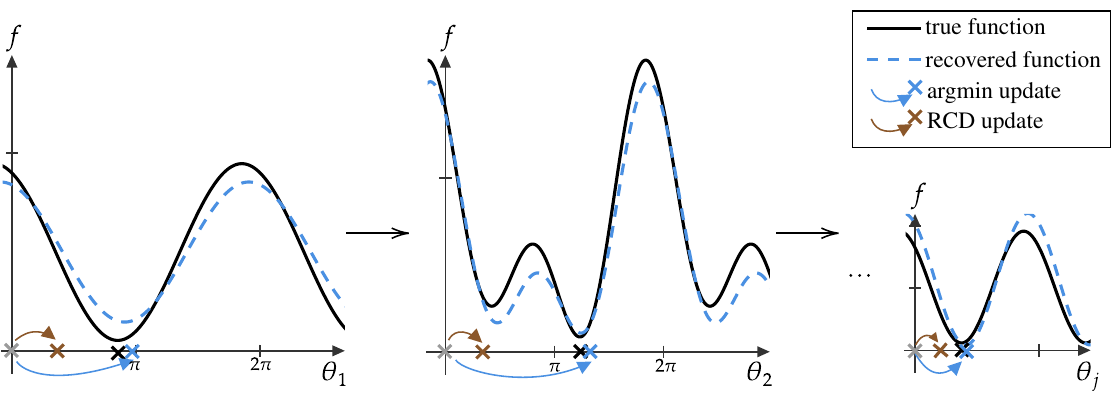}
    \caption{}
  \end{subfigure}
  
  \caption{
    (a) A diagram illustrating the ICD algorithm workflow. (b) Suppose we first update $\theta_1$, then update $\theta_2$, and so on. 
    We move the current value of $\theta_j$ (gray cross) to the origin. The black solid line represents the true curve of $f$ with respect to $\theta_j$, and we aim to find the true minimum (black cross). 
    By using an interpolation method under noisy conditions, we obtain relatively accurate estimates of $a_0$, $a_k$, and $b_k$ in \cref{eq-653}. 
    Using these estimated coefficients, we recover an approximate function (blue dashed line). This approximate function can be used in place of the original cost function, and its value at any point can be evaluated using a classical computer.
    In each update step, ICD finds the global minimum of the approximate function (blue cross), i.e., takes the argmin, which results in a significantly larger descent compared to the RCD method using one-step update (brown cross).
    }
    \label{fig:process_oicd}
\end{figure}

The interpolation method in our ICD has advantages similar to those of parameter shift rule (PSR), as it only requires finite function evaluations at some positions (called \textit{interpolation nodes}) to reconstruct the original true function, without the need for an additional ansatz. However, since the cost function is an expectation value, its exact values are generally unavailable, i.e., the function evaluations are inherently noisy due to at least the statistical errors. 
To this end, in our ICD method, we precompute an optimal set of interpolation nodes to minimize the impact of noise for each $\theta_j$. 
These optimal interpolation nodes are calculated only once and can be reused across subsequent iterations. The effects of different interpolation nodes are shown in \cref{fig:diff_nodes}. 
A detailed explanation will be provided in \cref{sec-ICD}.

\begin{figure}[htbp]
    \includegraphics[width=0.9\linewidth]{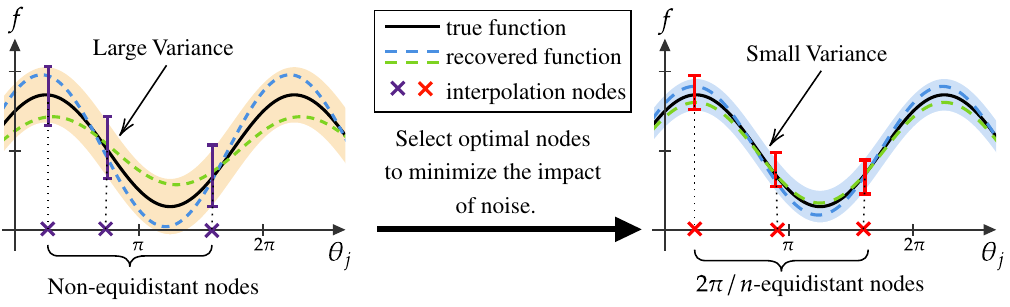}
    \caption{Variance between the approximate functions (recovered from different interpolation nodes) and the true function varies. Suppose we consider $\theta_j$ and theoretical true curve (black solid line) is $\theta_j \mapsto f(\boldsymbol{\theta}) = a_0 + a_1 \cos (\theta_j) + b_1 \sin (\theta_j).$
    To recover $n = 3$ values ($a_0$, $a_1$, and $b_1$), we simply select $n$ different nodes, evaluate their corresponding $f$ values, and solve a linear equation. The details of the interpolations will be given in \cref{sec-ICD}.
    Since the $f$ values always contain noise, the recovered function only approximate the true function within a certain range. It can be proven that for any positive integer $r_j$ and $\Omega_k^j = k$, the equidistant nodes with spacing $2 \pi /(2 r_j+1)$ are the optimal interpolation nodes, as they yield the closest approximation to the true function. }
    \label{fig:diff_nodes}
\end{figure}

\cref{tab:my_table} summarizes the comparison of the number of function evaluations $N_{\text {eval}}$ required per iteration for our ICD, as well as for SGD and RCD. 
ICD can be divided into two variants: standard and reduced. The only difference is that reduced ICD\footnote{For better readability, we include the reduced ICD in \cref{app-pratical-ICD}.} reuses the result from the previous iteration, thereby saving \textit{one} function evaluation and matching the computational cost of RCD. As shown in \cref{tab:comparison}, SMO and ExcitationSolve correspond to reduced ICD, while all other methods are the standard ICDs. \cref{tab:my_table} indicates that ICD and RCD require nearly the same quantum resources; however, ICD is an interpolation-based method, whereas RCD is a gradient-based method.

\begin{table}[htbp]
\centering
\renewcommand{\arraystretch}{1.3}
\setlength{\tabcolsep}{6pt}
\begin{tabularx}{\textwidth}{>{\raggedright\arraybackslash}X *{4}{>{\centering\arraybackslash}X}}
\toprule
Methods & SGD & RCD & standard ICD (\cref{sec-standard-ICD}) & reduced ICD (\cref{app-pratical-ICD}) \\
\midrule
Number of circuit evaluations $N_{\text{eval}}$ 
& $2 \|\boldsymbol{r}\|_1$ 
& $2 r_j$ 
& $2 r_j + 1$ 
& $2 r_j$ \\
\bottomrule
\end{tabularx}
\caption{Number of distinct circuit evaluations $N_{\text{eval}}$ per single update. Here, $\|\boldsymbol{r}\|_1 = \sum_{j=1}^m r_j$, and the integers $r_j$ are the numbers of terms in the trigonometric expansion of the cost function in \cref{eq-653}, to be formally introduced in \cref{sec_obs2}. Note that we use the PSR instead of finite differences to compute derivatives for SGD and RCD; see \cref{app-psr-fd} for more details.}
\label{tab:my_table}
\end{table}

\subsection{Contribution} 
The main contributions of this work can be summarized as follows. 
\begin{enumerate}
    \item We propose an interpolation-based coordinate descent (ICD) method to address the parameter optimization problem in PQCs. ICD integrates all the existing structural optimization methods into a unified framework. By incorporating interpolation techniques, the ICD method significantly reduces reliance on quantum devices, thereby enhancing computational efficiency. This is because, in ICD, the evaluated function values are not directly used for computing gradients/partial derivatives but rather to reconstruct the global landscape of the cost function on a classical computer as accurately as possible. In contrast, gradient-based methods like SGD and RCD use the evaluated function values directly to compute the derivatives at the current parameter, and the derivatives can only be used for a single update (requiring re-evaluation for subsequent updates). In ICD, however, the reconstructed function can be used for multiple updates using any optimization solvers, without additional quantum device operations. In the numerical simulation, we test the MaxCut problem, the transverse field Ising model (TFIM), and the XXZ problem, demonstrating that ICD is more effective than RCD and SGD. 
 
    \item For the case of $r$ equidistant frequencies (which is most common in PQCs, particularly when $H_j$ in \cref{eq-eiHtheta} are Pauli words, or sum of commuting Pauli words), we have shown that using $\frac{2\pi}{2r+1}$-equidistant interpolation nodes is an optimal scheme (under constant variance assumption). This specific scheme simultaneously satisfies the following three criteria: (1) minimization of the mean squared error between estimated Fourier coefficients and true coefficients, (2) minimization of the condition number of the interpolation matrix, and  (3) minimization of the average variance of the estimated derivatives. Moreover, we find that the subproblem of optimizing a single parameter in each iteration can be exactly solved by eigenvalue method proposed in \cite{boyd2006computing}. As shown in \cref{tab:comparison}, all existing structure optimization methods except \cite{jager2024fast} do not employ the optimal interpolation spacing. In our numerical experiments, we identify for the first time a relationship between node spacing and noise robustness: placing interpolation nodes at or near the optimal spacing enables ICD to retain convergence even under low shot counts.
\end{enumerate}

ICD also has its shortcomings. For the optimization task of PQCs, the most challenging obstacle is the barren plateau \cite{mcclean2018bp}. Unfortunately, ICD is unable to overcome this issue, despite not relying on gradient information. We will discuss these limitations in \cref{sec-r2-xxz-bp} later.

\subsection{Organization} 
This paper is organized as follows. 
In \cref{sec-resate}, we reformulate the cost function within the framework of optimization theory, offering a clear mathematical interpretation. 
In \cref{sec-ICD}, we propose our ICD method for the general case where Hermitian $H_j$ in the PQC can be arbitrary. 
In \cref{sec-equ-opt}, we discuss the equidistant frequency case, which is the most common in practical applications, and demonstrate further theoretical advantages of our proposed ICD method. 
In \cref{sec-experiments}, we discuss our numerical experiments.
We conclude the paper in \cref{sec-discussion} with a summary of our work and potential future directions. 

\subsection{Notations}
The superscript ${}^{\dagger}$ denotes the transpose for real matrices/vectors and the complex conjugate transpose for complex matrices/vectors.
We use $\operatorname{VAR}[\cdot]$ to denote the covariance matrix of a random vector, $\operatorname{Var}[\cdot]$ to represent the variance of a random variable, and $\operatorname{Cov}[\cdot , \cdot]$ to indicate the covariance between two random variables.

\section{Restating the problem} \label{sec-resate}

In this section, we restate \cref{eq:cost-function} from the perspective of optimization theory, providing its physical background in a precise mathematical context.

\subsection{Observation I: statistical nature from quantum measurement postulate}\label{sec_obs1}

According to the quantum measurement postulate \cite{nielsen2010quantum}, the cost function in \cref{eq:cost-function} is the expected value of a discrete random variable $\Lambda$. For an observable corresponding to a Hermitian operator $M \in \mathbb{C}^{N \times N}$ with spectral decomposition $M = \sum_{m=1}^{N} \lambda_m P_m$, the probability of obtaining eigenvalue $\lambda_m$ when measuring the state $|\psi(\boldsymbol{\theta})\rangle = U(\boldsymbol{\theta})|\psi_0\rangle$ is given by $p_{\boldsymbol{\theta}}(m) = \langle \psi(\boldsymbol{\theta}) | P_m | \psi(\boldsymbol{\theta}) \rangle \geq 0,$ and the expectation of $\Lambda$ is
\begin{equation}
    \operatorname{E}_{\Lambda \sim p_{\boldsymbol{\theta}}}[\Lambda] = \langle \psi(\boldsymbol{\theta}) | M | \psi(\boldsymbol{\theta}) \rangle = f(\boldsymbol{\theta}).
\end{equation}
The variance of $\Lambda$ is
\begin{equation}\label{eq-4280}
    \operatorname{Var}_{\Lambda \sim p_{\boldsymbol{\theta}}}[\Lambda] = \langle \psi(\boldsymbol{\theta}) | M^2 | \psi(\boldsymbol{\theta}) \rangle - [f(\boldsymbol{\theta})]^2 =: \sigma^2(\boldsymbol{\theta}).
\end{equation}
To estimate $f(\boldsymbol{\theta})$, we perform $\mathfrak{n}$ identical measurements (shots) and compute the sample mean $\bar{\Lambda} = \frac{1}{\mathfrak{n}} \sum_{i=1}^{\mathfrak{n}} \Lambda_i,$ where the $\Lambda_i$'s are i.i.d. samples from $p_{\boldsymbol{\theta}}$. We then have
\begin{equation}\label{eq-var-single}
    \operatorname{E}_{\Lambda \sim p_{\boldsymbol{\theta}}}[\bar{\Lambda}] = f(\boldsymbol{\theta}), \quad
    \operatorname{Var}_{\Lambda \sim p_{\boldsymbol{\theta}}}[\bar{\Lambda}] = \frac{\sigma^2(\boldsymbol{\theta})}{\mathfrak{n}}.
\end{equation}
By the central limit theorem, for large $\mathfrak{n}$, the sample mean $\bar{\Lambda}$ approximates a Gaussian distribution $\bar{\Lambda} \sim \mathcal{N} \left(f(\boldsymbol{\theta}), \frac{\sigma^2(\boldsymbol{\theta})}{\mathfrak{n}}\right).$
Thus, each evaluation of $f(\boldsymbol{\theta})$ is subject to zero-mean Gaussian noise:
\begin{equation}
    \tilde{f}(\boldsymbol{\theta}) = f(\boldsymbol{\theta}) + \mathcal{N}\left(0, \frac{\sigma^2(\boldsymbol{\theta})}{\mathfrak{n}}\right).
\end{equation}
This noise arises from the statistical nature of quantum measurements. In real quantum systems, there are also various hardware-induced noise sources (e.g., decoherence, gate errors). For unbiased noises, according to the central limit theorem again, their effect in the sample mean can also be approximated as Gaussian, and thus can be incorporated into the same mathematical framework. Note that we assume unbiased noise; handling biased noise would necessitate extending the framework to account for systematic errors.

\subsection{Observation II: trigonometric representation from quantum circuit structures}\label{sec_obs2}

Another feature of the cost function $f(\boldsymbol{\theta})$ is that it can be expressed as a trigonometric polynomial, which is the key for designing our ICD method. 
Consider a parameter vector $\boldsymbol{\theta} \in \mathbb{R}^m$, where all entries are fixed except for $\theta_j \in \mathbb{R}$ ($j = 1, \dots, m$). 
When we optimize a single variable $\theta_j$, the operations unrelated to $\theta_j$ can be absorbed into the input state and the observable. This leads to the following univariate cost function, 
\begin{equation}\label{eq-f-thetaj}
    \theta_j \mapsto f\left(\theta_j\right)
    =\langle\psi|U_j(\theta_j)^{\dagger} M^{\prime} U_j(\theta_j)| \psi\rangle,
\end{equation}
where $|\psi\rangle:= V_{j-1} U_{j-1}\left(\theta_{j-1}\right) \cdots V_1 U_1\left(\theta_1\right) |\psi_0\rangle$ is the state prepared by the subcircuit preceding $U_j\left(\theta_j\right)$, and $M^{\prime}:=V_j^{\dagger} \cdots U_m\left(\theta_m\right)^{\dagger} V_m^{\dagger} M V_m U_m\left(\theta_m\right) \cdots V_j$ includes the subcircuit following $U_j\left(\theta_j\right)$. Throughout the paper, we refer to $\theta_j \mapsto f\left(\theta_j\right)$ as the \textit{restricted univariate function} of $f(\boldsymbol{\theta})$.
For notation convenience, we simply write it as $f\left(\theta_j\right)$, and we can distinguish it from the original multivariate function by the argument, whether $\theta_j$ or $\boldsymbol{\theta}$. 

For a fixed index $j = 1, \dots, m$, let the eigenvalues of $H_j$ in \cref{eq-eiHtheta} be denoted by $\{\lambda_l^{j}\}_{l=1}^{N}$, and define the set of all unique positive differences between these eigenvalues, referred to as the \textit{frequencies}, to be\footnote{The definition of $\Omega_k$ given here is the loosest upper bound; in practice, $\Omega_k$ is often highly sparse, as discussed in \cref{app-trig,app-sparse-frequency}.}
\begin{equation}\label{eq-freq-set}
    \{\Omega_{k}^{j}\}_{k  =1}^{r_j} :=\{ | \lambda_l^{j} -\lambda_{l^{\prime}}^{j}| > 0 \mid \forall l, l^{\prime} =1,\ldots,N\},
\end{equation}
where $r_j := |\{\Omega^{j}_{k}\}|$. Here, the frequencies $\{\Omega_{k}^{j}\}_{k =1}^{r_j}$ are re-indexed in ascending order. According to \cite{wierichs2022general}, the restricted univariate function in \cref{eq-f-thetaj} can be expressed as a trigonometric polynomial (a finite-term Fourier series) as 
\begin{equation}\label{eq-trig-thetaj}
    f\left(\theta_j\right)
    =\frac{1}{\sqrt{2}} a_0+\sum_{k =1}^{r_j} \left[ a_{k} \cos \left(\Omega_{k}^{j} \theta_j \right)+b_{k} \sin \left(\Omega_{k}^{j} \theta_j \right) \right],
\end{equation}
where $a_0, a_k$ and $b_k$ are some real coefficients. 
This representation of $f(\theta_j)$ as a trigonometric polynomial succinctly captures the dependence of the cost function on the single parameter $\theta_j$, with each term oscillating at distinct frequencies determined by the eigenvalue differences of the generator $H_j$.

For completeness, we provide a detailed proof of \cref{eq-trig-thetaj} in \cref{app-trig}.
It should be noted that whenever we consider $f(\theta_j)$, we implicitly fix the values of the other $m-1$ parameters $\{\theta_i\}_{i\neq j}$. When those fixed parameters vary, the univariate function $f (\theta_j)$ itself changes, and this change is entirely absorbed into the Fourier coefficients $a_0$, $a_k$, and $b_k$ in \cref{eq-trig-thetaj}, while underlying frequency set $\{\Omega_k^j\}_{k=1}^{r_j}$ remains unchanged. The trigonometric nature of the cost function arises from the specific circuit structure in \cref{eq-Utheta}, especially since the parameterized gates are defined using $e^{i H_j \theta_j}$.

\subsection{Reformulated problem}

We are now ready to reformulate \cref{eq:cost-function} from the perspective of optimization theory. 

\begin{problem}[Reformulated PQC optimization problem]\label{Problem}
The goal is to find an efficient algorithm to solve the optimization problem,
\begin{equation}
    \boldsymbol{\theta}^* = \operatorname{argmin}_{\boldsymbol{\theta} \in \mathbb{R}^m} f(\boldsymbol{\theta}),
\end{equation}
under the following two observations.
\begin{description}
    \item[Observation I] Each function evaluation of $f(\boldsymbol{\theta})$ is subject to zero-mean Gaussian noise. Specifically, this gives rise to a random variable $\tilde{f}(\boldsymbol{\theta}) $ defined as
\begin{equation}\label{eq-problem-1}
    \tilde{f}(\boldsymbol{\theta}) = f(\boldsymbol{\theta}) + \mathcal{N} \left(0, \frac{\sigma^2(\boldsymbol{\theta})}{\mathfrak{n}}\right).
\end{equation}
    \item[Observation II] For each coordinate $j=1, \ldots, m$, the restricted univariate function \cref{eq-f-thetaj} of $f(\boldsymbol{\theta})$, has the trigonometric polynomial form 
\begin{equation}\label{eq-coordinate-j}
    f\left(\theta_j\right) = \frac{1}{\sqrt{2}} a_0 + \sum_{k=1}^{r_j}\left[a_k \cos(\Omega_k^{j} \theta_j) + b_k \sin(\Omega_k^{j} \theta_j)\right],
\end{equation}
where $a_0, a_k$ and $b_k$ are real coefficients determined by the remaining $\theta_i$'s with $i \neq j$.
\end{description}
\end{problem}

We will make further assumptions on the variance. 
While the variance of an individual measurement, $\sigma^2(\boldsymbol{\theta})$, technically depends on $\boldsymbol{\theta}$ as in \cref{eq-var-single}, evaluating the variance is often computationally prohibitive. Following the convention in existing studies \cite{mari2021estimating,wierichs2022general,markovich2024parameter}, we consider the following assumption, which is usually a good approximation in practice. 

\begin{assumption}[Constant variance]\label{assp-var}
We assume a constant noise level as follows: given any $\boldsymbol{\theta}$, $\sigma^2(\boldsymbol{\theta}) \approx \sigma^2(\boldsymbol{\theta} + s \boldsymbol{e}_j)$ for all $s \in \mathbb{R}.$ Here, $\boldsymbol{e}_j$ represents the standard basis vector in the $j$-th direction. 
\end{assumption}

\section{Interpolation-based coordinate descent method}\label{sec-ICD}

In this section, we will propose our interpolation-based coordinate descent (ICD) method for solving \cref{Problem}. 

\subsection{Overview}\label{sec-overview-ICD}

We first provide an overview of the original coordinate descent (CD) method \cite{shi2016primer}. The original CD method to \cref{Problem} works as follows: given current parameters $\boldsymbol{\theta}$, we first select a coordinate $j$ and consider the restricted univariate function $f\left(\theta_j\right)$ as in \cref{eq-coordinate-j}. Then, we update the $j$-th component of $\boldsymbol{\theta}$ by,
\begin{equation}\label{eq-exact-prob}
\theta_j^{\text{new}} \leftarrow \underset{\theta_j \in \mathbb{R}}{\operatorname{argmin}\;}   f\left(\theta_j\right),
\end{equation}
which is simply a single-variable optimization subproblem and is easy to solve. 
Usually, one uses the gradient descent method to solve \cref{eq-exact-prob}.
When the coordinate $j$ is selected randomly, 
and a single step is taken in the direction of the negative gradient (which, in our case, becomes a negative derivative) with some learning rate $\alpha>0$ as 
\begin{equation}\label{eq-rcd}
\theta_j^{\text{new}}  \leftarrow \theta_j^{\text{old}} -\alpha \left.\frac{\mathrm{d} f\left(\theta_j\right)}{\mathrm{d}\theta_j} \right|_{\theta_j=\theta_j^{\text{old}}},
\end{equation}
the CD method becomes the famous random coordinate descent (RCD) \cite{ding2024random}. After updating coordinate $j$, we next select a new coordinate and repeat the above process. 

No matter what method we use to solve the subproblem in \cref{eq-exact-prob}, the CD method updates only one parameter at each iteration. However, it can be computationally expensive to directly apply existing solvers to \cref{eq-exact-prob}, because most solvers rely on iterative methods requiring numerous function evaluations $\theta_j \mapsto f\left(\theta_j\right)$. In PQC context, each function evaluation necessitates re-tuning the quantum device parameters and repeating measurements, making this approach prohibitively expensive in terms of both time and quantum resources. 

To address this issue, we leverage the trigonometric structure of the cost function in Observation II. 
Our approach is to recover the coefficients $a_0$, $a_k$, and $b_k$ in \cref{eq-coordinate-j} using the \textit{interpolation} method (which will be discussed in the next subsection). This method involves only a limited number (specifically $2r_j + 1$) of evaluations of $f(\theta_j)$ at different points. On the other hand, since function evaluations are always noisy as per Observation I, the recovered function can never be exact: it will only serve as an approximation to \cref{eq-coordinate-j}. Then, with the approximated function 
\begin{equation}\label{eq-inter-hat-f}
\theta_j \mapsto \hat{f}\left(\theta_j\right) = \frac{1}{\sqrt{2}} \hat{a}_0 + \sum_{k=1}^{r_j}[\hat{a}_k \cos(\Omega_k^j \theta_j) + \hat{b}_k \sin(\Omega_k^j \theta_j)]
\end{equation}
in hand (symbol $\,{}^{\hat{ }}\,$ indicates an estimated value), we will solve an approximated subproblem
\begin{equation}\label{eq-approximated-prob}
\theta_j^{\text{new}} \leftarrow \underset{\theta_j \in \mathbb{R}}{\operatorname{argmin}\;}   \hat{f}\left(\theta_j\right).
\end{equation}
Importantly, all the information of $\hat{f}$ (i.e., the estimated $\hat{a}_0, \hat{a}_k, \hat{b}_k$) is stored on the classical computer, and the callings of $\theta_j \mapsto \hat{f}\left(\theta_j\right)$ are completely independent of the quantum device. As a result, there is no additional quantum-related burden in solving \cref{eq-approximated-prob}. The other process remains the same as in the original CD. We call this the interpolation-based coordinate descent (ICD) method to \cref{Problem}. 

In the simplest case where $r_j = 1$ and $\Omega_k^j = 1$, we will solve
\begin{equation}
\theta_j^{\text{new}} \leftarrow \underset{\theta_j \in \mathbb{R}}{\operatorname{argmin}\;}  \hat{f}\left(\theta_j\right) = \frac{1}{\sqrt{2}} \hat{a}_0 + \hat{a}_1 \cos(\theta_j) + \hat{b}_1 \sin(\theta_j),
\end{equation}
which has a closed-form solution $\theta_j^{\text{new}} \leftarrow \theta_j^*=\operatorname{arctan2}(\hat{b}_1, \hat{a}_1)+\pi.$ This special case corresponds exactly to Rotosolve \cite{ostaszewski2021structure}. In the general case $r_j \geq 2$ and $\Omega_k^j = k$, there is no closed-form solution anymore, however, we can adopt an eigenvalue method (discussed in \cref{app-eig-method}) to exactly solve the approximated subproblem in \cref{eq-approximated-prob}. 

Clearly, the effectiveness of the ICD method lies in how to best recover the approximation function $\hat{f}\left(\theta_j\right)$ for true function $f\left(\theta_j\right)$ in the presence of unavoidable noise. The accuracy of the solution to the approximated subproblem \cref{eq-approximated-prob}, relative to the exact subproblem \cref{eq-exact-prob}, is directly determined by the error between $\hat{f}$ and $f$. In subsequent subsections, we will demonstrate how to enhance the interpolation method to minimize noise impact. 

\subsection{Interpolation method to recover restricted univariate functions}\label{sec-interp-real}

In this subsection, we discuss how to perform interpolation to recover the restricted univariate function in \cref{eq-coordinate-j}. 
For notation simplicity, we omit index $j$ and replace the variable $\theta_j$ with $x$, and consider the trigonometric polynomial $f \colon \mathbb{R} \to \mathbb{R}$ of order $r$
\begin{equation}\label{eq-trig-poly-real}
    f(x) = \frac{1}{\sqrt{2}} a_0 + \sum_{k=1}^r \left[ a_k \cos(\Omega_k x) + b_k \sin(\Omega_k x) \right],
\end{equation}
where $a_k$'s and $b_k$'s are $n \equiv 2r + 1$ unknown real parameters. 
Due to the specific construction, all information about $f$ is equivalent to these coefficients. The goal of the interpolation is to recover all coefficients above by evaluating the function $f(x)$ at various points $x$. We next discuss the noise-free case and the noisy case, respectively. 

\subsubsection{Interpolation with true data}\label{sec:no-noise}

Suppose we have access to the calling $x \mapsto f(x)$ without noise for any argument $x$. This case is the foundation for our subsequent consideration of interpolation under noise. Indeed, knowing any set of $n$ distinct true data points $\{ (x_i, f(x_i)) \}_{i=0}^{2r}$ allows us to \textit{exactly} recover the coefficient vector
\begin{equation}
    \mathbf{z} := \left[ a_0, a_1, b_1, \cdots, a_r, b_r \right]^{\dagger} \in \mathbb{R}^n,
\end{equation}
thereby giving us a complete understanding of $f$ in \cref{eq-trig-poly-real}.
Let $\mathbf{x} := [x_0, x_1, \cdots, x_{2r}]^{\dagger} \in \mathbb{R}^n$ with distinct entries.
Define the \textit{true data vector}
\begin{equation}\label{eq-yx}
\mathbf{y}_{\mathbf{x}} := [f(x_0), f(x_1), \cdots, f(x_{2r})]^{\dagger} \in \mathbb{R}^n.
\end{equation}
Now, using the chosen $\mathbf{x}$, we construct the interpolation matrix
\begin{equation}\label{eq-A}
A_{\mathbf{x}}:=\left[\begin{array}{cccccc}
\frac{1}{\sqrt{2}} & \cos \left(\Omega_1 x_0\right) & \sin \left(\Omega_1 x_0\right) & \cdots & \cos \left(\Omega_r x_0\right) & \sin \left(\Omega_r x_0\right) \\
\frac{1}{\sqrt{2}} & \cos \left(\Omega_1 x_1\right) & \sin \left(\Omega_1 x_1\right) & \cdots & \cos \left(\Omega_r x_1\right) & \sin \left(\Omega_r x_1\right) \\
\vdots & \vdots & \vdots & \ddots & \vdots & \vdots \\
\frac{1}{\sqrt{2}} & \cos \left(\Omega_1 x_{2 r}\right) & \sin \left(\Omega_1 x_{2 r}\right) & \cdots & \cos \left(\Omega_r x_{2 r}\right) & \sin \left(\Omega_r x_{2 r}\right)
\end{array}\right] \in \mathbb{R}^{n \times n} .
\end{equation}
Given these definitions, we can solve the linear equation 
\begin{equation}\label{eq-yAz}
     A_{\mathbf{x}} \mathbf{z}=\mathbf{y}_{\mathbf{x}} 
\end{equation}
to recover the coefficient $\mathbf{z}$. Here, $\mathbf{x}$ serves as a set of \textit{interpolation nodes}. Moreover, due to the special structure of $\mathbf{y}_{\mathbf{x}}$ and $A_{\mathbf{x}}$  in \cref{eq-yx} and \cref{eq-A}, the solution $\mathbf{z} \equiv A^{-1}_{\mathbf{x}} \mathbf{y}_{\mathbf{x}}$ is independent of $\mathbf{x}$. This implies that any interpolation nodes $\mathbf{x}$ yield exactly the same results. It should be noted that this property holds only in the ideal setting where we can evaluate $x \mapsto f(x)$ without noise.

Actually, here, we implicitly assume that the chosen $\mathbf{x}$ results in a non-singular $A_{\mathbf{x}}$. Since $A_{\mathbf{x}}$ contains parameters $\Omega_k$, it is difficult to derive the condition for $\mathbf{x}$ that guarantees non-singularity. In the algorithms later, we will use an optimally chosen $\mathbf{x}^*$ by solving \cref{pro-min-mse}, which naturally ensures that $A_{\mathbf{x}^*}$ is non-singular. In \cref{sec-equ-opt}, for the case of $\Omega_k = k$, we can prove that the condition for $A_{\mathbf{x}}$ to be non-singular is precisely that all entries of $\mathbf{x}$ are distinct modulo $2\pi$ (\cref{app-recover-complex}). 

\subsubsection{Interpolation with noisy data}\label{sec:noise}

In practice, we can only access the function $x \mapsto f(x)$ with noise. According to Observation I and \cref{eq-problem-1}, the observed value of $f(x)$ has an additive zero-mean Gaussian noise $\epsilon_x \sim \mathcal{N}(0, \sigma^2(x)/\mathfrak{n})$. By \cref{assp-var}, the variance $\sigma^2(x)$ is independent of $x$, allowing us to simply denote it as $\sigma^2$. Since we always consider a constant number $\mathfrak{n}$ throughout this paper, $\mathfrak{n}$ can be absorbed into $\sigma^2$. 
Then we simply have $\epsilon_x \sim \mathcal{N}(0, \sigma^2)$, and each evaluation of the cost function gives noisy data as a random variable 
\begin{equation}\label{assm-tilde-f}
\tilde{f}(x) = f(x) + \epsilon_x, \quad \epsilon_x \sim \mathcal{N}(0, \sigma^2).
\end{equation}

Now, given $n$ many noisy data points $\{ (x_i, \tilde{f}(x_i)) \}_{i=0}^{2r}$, rather than solving \cref{eq-yAz}, we will address its perturbed version:
\begin{equation}\label{eq-perturbed-equation}
    A_{\mathbf{x}} \hat{\mathbf{z}}_{\mathbf{x}} = \mathbf{y}_{\text{obs}} := \mathbf{y}_{\mathbf{x}} + \mathbf{e},
\end{equation}
where $\mathbf{e} = [\epsilon_0, \epsilon_1, \cdots, \epsilon_{2r}]^{\dagger} \sim \mathcal{N}(0, \sigma^2 I)$ is the normal random vector, $\mathbf{y}_{\mathbf{x}}$ is still the true (but unknown) data vector as \cref{eq-yx}, and 
\begin{equation}
\hat{\mathbf{z}}_{\mathbf{x}} = [\hat{a}_0, \hat{a}_1, \hat{b}_1, \cdots, \hat{a}_r, \hat{b}_r]^{\dagger} \in \mathbb{R}^n
\end{equation}
is solution of the perturbed equation. In fact, $\hat{\mathbf{z}}_{\mathbf{x}}$ is an estimator of the exact coefficients $\mathbf{z}$. We have 
\begin{equation}
    \hat{\mathbf{z}}_{\mathbf{x}} 
= A_{\mathbf{x}}^{-1} (\mathbf{y}_{\mathbf{x}} + \mathbf{e}) = \mathbf{z} + A^{-1}_{\mathbf{x}} \mathbf{e}. 
\end{equation}
Notice that $\mathbf{e}$ is normal, and $ \hat{\mathbf{z}}_{\mathbf{x}} 
= \mathbf{z} + A^{-1}_{\mathbf{x}} \mathbf{e}$ represents an affine transformation of $\mathbf{e}$, so $\hat{\mathbf{z}}_{\mathbf{x}} $ is also a normal random vector. Moreover, we have 
\begin{equation}
    \operatorname{E}[\hat{\mathbf{z}}_{\mathbf{x}}] = \mathbf{z} + A^{-1}_{\mathbf{x}}  \operatorname{E}[\mathbf{e}] = \mathbf{z},
\end{equation}
and the covariance matrix of $\hat{\mathbf{z}}_{\mathbf{x}} $ is 
\begin{align}\label{eq-var}
    \operatorname{VAR}[\hat{\mathbf{z}}_{\mathbf{x}}]
    = \operatorname{E}[(\hat{\mathbf{z}}_{\mathbf{x}} - \mathbf{z})(\hat{\mathbf{z}}_{\mathbf{x}} - \mathbf{z})^{\dagger}]
    = A_{\mathbf{x}}^{-1} \operatorname{E}[\mathbf{e} \mathbf{e}^{\dagger}] (A_{\mathbf{x}}^{-1})^{\dagger}
   = \sigma^2 [A_{\mathbf{x}}^{\dagger} A_{\mathbf{x}}]^{-1}.
\end{align}
It is evident that $\hat{\mathbf{z}}_{\mathbf{x}}$ serves as an unbiased estimator of the true coefficients $\mathbf{z}$. 
However, its variance depends on the interpolation nodes $\mathbf{x}$. 
This naturally raises the question: can we determine optimal interpolation nodes $\mathbf{x}$ that provide the best possible approximation to the true coefficients $\mathbf{z}$? 

To quantify the estimation error of $\hat{\mathbf{z}}_{\mathbf{x}}$ with respect to the true $\mathbf{z}$, a common metric is the mean squared error (MSE). Let the estimation error be defined as
\begin{equation}
    \Delta \mathbf{z}_{\mathbf{x}} := \hat{\mathbf{z}}_{\mathbf{x}} - \mathbf{z} = A_{\mathbf{x}}^{-1} \mathbf{e}.
\end{equation}
The MSE of $\hat{\mathbf{z}}$ is then given by $\operatorname{MSE}(\hat{\mathbf{z}}_{\mathbf{x}}) := \operatorname{E}[\|\hat{\mathbf{z}}_{\mathbf{x}} - \operatorname{E}[\hat{\mathbf{z}}_{\mathbf{x}}]\|^2] = \operatorname{E}[\|\Delta \mathbf{z}_{\mathbf{x}}\|^2].$ Indeed, we have
\begin{align}
    \operatorname{MSE}(\hat{\mathbf{z}}_{\mathbf{x}})
    = & \operatorname{E}[\operatorname{tr}((A_{\mathbf{x}}^{-1} \mathbf{e})^{\dagger}(A_{\mathbf{x}}^{-1} \mathbf{e}))] 
    = \operatorname{E}[\operatorname{tr}((A_{\mathbf{x}}^{-1} \mathbf{e})(A_{\mathbf{x}}^{-1} \mathbf{e})^{\dagger})] \\
     = & \operatorname{tr}(\operatorname{E}[(A_{\mathbf{x}}^{-1} \mathbf{e})(A_{\mathbf{x}}^{-1} \mathbf{e})^{\dagger}]) 
    = \sigma^2 \operatorname{tr}([A_{\mathbf{x}}^{\dagger} A_{\mathbf{x}}]^{-1})
    = \sigma^2 \|A_{\mathbf{x}}^{-1}\|_F^2, \label{eq-mse-norm}
\end{align}
where $\|\cdot\|_F$ is the Frobenius norm, i.e., $\|X\|_F=\sqrt{\operatorname{tr}\left(X^{\dagger} X\right)}
$.
The second-to-last equality follows from \cref{eq-var}. In fact, for any unbiased estimator $\hat{\Theta} \in \mathbb{R}^n$ of the true vector $\Theta$, one has $\operatorname{MSE}(\hat{\Theta})=\operatorname{tr}(\operatorname{VAR}(\hat{\Theta}))$.

Clearly, the mean squared error $\operatorname{MSE}(\hat{\mathbf{z}}_{\mathbf{x}})$ depends on the specific choice of interpolation nodes $\mathbf{x}$. Naturally, we seek to select the optimal nodes
\begin{equation}\label{pro-min-mse}
    \mathbf{x}^* = \underset{\mathbf{x} \in \mathbb{R}^n}{\operatorname{argmin}\;}   \operatorname{MSE}(\hat{\mathbf{z}}_{\mathbf{x}}) = \sigma^2 \|A_{\mathbf{x}}^{-1}\|_F^2,
\end{equation}
to achieve the best approximation, which is equivalent to minimizing $\|A_{\mathbf{x}}^{-1}\|_F^2$ since $\sigma^2$ is a constant. We observe that minimizing $\|A_{\mathbf{x}}^{-1}\|_F^2$ inherently forces $A_{\mathbf{x}}$ to be invertible; otherwise, the objective value would tend to infinity. Since $A_{\mathbf{x}}$ in \cref{eq-A} involves complicated parameters $\Omega_k$, it is difficult to obtain an analytical solution $\mathbf{x}^*$ to \cref{pro-min-mse}. However, a numerical solution is sufficient for our algorithmic design.
In practice, we can use common algorithms like Adam to solve it.
In \cref{sec-equ-opt}, we will show that for the equidistant frequency case ($\forall  \Omega_k=k$), a global optimal analytical solution to \cref{pro-min-mse} exists.

\begin{remark}[Constant variance is not realistic]
All theoretical analyses in this paper critically rely on \cref{assp-var} (constant variance), which significantly simplifies the theoretical complexity and has also nice empirical performance in our numerical experiments. However, \cref{assp-var} is not realistic. As a complement, in \cref{app-mse-nonconst}, we analyze the true MSE without relying on the constant variance assumption. From both an error bound and experimental perspective, we justify the practical reasonableness of \cref{assp-var}.
\end{remark}

\subsection{MSE and variance of approximated function}\label{sec-mse-var}

We are curious whether there exist alternative criteria for choosing the interpolation nodes, beyond \cref{pro-min-mse}. 
To explore this, we conduct the following analysis. We return to the notation introduced in \cref{sec-interp-real}. Note that the approximated function
\begin{equation}\label{eq-app-fx}
   \hat{f}(x) = \frac{1}{\sqrt{2}} \hat{a}_0 + \sum_{k=1}^r [ \hat{a}_k \cos(\Omega_k x) + \hat{b}_k \sin(\Omega_k x) ]
\end{equation}
can be fully processed on a classical computer once we have obtained the estimated coefficients $\hat{\mathbf{z}}_{\mathbf{x}}$. Let us define the vector
\begin{equation}\label{eq-t0x}
\mathbf{t}(x):=[1/ \sqrt{2}, \cos (\Omega_1 x), \sin (\Omega_1 x), \ldots, \cos (\Omega_r x), \sin (\Omega_r x)]^{\dagger} \in \mathbb{R}^n.
\end{equation}
Note that $\mathbf{t}(x)^{\dagger} \mathbf{t}(x) = \frac{2r+1}{2}=\frac{n}{2}$ for any $x$. Then, $\hat{f}(x)$ can be expressed as $\hat{f}(x)= \mathbf{t}(x)^{\dagger} \hat{\mathbf{z}}_{\mathbf{x}}.$ Due to the normal distribution properties of the random vector $\hat{\mathbf{z}}_{\mathbf{x}}$, $\hat{f}(x)$ itself also follows a normal distribution, with expectation
\begin{equation}\label{eq-unbaised-hatf}
    \operatorname{E}[\hat{f}(x)]= \mathbf{t}(x)^{\dagger} \operatorname{E}[\hat{\mathbf{z}}_{\mathbf{x}}]= \mathbf{t}(x)^{\dagger} \mathbf{z}=f(x).
\end{equation}
Hence, $\hat{f}(x)$ provides unbiased estimates of the true evaluation $f(x)$ at any $x \in \mathbb{R}$, regardless of the interpolation nodes $\mathbf{x}$. 
Furthermore, as shown in \cref{lem-cov-var} of \cref{sec-proof-thm-3}, the variance can be computed as
\begin{equation}\label{eq-var-fx}
\operatorname{Var}[\hat{f}(x)]
=\mathbf{t}(x)^{\dagger} \operatorname{VAR}[\hat{\mathbf{z}}_{\mathbf{x}}] \mathbf{t}(x) = \langle\operatorname{VAR}[\hat{\mathbf{z}}_{\mathbf{x}} ], \mathbf{t}(x) \mathbf{t}(x)^{\dagger}\rangle .
\end{equation}
This allows us to directly assess how well $\hat{f}$ approximates $f$. However, this variance depends on both the interpolation nodes $\mathbf{x}$ and the univariate variable $x$. 

To address this, we seek a global variance measure that removes the dependence on the specific univariate input $x$ and focuses solely on the effect of $\mathbf{x}$. This, however, requires additional conditions to become tractable. For example, assuming $f$ and $\hat{f}$ have a period $T>0$ (usually $2\pi$), we can use the following quantity
\begin{equation}\label{eq-hx-0}
h(\mathbf{x}) := \frac{1}{T} \int_0^{T} \operatorname{Var}[\hat{f}(x)]   \;\mathrm{d} x =\langle\operatorname{VAR}[\hat{\mathbf{z}}_{\mathbf{x}} ], \underbrace{\frac{1}{T} \int_0^{T} 
 \mathbf{t}(x) \mathbf{t}(x)^{\dagger}   \mathrm{d} x}_{\text{const.}} \rangle
\end{equation}
to evaluate the quality of the interpolation nodes. Minimizing $h(\mathbf{x})$ could give us a better overall estimate of $f$. However, in general, if the frequencies $\Omega_k$'s are not assumed to be rational or integer, then $f$ or $\hat{f}$ might not have a period (e.g., $\sin\left(x\right)+\sin\left(\pi x\right)$). If the integration interval in \cref{eq-hx-0} is the entire $\mathbb{R}$, then $h(\mathbf{x})$ would be difficult to handle. Fortunately, in the next section, we will consider equidistant frequencies ($\forall \Omega_k=k$), which will make $h(\mathbf{x})$ easier to handle.

Here, for general frequencies, we can still provide a upper bound for the variance for all $x$, namely, the MSE itself. We observe first that the covariance matrix $\operatorname{VAR}[\hat{\mathbf{z}}_{\mathbf{x}} ]$ is always symmetric and positive semi-definite, then its Rayleigh quotient \cite{horn2012matrix} satisfies
\begin{equation}
    \mathbf{t}(x)^{\dagger} \operatorname{VAR}[\hat{\mathbf{z}}_{\mathbf{x}} ] \;\mathbf{t}(x) \leq \lambda_{\max} (\operatorname{VAR}[\hat{\mathbf{z}}_{\mathbf{x}} ])\, \|\mathbf{t}(x)\|^2
  =\frac{n}{2}\lambda_{\max} (\operatorname{VAR}[\hat{\mathbf{z}}_{\mathbf{x}} ]).
\end{equation}
Finally, because $\lambda_{\max}(\operatorname{VAR}[\hat{\mathbf{z}}_{\mathbf{x}} ])\le\sum_i\lambda_i(\operatorname{VAR}[\hat{\mathbf{z}}_{\mathbf{x}} ])=\operatorname{tr}(\operatorname{VAR}[\hat{\mathbf{z}}_{\mathbf{x}} ])$, we obtain the clean bound:
\begin{equation}\label{eq-9555}
   \forall x\in \mathbb{R}, \quad \operatorname{Var}[\hat{f}(x)] \le\frac{n}{2}\operatorname{tr}(\operatorname{VAR}[\hat{\mathbf{z}}_{\mathbf{x}} ])=\frac{n}{2}\operatorname{MSE}(\hat{\mathbf{z}}_{\mathbf{x}}). 
\end{equation}
Equality on the upper bound is attained precisely when $\operatorname{MSE}(\hat{\mathbf{z}}_{\mathbf{x}})$ is a multiple of the identity. Thus, the optimal $\mathbf{x}^*$ under the criterion of \cref{pro-min-mse} indeed ensures that the overall variance of the approximation $\hat{f}$ is controlled. Taking all the above into account, for the general frequency setting, we adopt the MSE rather than the average variance as the evaluation criterion.

\subsection{Standard ICD algorithm}\label{sec-standard-ICD}

Building on the previous subsections, we now formally introduce the interpolation-based coordinate descent (ICD) method for solving \cref{Problem}.
First, we introduce the standard ICD in \cref{alg-standard-ICD}.
As discussed in \cref{sec-overview-ICD}, a key initial step of ICD is to obtain the optimal interpolation nodes $\mathbf{x}^{j, *}$ and corresponding interpolation matrices $A_{\mathbf{x}^{j,*}}$ for each  $j = 1, \ldots, m$. To clarify this step, we present it separately in \cref{alg-opt-interp}. 
Note that \cref{alg-opt-interp} is a preparatory step, and the interpolation schemes generated here can be reused in each iteration of ICD \cref{alg-standard-ICD}. We explicitly compute the inverse of $A_{\mathbf{x}^{j,*}}$ for two reasons: 
in practice, the scale of $A_{\mathbf{x}^{j,*}}$ is small; during the iterations, we need to repeatedly solve the equation $A_{\mathbf{x}^{j,*}} \hat{\mathbf{z}} = \mathbf{y}_{\text{obs}}$. Note that $A_{\mathbf{x}^{j,*}}$ is fixed, while $\mathbf{y}_{\text{obs}}$ is constantly updated. Therefore, computing the inverse of $A_{\mathbf{x}^{j,*}}$ once and solving the equation using $\hat{\mathbf{z}} = A_{\mathbf{x}^{j,*}}^{-1} \mathbf{y}_{\text{obs}}$ is more efficient.

\begin{algorithm}[H]
\caption{Standard ICD Method for \cref{Problem}}
\label{alg-standard-ICD}
\SetAlgoLined
\SetKwInOut{Input}{Input}
\SetKwInOut{Output}{Output}
\Input{Initial parameters $\boldsymbol{\theta}^0 = [\theta_1^0,\ldots, \theta_m^0]^{\dagger}$, and the number of iterations $\mathsf{T}$.}
\Output{Optimized parameters $\boldsymbol{\theta}^{\mathsf{T}}$ after $\mathsf{T}$ iterations.}
\BlankLine
Obtain the optimal interpolation schemes $\{(\mathbf{x}^{j,*},A_{\mathbf{x}^{j,*}}^{-1})\}_{j=1}^m$ using \cref{alg-opt-interpolation}\; 
\For{$t = 0$ \KwTo $\mathsf{T}$}{
    Select a coordinate $j \in \{1, \ldots, m\}$, either sequentially or uniformly at random\;
    Fix all parameters of $\boldsymbol{\theta}^t$ except for $\theta_j^t$, and consider the restricted univariate function $\theta_j \mapsto f(\theta_j)$\; 
 (Quantum burden) Construct the observed data vector at $\mathbf{x}^{j,*}$, i.e., 
\begin{equation}\label{eq-obs-y}
\mathbf{y}_{\text{obs}} := [\tilde{f}(x^{j,*}_0), \tilde{f}(x^{j,*}_1), \ldots, \tilde{f}(x^{j,*}_{2r_j})]^{\dagger};
\end{equation}
    \hspace{-3mm} 
    Compute the estimated coefficients $\hat{\mathbf{z}} := A_{\mathbf{x}^{j,*}}^{-1}  \mathbf{y}_{\text{obs}}$ and recover the estimated function $\hat{f}(\theta_j)$ as in \cref{eq-inter-hat-f}\;
    Let $\theta_j^{t+1} := \underset{\theta_j \in \mathbb{R}}{\operatorname{argmin}\;}   \hat{f}(\theta_j)$\;
    Let $\theta_i^{t+1} := \theta_i^{t}$ for all $i \neq j$\;
}
\end{algorithm}

\begin{algorithm}[H]\label{alg-opt-interpolation}
\caption{Obtain the Optimal Interpolation Schemes for General Frequencies}
\label{alg-opt-interp}
\SetAlgoLined
\SetKwInOut{Input}{Input}
\SetKwInOut{Output}{Output}
\Input{All Hermitian generators $H_j \in \mathbb{C}^{N \times N}$, $j=1,\ldots,m$, in \cref{eq-eiHtheta}.}
\Output{Optimal interpolation nodes $\mathbf{x}^{j,*}$ and inverse of optimal interpolation matrices $A_{\mathbf{x}^{j,*}}^{-1}$, $j=1,\ldots,m$. 
}
\For{$j=1,\ldots,m$}{
    Determinate the frequency set $\{\Omega_k^j\}_{k=1}^{r_j}$ of $H_j$ as defined in \cref{eq-freq-set}\;
    Let $n_j := 2r_j + 1$. Consider the associated interpolation matrix $A_{\mathbf{x}} \in \mathbb{R}^{n_j \times n_j}$ defined in \cref{eq-A} based on the computed frequencies\;  
    Solve for the optimal interpolation nodes  $\mathbf{x}^{j,*}
    :=\underset{\substack{\mathbf{x} \in \mathbb{R}^{n_j}}}{\operatorname{argmin}\;} 
      \|A_{\mathbf{x}}^{-1}\|_F^2$\; 
    Compute the inverse of optimal interpolation matrix $A_{\mathbf{x}^{j,*}}^{-1}$\; 
}
\end{algorithm}

\section{Optimal interpolation nodes for equidistant frequencies}\label{sec-equ-opt}

So far, all the discussions have been for the general frequency $\Omega_k$, which appears in \cref{eq-coordinate-j} of Observation II. This section discusses the equidistant frequency case, which is the most common in practical applications. Moreover, in this case, our proposed ICD algorithm has many elegant theoretical results. 
More specifically, throughout this section, we adopt the following assumption \cite{wierichs2022general}. 

\begin{assumption}[Equidistant frequencies]\label{assm-equi}
For every $H_j$ in \cref{eq-Utheta}, we assume the frequencies $\{\Omega_{k}^{j}\}_{k  =1}^{r_j}$ are equidistant, i.e., $\Omega_{k}^{j}=k  \Omega^{j}$ $(k=1,\ldots,r_j)$ for some constant $\Omega^{j}$. Without loss of generality\footnote{For $\Omega^{j} \neq 1$, we can rescale the function argument to achieve $\Omega_{k}^{j}=k $. Once the rescaled function is constructed, the original function is readily available. }, we further restrict the frequencies to integer values, i.e., $\Omega_{k}^{j}=k $ $(k=1,\ldots,r_j)$.
\end{assumption}

The equidistant frequency patterns often arise in practical scenarios. 
A particularly important case occurs when the Hermitian operator $H_j$ can be expressed as a sum of $R$ commuting Pauli words $\mathcal{P}_k$ with coefficients of $\pm 1$, i.e., $H_j = \sum_{k=1}^{R} \pm \mathcal{P}_k$. This structure results in equidistant frequencies and $r_j = R$.
Notice that while equidistant eigenvalues do imply equidistant frequencies, the converse is not always true. It is possible for a generator to have non-equidistant eigenvalues but still produce equidistant frequencies, e.g., $H_j$ has three distinct $\lambda_1=0, \lambda_2=1, \lambda_3=3$.

For the equidistant frequencies case, the previously proposed \cref{alg-standard-ICD} can be directly applied without any modification. Moreover, the corresponding optimal interpolation nodes $\mathbf{x}^*$ and $A_{\mathrm{x}^*}^{-1}$ have analytical forms, allowing us to skip the entire process of \cref{alg-opt-interp}. 

\subsection{Optimal interpolation nodes are \texorpdfstring{$\frac{2 \pi}{n}$}{}-equidistant nodes}

We again use the notations in \cref{sec-interp-real} without the index $j$. 
As shown in \cref{app-recover-complex}, if the interpolation nodes $\mathbf{x} = \left[x_0, x_1, \dots, x_{2r}\right]^{\dagger} \in \mathbb{R}^n$ have distinct entries modulo $2\pi$, the matrix $A_{\mathbf{x}}$ must be non-singular, ensuring that $\hat{\mathbf{z}}_{\mathbf{x}}$ in \cref{eq-perturbed-equation} is well-defined.
The main findings of this subsection can be summarized as follows: the \textit{$\frac{2\pi}{n}$-equidistant nodes} $\mathbf{x}^* \in \mathbb{R}^n$ with $n = 2r + 1$, defined by
\begin{equation}\label{eq-2pi/n-equidistant}
    x_k^* = s + \frac{2\pi}{n}k, \quad k = 0,1,\ldots,2r,
\end{equation}
where $s \in \mathbb{R}$ is a shift value, achieve global optimality under the following three criteria simultaneously (and independently of $s$):
\begin{enumerate}
    \item Minimization of the mean squared error, $\operatorname{MSE}(\hat{\mathbf{z}}_{\mathbf{x}})$.
    \item Minimization of the condition number of the interpolation matrix $A_{\mathbf{x}}$.
    \item Minimization of the average variance of the estimated derivatives $\hat{f}^{(d)}(x)$ of all orders $d\geq 0$.
\end{enumerate}
\noindent Furthermore, in this case, one can readily verify that $A_{\mathbf{x}^*}^{-1}$ has an explicit form given by 
\begin{equation}
    A_{\mathbf{x}^*}^{-1} = \frac{2}{n}\left(\begin{array}{cccc}
    \tfrac{1}{\sqrt{2}} & \tfrac{1}{\sqrt{2}} & \cdots & \tfrac{1}{\sqrt{2}} \\
    \cos \left(x^*_0\right) & \cos \left(x^*_1\right) & \cdots & \cos \left(x^*_{2 r}\right) \\
    \sin \left(x^*_0\right) & \sin \left(x^*_1\right) & \cdots & \sin \left(x^*_{2 r}\right) \\
    \vdots & \vdots & & \vdots \\
    \cos \left(r x^*_0\right) & \cos \left(r x^*_1\right) & \cdots & \cos \left(r x^*_{2 r}\right) \\
    \sin \left(r x^*_0\right) & \sin \left(r x^*_1\right) & \cdots & \sin \left(r x^*_{2 r}\right)
    \end{array}\right) .
\end{equation}
However, note that this inverse expression does not hold for general nodes $\mathbf{x}$. In the following, we examine each of these optimality criteria in detail.


\subsubsection{Criteria I: minimal mean squared error}\label{subsec:first-view}


\begin{theorem}[Minimal mean squared error]\label{thm-first-veiw}
When \cref{assp-var,assm-equi} holds, the $\frac{2 \pi}{n}$-equidistant nodes $\mathbf{x}^*$ with an arbitrary shift value, as defined in \cref{eq-2pi/n-equidistant}, globally solves
\begin{equation}\label{eq-min-mse}
\mathbf{x}^* 
= \underset{\substack{\mathbf{x} \in \mathbb{R}^n \\ \text{ $x_i$ distinct modulo $2 \pi$}}}{\operatorname{argmin}} 
\operatorname{MSE}(\hat{\mathbf{z}}_{\mathbf{x}}) = \sigma^2 \|A_{\mathbf{x}}^{-1}\|_F^2,
\end{equation}
where the global minimum is $2\sigma^2$.
\end{theorem}
For proofs, see \cref{sec-proof-thm-1}. 

\subsubsection{Criteria II: minimal condition number}

We can also analyze the stability of linear equation from the perspective of classical numerical analysis. Using estimation error $\Delta \mathbf{z}_{\mathbf{x}}=\hat{\mathbf{z}}_{\mathbf{x}}-\mathbf{z}$, we can rewrite the perturbed linear \cref{eq-perturbed-equation} as
\begin{equation}
    A_{\mathbf{x}}(\Delta \mathbf{z}_{\mathbf{x}} + \mathbf{z}) = \mathbf{y}_{\mathbf{x}} + \mathbf{e}.
\end{equation}
Let $\|\cdot\|_2$ stand for the spectral norm and the condition number $\kappa_2(A_{\mathbf{x}})=\|A_{\mathbf{x}}\|_2\left\|A_{\mathbf{x}}^{-1}\right\|_2$.
Based on standard results in numerical stability analysis \cite{trefethen2022numerical}, $\kappa_2(A_{\mathbf{x}})$ provides an upper bound on the relative error by inequality
\begin{equation}
    \frac{\|\Delta \mathbf{z}_{\mathbf{x}}\|}{\|\mathbf{z}\|} \leq \kappa_2(A_{\mathbf{x}}) \frac{\|\mathbf{e}\|}{\|\mathbf{y}_{\mathbf{x}}\|}.
\end{equation}
This reveals that a smaller condition number of interpolation matrix $A_{\mathbf{x}}$ ensures better numerical stability for solutions. This motivates us to seek interpolation nodes that minimize $\kappa_2(A_{\mathbf{x}})$, leading to the following theorem. For proofs, see \cref{sec-proof-thm-2}. 

\begin{theorem}[Minimal condition number]\label{thm-second-veiw}
When \cref{assp-var,assm-equi} holds, the $\frac{2\pi}{n}$-equidistant nodes $\mathbf{x}^*$ with an arbitrary shift value, as defined in \cref{eq-2pi/n-equidistant}, globally solves
\begin{equation}\label{eq-min-cn}
\mathbf{x}^* 
=\underset{\substack{\mathbf{x} \in \mathbb{R}^n \\ \text{ $x_i$ distinct modulo $2\pi$}}}{\operatorname{argmin}} 
\kappa_2(A_{\mathbf{x}}) ,
\end{equation}
where the global minimum is $1$.
\end{theorem}
Since the condition number of any matrix is always greater than or equal to 1, we have achieved the minimal condition number, even in our cases where $A_{\mathbf{x}}$ possesses specific structural characteristics as in \cref{eq-A}.

\subsubsection{Criteria III: minimal average variance of derivatives}

In \cref{sec-mse-var}, we have discussed the variance $\operatorname{Var}[\hat{f} (x)]$, which quantifies the approximation accuracy of $\hat{f}$ to the true function $f$.
If we treat $\hat{f}^{(0)}\equiv \hat{f}$ and $f^{(0)}\equiv f$ as zero-order derivatives, a similar discussion can be extended to derivatives of any $d$ order. Let $\mathbf{t}^{(0)}(x) \equiv \mathbf{t}(x)$ defined as \cref{eq-t0x} and for any integer $d \geq 1$, let
\begin{equation}\label{eq-tdx}
\mathbf{t}^{(d)}(x):= \begin{bmatrix}
    0 \\
    \Omega_1^d \cos(\Omega_1 x + \frac{d\pi}{2}) \\
    \Omega_1^d \sin(\Omega_1 x + \frac{d\pi}{2}) \\
    \vdots \\
    \Omega_r^d \cos(\Omega_r x + \frac{d\pi}{2}) \\
    \Omega_r^d \sin(\Omega_r x + \frac{d\pi}{2})
    \end{bmatrix} \in \mathbb{R}^n
\end{equation}
denote the $d$-th derivative of $\mathbf{t}^{(0)}(x)$. Then, the $d$-th order derivative of $\hat{f}(x)$ is $\hat{f}^{(d)}(x)= \mathbf{t}^{(d)}(x)^{\dagger} \hat{\mathbf{z}}_{\mathbf{x}}.$ Similar to the discussion in \cref{sec-mse-var}, $\hat{f}^{(d)}(x)$ is normal and unbiased estimates of the true derivatives $f^{(d)}(x)$ since
\begin{equation}
    \operatorname{E}[\hat{f}^{(d)}(x)]=f^{(d)}(x).
\end{equation}
Again by \cref{lem-cov-var} of \cref{sec-proof-thm-3},
\begin{equation}
\operatorname{Var}[\hat{f}^{(d)}(x)]
=\langle\operatorname{VAR}[\hat{\mathbf{z}}_{\mathbf{x}} ], \mathbf{t}^{(d)}(x) \mathbf{t}^{(d)}(x)^{\dagger}\rangle .
\end{equation}

Rather than focusing on the variance at specific $x$, we are more interested in the average variance of $\hat{f}^{(d)}(x)$ over its domain. When \cref{assm-equi} holds $\left(\forall \Omega_k=k\right)$, both $f^{(d)}$ and $\hat{f}^{(d)}$ are periodic functions with a period of $2\pi$. Thus, it suffices to consider the average variance over $[0,2\pi)$:
\begin{equation}\label{eq-hx}
\begin{aligned}
h^{(d)}(\mathbf{x}):=\frac{1}{2 \pi} \int_0^{2 \pi} \operatorname{Var}[\hat{f}^{(d)}(x)] \mathrm{~d} x
=\left\langle\operatorname{VAR}[\hat{\mathbf{z}}_{\mathbf{x}} ], \mathcal{I}_d\right\rangle,
\end{aligned}
\end{equation}
where the matrix
\begin{equation}\label{eq-int-ttt}
\mathcal{I}_d := \frac{1}{2 \pi} \int_0^{2 \pi} \mathbf{t}^{(d)}(x) \mathbf{t}^{(d)}(x)^{\dagger}  \mathrm{~d} x
    =
    \begin{cases}
    \frac{1}{2} I_n, & \text{if } d=0, \\
    \frac{1}{2} \operatorname{diag}(0, 1, 1, 2^{2d}, 2^{2d}, \ldots, r^{2d}, r^{2d}), & \text{if } d \geq 1,
    \end{cases}
\end{equation}
is constant with respect to $\mathbf{x}$. Then, $h^{(d)}(\mathbf{x})$ measures the quality (at an average level) of how well the derivatives of the approximated function estimate the true derivatives. This motivates us to seek interpolation nodes that minimize the value of $h^{(d)}(\mathbf{x})$, leading to the following theorem. For proofs, see \cref{sec-proof-thm-3}. 

\begin{theorem}[Minimal average variance of derivatives]\label{thm-third-veiw}
Given any integer $d \geq 0$. When \cref{assp-var,assm-equi} holds, the $\frac{2\pi}{n}$-equidistant nodes $\mathbf{x}^*$ with an arbitrary shift value, as defined in \cref{eq-2pi/n-equidistant}, globally solves
\begin{equation}\label{eq-min-ave}
\mathbf{x}^* 
=\underset{\substack{\mathbf{x} \in \mathbb{R}^n \\ \text{$x_i$ distinct modulo $2\pi$}}}{\operatorname{argmin}} 
h^{(d)}(\mathbf{x})=\left\langle\operatorname{VAR}[\hat{\mathbf{z}}_{\mathbf{x}} ], \mathcal{I}_d\right\rangle,
\end{equation}
where the global minimum is $2\sigma^2$ if $d=0$, and $\frac{2\sigma^2}{n} \sum_{k=1}^r k^{2d}$ if $d \geq 1$.
\end{theorem}

For the special case of $d=0$, \cref{eq-int-ttt} shows that
\begin{equation}\label{eq-case-d=0}
    h^{(0)}(\mathbf{x})
    =\frac{1}{2}\left\langle\operatorname{VAR}[\hat{\mathbf{z}}_{\mathbf{x}} ], I_n\right\rangle
    =\frac{1}{2} \operatorname{MSE}(\hat{\mathbf{z}}_{\mathbf{x}} ).
\end{equation}
This coincides with the problem addressed in \cref{thm-first-veiw} and establishes the same result. An important observation is that the $\frac{2\pi}{n}$-equidistant nodes constitute the optimal solution $\mathbf{x}^*$ for all $d \geq 0$, but the minimum average variance grows exponentially with respect to $d$. 

\begin{remark}
    In the case of equidistant frequencies, we find that the \(\frac{2\pi}{n}\)-equidistant nodes can simultaneously minimize the three different optimal criteria. However, this result does not hold in the case of general frequencies. In general, minimizing any one of the three criteria, such as the MSE criteria, will also reduce the cost function value of the other criteria, but it cannot achieve the optimality simultaneously.
\end{remark}

\subsection{More discussion on variance of approximated functions}

This subsection analyzes the variance properties of $\hat{f}(x)$ and $\tilde{f}(x)$, showing they share the same variance under specific conditions, and compares the derivative estimators of ICD and PSR, highlighting ICD's computational efficiency by avoiding repeated quantum evaluations.
Throughout this subsection, suppose that \cref{assm-equi} holds, and we have chosen the $\frac{2 \pi}{n}$-equidistant nodes $\mathbf{x}^*$. Then, the covariance matrix of $\hat{\mathbf{z}}_{\mathbf{x}}$ becomes $\operatorname{VAR}[\hat{\mathbf{z}}_{\mathbf{x}^*}] =
    \frac{2\sigma^2}{n} I.$  In this case, as shown in \cref{lem-cov-var} of \cref{sec-proof-thm-3}, we have (for any $d \geq 0$):
\begin{align}
\operatorname{Cov}[\hat{f}^{(d)}(x_1), \hat{f}^{(d)}(x_2)]
&= \frac{2}{n} \sigma^2\cdot\mathbf{t}^{(d) }(x_1) ^{\dagger}  \mathbf{t}^{(d)}(x_2), \quad \forall x_1, x_2 \in \mathbb{R}, \\
\operatorname{Var}[\hat{f}^{(d)}(x)] &=\frac{2}{n} \sigma^2\cdot\mathbf{t}^{(d) }(x) ^{\dagger}  \mathbf{t}^{(d)}(x), \quad \forall x \in \mathbb{R}.
\end{align}

\subsubsection{Comparison of function evaluation variance: \texorpdfstring{$\hat{f}(x)$ vs. $\tilde{f}(x)$}{}}

By setting $d=0$ and defining $s := |x_1 - x_2|$, we obtain
\begin{align}\label{eq-cov-hatf}
\operatorname{Cov}[\hat{f}(x_1), \hat{f}(x_2)]
=\frac{2}{n} \sigma^2\left(\frac{1}{2}+\sum_{k=1}^r \cos (k s)\right) 
=\frac{\sin \left(\frac{n}{2} s\right)}{n \sin \left(\frac{1}{2} s\right)} \sigma^2. 
\end{align}
Notably, for $x = x_1 = x_2$, we have
\begin{equation}\label{eq-1101}
    \operatorname{Var}[\hat{f}(x)] = \operatorname{Cov}[\hat{f}(x), \hat{f}(x)] = \sigma^2, \quad  \forall x \in \mathbb{R}.
\end{equation}
This leads to the assertion that
\begin{equation}\label{eq-hatf-normal}
    \hat{f}(x) \sim \mathcal{N}(f(x), \sigma^2).
\end{equation}
Interestingly, this distribution is identical to $\tilde{f}(x) \sim \mathcal{N}(f(x), \sigma^2)$ given in \cref{assm-tilde-f} from Observation I. But, what is the difference? Recall that evaluating $x \mapsto \tilde{f}(x)$ incurs quantum overhead, as it requires numerous measurements after running the quantum circuit. Specifically, when the argument $x$ is modified, the entire process must be repeated from the beginning. Conversely, evaluating $x \mapsto \hat{f}(x)$ for arbitrary $x$ incurs no quantum overhead, as we rely entirely on the approximated function \cref{eq-app-fx} using a classical computer. It is important to note, however, that $\tilde{f}(x)$ are independent random variables for any two distinct $x$, which is generally not true for $\hat{f}(x)$. The last equality in \cref{eq-cov-hatf} implies that $\operatorname{Cov}[\hat{f}(x_1), \hat{f}(x_2)] = 0$ if and only if $|x_1 - x_2| = \frac{2 \pi}{n} k$ for any integer $k \neq 0 \pmod{n}$.

\subsubsection{Comparison of derivative variance: ICD vs. PSR}

In general, for any $d \geq 1$, we have
\begin{equation}\label{eq-var-hatf-orderd}
\operatorname{Var}[\hat{f}^{(d)}(x)] = \frac{2}{n} \sigma^2 \cdot \mathbf{t}^{(d)}(x)^{\dagger} \mathbf{t}^{(d)}(x) = \frac{2}{n} \sigma^2 \sum_{k=1}^r k^{2d}, \quad \forall x \in \mathbb{R}.
\end{equation}
This result indicates that, with $\frac{2 \pi}{n}$-equidistant interpolation nodes, the quantity above not only represents the minimum average variance as established in \cref{thm-third-veiw}, but also provides the same variance across all $x \in \mathbb{R}$. Given the central role of the first-order derivative in optimization, we define the symbols 
\begin{equation}
    g(x) := \frac{\mathrm{d} f}{\mathrm{d} x}(x) = \mathbf{t}^{(1)}(x)^{\dagger} \mathbf{z}
\end{equation}
and 
\begin{equation}
    g_{\text{icd}}(x) := \frac{\mathrm{d} \hat{f}}{\mathrm{d} x}(x) = \mathbf{t}^{(1)}(x)^{\dagger} \hat{\mathbf{z}}_{\mathbf{x}}.
\end{equation}
Setting $d = 1$ in \cref{eq-var-hatf-orderd} yields
\begin{equation}
    \operatorname{Var}[g_{\text{icd}}(x) ]
    = \frac{2}{n} \sigma^2 \sum_{k=1}^r k^2
    = \frac{r(r+1)}{3} \sigma^2= \mathcal{O}(r^2\sigma^2), \quad \forall x \in \mathbb{R}.
\end{equation}
Thus, unbiased $g_{\text{icd}}$ has constant variance across all $x$. 

This reminds us of another approach for unbiasedly estimating the true derivative in PQCs, namely, parameter shift rule (PSR) \cite{crooks2019gradients,mari2021estimating,wierichs2022general,kyriienko2021generalized,hai2023lagrange,markovich2024parameter,hoch2025variational}. We skip the technical details of PSR and focus only on the conclusion; for more, see \cref{app-psr}. Under the equidistant frequencies in \cref{assm-equi}, the PSR \cite{wierichs2022general} employs the following estimator to approximate $g(x)$: 
\begin{equation}
g_{\text{psr}}(x) := \sum_{\mu=1}^{2r} \frac{(-1)^{\mu-1}}{4r \sin^2\left(\frac{1}{2} x_\mu\right)} \tilde{f}\left(x + x_\mu\right), \quad \forall x \in \mathbb{R},
\end{equation}
where $x_\mu = \frac{\pi}{2r} + (\mu - 1) \frac{\pi}{r}$ for $\mu = 1, 2, \ldots, 2r$. In fact, \cite{wierichs2022general} demonstrated that $\operatorname{E}[g_{\text{psr}}(x)] = g(x)$ and
\begin{align}
    \operatorname{Var}[g_{\text{psr}}(x)]
    = \sum_{\mu=1}^{2r} \frac{1}{16r^2 \sin^4\left(\frac{1}{2} x_\mu\right)} \operatorname{Var}[\tilde{f}\left(x + x_\mu\right)]
= \frac{2r^2 + 1}{6} \sigma^2= \mathcal{O}(r^2\sigma^2), \quad \forall x \in \mathbb{R},
\end{align}
where we use the fact that $\tilde{f}(x)$ are independent random variables for any two distinct $x$, and $\operatorname{Var}[\tilde{f}(x + x_\mu)] \approx \sigma^2$ according to \cref{assp-var}. We can see that both $g_{\text{icd}}$ and $g_{\text{psr}}$ serve as unbiased estimators of $g(x)$ for any $x$. 
The variance of $g_{\text{icd}}$ is slightly higher than that of $g_{\text{psr}}$, but the difference becomes negligible as $r \to \infty$. 

Notably, $g_{\text{icd}}$ is computed entirely on a classical computer from $\hat{f}$, which requires a total of $n = 2r + 1$ evaluations of $\tilde{f}$. Once $\hat{f}$ is available, we can calculate $g_{\text{icd}}(x)$ for any $x$ without further quantum overhead. In contrast, $g_{\text{psr}}(x)$ requires recalculating $\tilde{f}$ for $2r$ times for each new argument $x$. Similarly, as before, the unbiased derivatives $g_{\text{icd}}$ at different $x$ values are generally not independent, whereas $g_{\text{psr}}$ is typically independent. Overall, although both $g_{\text{psr}}$ and $g_{\text{icd}}$ are unbiased estimators of the true $g$, they are derived from two different techniques: $g_{\text{psr}}$ is obtained via the parameter shift rule, while $g_{\text{icd}}$ is derived through interpolation.

\section{Numerical simulation}\label{sec-experiments}

In this section, we use numerical experiments to demonstrate the efficiency of the ICD algorithms. We will use two versions of ICD: the standard ICD, shown in \cref{alg-standard-ICD}, and the reduced ICD, presented in \cref{alg-reduced-ICD} in \cref{app-pratical-ICD}.
To implement the quantum circuits, we used the IBM Qiskit \cite{qiskit}, which simulates the sampling process and the noisy environment of a real machine.
All tests were executed on a computer equipped with an AMD Ryzen 7 8845H CPU and 32 GB of RAM. 
The code is publicly available.\footnote{\url{https://github.com/GALVINLAI/ICD_for_VQA}}
We use energy error as the performance metric. Let $E_{\text{ground}}$ denote the true ground state energy of the given Hamiltonian $H$. The energy error is defined as $ E_{\text{error}}: = |E (\boldsymbol{\theta}^*) - E_{\text{ground}}|,$ where $E (\boldsymbol{\theta}) = \langle \psi (\boldsymbol{\theta}) | H | \psi (\boldsymbol{\theta}) \rangle$ is the true expectation value of the cost function.
We consider the ground state to be successfully found if $E_{\text{error}} < 10^{-2}$.

\subsection{Problem Setting}\label{subsec-num-problems}

We consider the MaxCut problem, the transverse-field Ising model (TFIM), and the XXZ model.
Further details regarding these models and their corresponding quantum circuits can be found in \cite{wiersema2020exploring,ding2024random}.
Let $N$ denote the number of qubits, and define the Pauli operators as follows:
\begin{equation}
    X \equiv \begin{pmatrix}
    0 & 1 \\
    1 & 0
    \end{pmatrix}, \quad
    Y \equiv \begin{pmatrix}
    0 & -i \\
    i & 0
    \end{pmatrix}, \quad
    Z \equiv \begin{pmatrix}
    1 & 0 \\
    0 & -1
    \end{pmatrix}.
\end{equation}
The symbol $X_i$ denotes a Pauli $X$ operator acting on the $i$-th qubit, with $Y_i$ and $Z_i$ defined analogously.

\subsubsection{MaxCut with HEA}\label{sec-pro-maxcut}

We consider a simple MaxCut problem defined on an undirected graph $G = (V, E)$, where the set of nodes is $V = \{1, 2, 3, 4\}$ and the set of edges is $E = \{(1,2),(1,3),(1,4),(2,3),(3,4)\} .$ The goal is to maximize the expression $\sum_{ (i, j) \in E} x_i (1 - x_j)$ with $x_i \in\{0,1\}$. For this graph, the corresponding MaxCut Hamiltonian can be formulated as \cite{ding2024random,grange2023introduction}
\begin{equation}
    H=\frac{1}{2} I-3 Z_1+\frac{1}{2} Z_1 Z_2+\frac{1}{2} Z_1 Z_3+\frac{1}{2} Z_1 Z_4+\frac{1}{2} Z_2 Z_3+\frac{1}{2} Z_3 Z_4 .
\end{equation}
This is a problem involving $N=4$ qubits. We construct a $p = 5$-layer hardware efficient ansatz (HEA) circuit initialized in the $|0\rangle^{\otimes 4}$ state, as illustrated in \cref{fig:qc_maxcut} of \cref{app_qc}.
Since all tunable gates in the HEA are $R_Y(\theta_i) = e^{-i \frac{\theta_i}{2}Y }$, the frequency set corresponding to each parameter is the singleton $\{1\}$ with $r= 1$ and $n= 3$. 

\subsubsection{TFIM model with HVA}\label{sec-pro-tfim}

The Hamiltonian for TFIM is given by
\begin{equation}
    H_{\mathrm{TFIM}} = H_{zz} + \Delta\, H_x
\end{equation}
where
\begin{equation}
    H_{zz} = \sum_{i=1}^N Z_i Z_{i+1},
    \quad
    H_x = \sum_{i=1}^N X_i,
\end{equation}
with $\Delta=0.5$ and periodic boundary conditions such that $Z_{N+1}\equiv Z_1$. 
We consider a depth-$p$ Hamiltonian variational ansatz (HVA) circuit \cite{wiersema2020exploring} for the TFIM, initialized in the $|+\rangle^{\otimes N}$ state, which corresponds to
\begin{equation}
U(\boldsymbol{\beta}, \boldsymbol{\gamma})
=
\prod_{l=1}^p 
G\left(\gamma_l, H_x\right)G\left(\beta_l, H_{z z}\right),
\end{equation}
where $G(x, H)=\exp \left(-i \frac{x}{2} H \right)$. Hence, for a depth-$p$ circuit, we have $2 p$ parameters. \cref{fig:qc_tifm_HVA} in \cref{app_qc} illustrates the quantum circuit for $N=4$ and $p=1$.

The frequencies associated with $\beta_l$ and $\gamma_l$ require more careful consideration. 
For the generator $H_x$, the eigenvalues range from $-N$ to $N$ in steps of 2, yielding a frequency set $\{2, 4, \ldots, 2N\}$, which reduces to $\{1, 2, \ldots, N\}$ since the prefactor $\frac{1}{2}$ is included in $G(x, H)$. Similarly, $H_{zz}$ has eigenvalues of the form $N - 4k$ for $k = 0, 1, \ldots, \lfloor N/2 \rfloor$, giving rise to frequencies $\{4, 8, \ldots, 4\lfloor N/2 \rfloor\}$, which become $\{2, 4, \ldots, 2\lfloor N/2 \rfloor\}$ under the same prefactor.
However, for any $N$ qubits and any $p$ layers of HVA, we observe that the actual frequency for each $\beta_l$ and $\gamma_l$ corresponds solely to the singleton set $\Omega = \{2\}$. This arises from the fact that $H_x$ and $H_{zz}$ serve as both the observables and the gate generators. A detailed derivation is provided in \cref{app-sparse-frequency}.

\subsubsection{XXZ model with HVA}\label{sec-pro-xxz}

The Hamiltonian for the XXZ model is given by
\begin{equation}
    H_{\mathrm{XXZ}} =H_{x x}+H_{y y}+\Delta H_{z z}
\end{equation}
where
\begin{equation}
    H_{xx} = \sum_{i=1}^N X_i\, X_{i+1},
    \quad
    H_{yy} = \sum_{i=1}^N Y_i\, Y_{i+1},
    \quad
    H_{zz} = \sum_{i=1}^N Z_i\, Z_{i+1}.
\end{equation}
The parameter $\Delta=0.5$ controls the spin anisotropy in the model. Also, we use periodic boundary conditions, i.e., $X_{N+1}\equiv X_1$, $Y_{N+1}\equiv Y_1$, and $Z_{N+1}\equiv Z_1$. 
We initialize the state as $\left|\psi_0\right\rangle=\bigotimes_{i=1}^{N / 2}\left|\Psi^{-}\right\rangle$ where $\left|\Psi^{-}\right\rangle=\frac{1}{\sqrt{2}}(|01\rangle-|10\rangle)$ is the Bell state of qubits $2 i-1$ and $2 i$. A depth-$p$ HVA circuit \cite{wiersema2020exploring} for the XXZ model corresponds to
\begin{equation}
    \begin{aligned}
    U(\boldsymbol{\theta}, \boldsymbol{\phi}, \boldsymbol{\beta}, \boldsymbol{\gamma})
    &= \prod_{l=1}^p
    G\bigl (\gamma_l, \, H_{xx}^{\mathrm{even}}\bigr)
    \; G\bigl (\gamma_l, \, H_{yy}^{\mathrm{even}}\bigr)
    \; G\bigl (\beta_l, \, H_{zz}^{\mathrm{even}}\bigr)\\
    &\quad\; \times
    G\bigl (\phi_l, \, H_{xx}^{\mathrm{odd}}\bigr)
    \; G\bigl (\phi_l, \, H_{yy}^{\mathrm{odd}}\bigr)
    \; G\bigl (\theta_l, \, H_{zz}^{\mathrm{odd}}\bigr),
    \end{aligned}
\end{equation}
where $G(x, H)=\exp \left(-i \frac{x}{2} H \right)$, $H_{xx}^{\mathrm{even}}=\sum_{i=1}^{N / 2} X_{2i-1} X_{2i}$ and $H_{xx}^{\text {odd }}=\sum_{i=1}^{N / 2} X_{2i} X_{2i+1}$, with $H_{yy,zz}^{\mathrm{even}}$ and $H_{yy,zz}^{\mathrm{odd}}$ defined analogously. Hence, for a depth-$p$ circuit, we have $4 p$ parameters. 
\cref{fig:qc_XXZ_HVA} in \cref{app_qc} illustrates a quantum circuit for $N=6$ and $p=1$.

The effective frequency set for XXZ with HVA is relatively intricate, depending on the number of qubits $N$, the circuit depth $p$, and the parameter position. However, the total number of frequencies $r$ grows linearly with $N$. 
Here, we only provide the specific $N, p$ configurations used in our experiments later. 
For $N = 6, p = 12$, all $\theta_l, \beta_l$ have $\{2\}$; all $\phi_l, \gamma_l$ have $ \{2, 4\}$.
For $N = 8, p = 16$, all $\theta_l, \beta_l$ have $\{2,4\}$; all $\phi_l, \gamma_l$ have $ \{2, 4, 6,8\}$, except for the last layer where $\theta_{16}, \beta_{16}$ have $\{2\}$ and $\phi_{16}, \gamma_{16}$ have $\{2, 4\}$.
For $N = 10,p = 20$, same as above.
For $N = 12,p = 24$, all $\theta_l, \beta_l$ have $\{2,4,6\}$; all $\phi_l, \gamma_l$ have $\{2, 4, 6, 8, 10, 12\}$, 
except for the second-to-last layer where $\theta_{23}, \beta_{23}$ have $\{2, 4\}$ and $\phi_{23}, \gamma_{23}$ have $\{2, 4, 6, 8\}$, and the last layer where $\theta_{24}, \beta_{24}$ have $\{2\}$ and $\phi_{24}, \gamma _{24}$ have $\{2, 4\}$.

\subsection{Result I: Impact of interpolation node optimality
for ICD methods}\label{subsec-num-re-1}

In this subsection, we use the MaxCut problem with HEA as an example to illustrate the impact of interpolation node selection on ICD algorithms. 
Since each parameter in MaxCut with HEA corresponds to the singleton frequency set $\Omega = \{1\}$, the optimal interpolation nodes are equally spaced with spacing $2\pi/3$. 
For comparison, we also consider equidistant nodes with spacing $k\pi/3$, where $k \in (0, 3)$. Without loss of generality, we assume the starting point of nodes is zero.

\subsubsection{Verification of the three optimality criteria}\label{sec-r1-1}

We first examine the numerical values of the three optimality criteria for equally spaced interpolation nodes with different spacings.
\cref{fig:metrics_with_diff_iterp_nodes} presents the three criteria from \cref{thm-first-veiw,thm-second-veiw,thm-third-veiw} as functions of the value $k$, showing that all three criteria simultaneously attain their minimum values when $k = 2$.
Both large and small spacings lead to increases in all three criteria.
This result supports the correctness of our theoretical findings in \cref{thm-first-veiw,thm-second-veiw,thm-third-veiw}. 
Note that we adopt the constant variance \cref{assp-var} and factor out the unknown constant variance $\sigma^2$ from the numerical evaluation of the three optimality criteria; for example, the mean squared error (MSE) in \cref{eq-min-mse} is given by the squared Frobenius norm of $A_{\mathbf{x}}^{-1}$, which attains a minimum of 2.
For the average variance of derivatives in \cref{thm-third-veiw}, we consider $h^{(1)}$ of order $d = 1$, where the minimum is $2/3$.

\subsubsection{Impact of interpolation nodes in the noiseless and noisy settings}\label{sec-r1-2}

In what follows, for ease of demonstration, we adopt a cyclic update scheme in the ICD algorithms to reduce randomness. The initial parameters are chosen uniformly from the range $[0, 2\pi]$. 
We select four sets of equally spaced interpolation nodes with spacings $k\pi/3$, where $k = 0.5, 1, 1.5,$ and $2$, corresponding to round MSE values of 225, 13, 3, and 2, respectively.
As $k$ increases to 2, the node spacing approaches the optimal.
Notably, the setting $k = 1.5$ (i.e., $\pi/2$), which we refer to as a \textit{suboptimal} configuration, is exactly the choice used in Rotosolve \cite{ostaszewski2021structure} and SMO \cite{nakanishi2020sequential} (see the Interpolation nodes spacing row in \cref{tab:comparison}).
These two works did not consider the relationship between adjustable interpolation nodes and sampling noise.
Later, we will show that this suboptimal configuration performs slightly worse than the optimal one in noisy experiments.

\paragraph{Noiseless experiments: all interpolation nodes are equally effective.}
We begin by evaluating the impact of interpolation nodes on the ICD algorithms in the ideal noiseless setting (i.e., $\infty$ shots). As shown in \cref{fig:True-Data-Fixed-k-Nodes}, all four node settings lead to identical convergence behavior. 
If we randomly select four sets of interpolation nodes (see  \cref{fig:True-Data-Random-Nodes}), with MSEs around 3569.1, 280.5, 1541.2, and 2.0 respectively, the convergence behavior remains unchanged: all configurations perform identically, and the trajectories overlap. These results demonstrate that in the noiseless setting (infinite shots), the choice of interpolation nodes has no impact on the performance of ICD.

\paragraph{Noise experiments: standard ICD is more robust than reduced ICD.}
We now turn to the noisy setting (with 1024 shots) and examine the robustness of the interpolation nodes using both the reduced and standard versions of ICD.
We first consider the performance of the reduced ICD under the four node configurations. 
As shown in \cref{fig:Noisy-Data-Practical-ICD}, only the optimal node setting ($k = 2$) achieves stable convergence. 
The settings $k = 1$ and $k = 1.5$ lead to an initial phase of normal descent, but eventually diverge and fail to converge. The $k = 0.5$ setting fails entirely from the beginning.
Next, we examine the standard ICD under the same four configurations. 
As shown in \cref{fig:Noisy-Data-Vanilla-ICD}, except for $k = 0.5$, all settings result in successful convergence.
Moreover, the overall performance is noticeably more stable compared to \cref{fig:Noisy-Data-Practical-ICD}.

This suggests that the standard ICD is more robust to the choice of interpolation nodes than the reduced ICD.
Even when the nodes are not exactly optimal but only approximately optimal, standard ICD still achieves good convergence performance.
In contrast, the reduced ICD tends to amplify noise over iterations.
This is because the reduced ICD's estimation at the first interpolation node reuses information from the previous iteration, leading to a loss of sampling independence and consequently, the accumulation and propagation of errors.
Therefore, despite saving one function evaluation per iteration, the reduced version lacks robustness unless optimal interpolation nodes are used.
As a special case of the ICD framework, SMO \cite{nakanishi2020sequential} introduces a remedy for this issue by periodically resetting the first interpolation node to be re-evaluated on the quantum device.

\subsubsection{Interpolation nodes closer to optimal spacing enable ICD convergence with fewer shots}

\label{sec-r1-3}
Motivated by the preceding experimental results, we further investigate whether there exists a relationship between the optimality of interpolation nodes and the number of measurement shots. Specifically, we will show that interpolation nodes closer to the optimal spacing allow the ICD algorithm to achieve effective convergence with fewer shots.

To this end, we adopt the standard ICD and vary the number of shots starting from 4096, halving each time down to 2 shots. We consider the same four sets of equally spaced interpolation nodes as in previous experiments. The results are shown in \cref{fig:maxcut_oicd_all_nshot_subplots_1}.
When the number of shots is 4096, all four nodes configurations converge successfully within 300 iterations.
For $k = 0.5$, the performance of ICD deteriorates rapidly as the number of shots decreases, exhibiting increasing fluctuations and eventual failure. 
For $k = 1$, the algorithm remains stable for shot counts greater than 512; however, as shots decreases from 512 to 2, convergence performance gradually degrades. 
In contrast, the configurations with $k = 1.5$ and $k = 2$ perform similarly: both achieve convergence when the number of shots is at least 256, and demonstrate significantly smaller fluctuations at 128, 64 and 32 shots compared to $k = 0.5$ and $k = 1$. When shots falls below 64, the noise becomes overwhelming, and even the optimal interpolation setting ($k = 2$) fails to ensure convergence. This suggests that in the extreme low shot regime, no interpolation strategy is effective.

To better examine the subtle differences between the $k = 1.5$ and $k = 2$ settings, we plot their convergence behaviors for shot counts ranging from 64 to 36 in \cref{fig:maxcut_oicd_all_nshot_subplots_2}. Recall that $k = 1.5$ (i.e., $2\pi/2$ spacing) is the choice used in Rotosolve \cite{ostaszewski2021structure} and SMO \cite{nakanishi2020sequential}. Overall, as shots number decreases, the $k = 1.5$ appears more sensitive to noise than $k = 2$, consistently exhibiting larger fluctuations during optimization. In summary, our experiments yield the following key observations:

\begin{itemize}
    \item In the high-shot regime (approximating a noiseless setting), ICD is insensitive to the choice of interpolation nodes; all configurations yield stable convergence.
    \item In the low-shot regime (high noise levels), ICD fails to converge under any interpolation nodes. In the most extreme case, such as with only single shot, effective optimization method appears infeasible.
    \item In the moderate-shot regime, interpolation nodes that are closer to the optimal spacing (corresponding to lower MSE) enable ICD to maintain convergence even under reduced shot counts.
\end{itemize}

These findings suggest that, provided the noise level is not excessively high, ICD exhibits the robustness with respect to the choice of interpolation nodes. 
For problems with non-equidistant frequencies, it is generally not possible to determine the exact optimal interpolation nodes. In such cases, numerically suboptimal nodes can still ensure convergence of ICD, provided that the number of shots is sufficiently large.

\subsection{Result II: Compare ICD methods with standard algorithms}\label{sec-tfim-xxz}

In this subsection, we use the TFIM model and XXZ model with HVA as examples to compare ICD algorithms with two standard algorithms: stochastic gradient descent (SGD) \cite{sweke2020stochastic} and random coordinate descent (RCD) \cite{ding2024random}. 
As shown in \cref{tab:comparison}, various structure optimization methods can be seen as special cases of ICD. For broader comparisons with other derivative-free and gradient-based methods, please refer to the Baseline algorithms row of \cref{tab:comparison}.

We set 1024 shots. 
The initial parameters are chosen uniformly from the range $[0, 2\pi]$. 
We use reduced ICD and use the eigenvalue method (see \cref{app-eig-method}) to exactly solve the subproblem. 
For ICD and RCD, we randomly select the coordinate index $j$ for updates. 
For SGD and RCD, we need to compute unbiased derivatives. To this end, we apply the general parameter shift rule stated in \cref{app-psr}. 
The learning rates for SGD and RCD are set to 0.01 and 0.02, respectively. 
In general, the performance of SGD and RCD is highly dependent on the setting of the learning rate, which varies across different problems. 
In contrast, our ICD method does not require the adjustment of any hyperparameters. 
For each model, we conducted the experiment 5 times, using different random initializations for each run, but all methods started from the same initial points. 
For each method, we plot the mean values against the number of function evaluations $\boldsymbol{\theta} \mapsto f(\boldsymbol{\theta})$, with the $x$-axis indicating the quantum overhead.

\subsubsection{TFIM with HVA: ICD methods outperform than SGD, RCD}\label{sec-r2-tfim}

The performance comparisons of ICD, SGD, and RCD for the TFIM model are shown in \cref{fig:tfim_hva_many}. 
We evaluated systems with qubit numbers $N = 6, 8, 10, 12, 14, 16$ and set the ansatz depth to $p = 2N$. 
Across all cases in \cref{fig:tfim_hva_many}, ICD consistently finds the optimal solution with fewer function evaluations compared to SGD and RCD.
In particular, coordinate descent (CD)-based algorithms (including RCD and ICD) are significantly more efficient than SGD.
ICD achieves the most rapid descent in the early stages of optimization, but as the iterates approach the optimum, the performance of RCD and ICD becomes similar. From a numerical perspective, this is because when the parameters are near the optimal solution, the decrease achieved by performing an argmin update is comparable to that from taking a single gradient step.

\subsubsection{XXZ with HVA: ICD fails to overcome the barren plateau}\label{sec-r2-xxz-bp}

One of the most significant challenges in PQC optimization is the well-known barren plateau (BP) problem \cite{mcclean2018bp,larocca2025barren}.
The BP typically refers to the exponential concentration of (some or all) partial derivatives, i.e., when parameters $\boldsymbol{\theta}$ are sampled uniformly over $[0, 2\pi]$, the variance of the partial derivatives decays exponentially with the number of qubits $N$. This phenomenon is equivalent to the exponential concentration of the cost function values themselves; see \cite{arrasmith2022equivalence}.
Unfortunately, our ICD algorithms cannot theoretically mitigate the barren plateau (BP) problem.
For any PQC cost function that exhibits a BP, the corresponding Fourier coefficients also decay exponentially; as shown and analyzed in detail in \cite{okumura2023fourier}, this leads to an essentially flat optimization landscape.
As a result, unless an exponential number of measurement shots is performed, the interpolation steps in ICD will be dominated by uninformative statistical fluctuations.

To underscore this challenge, we present a numerical experiment that exhibits the BP. 
According to \cite{wiersema2024classification}, the XXZ model with an HVA ansatz exhibits an exponentially large dynamical Lie algebra, which leads to the presence of the BP.
In the right column of \cref{fig:xxz_hva_many}, we illustrate this phenomenon using both the cost function values and the first-order partial derivatives.
Specifically, we uniformly sampled 50 parameter vectors and computed their statistical variances.
As shown, both variances decay exponentially with increasing qubit number, although at a relatively moderate rate \cite{wiersema2020exploring}.

We evaluated the performance of ICD, SGD, and RCD on this problem using systems with qubit numbers $N = 6, 8, 10, 12$ and ansatz depth $p = 2N$.
As shown in the left and middle columns of \cref{fig:xxz_hva_many}, as the number of qubits increases and the BP becomes more severe, ICD begins to struggle in reaching the optimal solution, while RCD and SGD exhibit increasingly oscillatory behavior.
Empirically, in regimes where the barren plateau effect is not overly severe, ICD still demonstrates superior performance compared to gradient-based optimizers.
It is worth noting that both RCD and SGD require careful tuning of the learning rate to achieve reasonable convergence: small learning rates lead to slow convergence, while large ones cause instability and oscillations, as illustrated in \cref{fig:xxz_hva_many}.
In contrast, our ICD method does not require hyperparameter tuning.

\section{Discussion}\label{sec-discussion}

In this work, we propose an Interpolation-based Coordinate Descent (ICD) method as a general and unifying framework for structure optimization methods in parameterized quantum circuits, such as Rotosolve and SMO, ExcitationSolve and others.
The ICD method approximates the trigonometric structure of the cost function via interpolation and performs an exact minimization over a single parameter in each iteration, using only function evaluations. Unlike previous structure optimization methods, our approach rigorously determines the optimal interpolation nodes to reduce the impact of statistical noise from quantum measurements. We show that in the case of equidistant frequency spectra, the optimal nodes are $2\pi/n$-equispaced under the constant variance assumption, and such a choice simultaneously minimizes three key criteria: mean squared error, condition number, and average variance of derivatives. Through numerical experiments on benchmark problems such as MaxCut, the transverse-field Ising model, and the XXZ model, we validate that ICD achieves more efficient convergence than stochastic gradient descent and random coordinate descent.

The Fourier structure in quantum machine learning has already been explored for the inference stage: \cite{schreiber2023classical} Fourierizes the model's mapping with respect to the input $x$, enabling a  post-training classical surrogate for inference and a baseline for quantum advantage; by contrast, we Fourierize the cost function's dependence on a parameter $\theta_j$ and, on that basis, design the ICD optimizer for training. Thus, \cite{schreiber2023classical} answers ``how to deploy or evaluate a quantum model classically after training,'' whereas we answer ``how to tune parameters more efficiently during training.'' Both exploit Fourier/trigonometric structure but focus on different variables ($x$ vs. $\theta_j$).
Below, we discuss several promising future works to conclude this paper.

\paragraph{Noise stability and robustness.}

Understanding the impact of noise on the stability of ICD is an important direction. Since noise is inherent in quantum systems, it is crucial to analyze how it affects the solution of the approximated subproblem. As discussed in \cref{app-eig-method}, the eigenvalue-based method for solving subproblems with equidistant frequencies allows us to establish an exact map relation between estimated coefficients and the final solution. This suggests that a rigorous stability analysis is feasible, albeit challenging, and could provide deeper insights into the resilience of ICD under realistic quantum conditions.

\paragraph{Sparsity in the Fourier series.}

Our experiments reveal that the Fourier series representation of $f(\theta_j)$ often exhibits sparsity, with the scale of $r_j$ depending on the ansatz, Hamiltonian, and boundary conditions. For instance, in the case of the TFIM model with an HVA ansatz, as discussed in \cref{sec-tfim-xxz}, we consistently observe $r_j = 1$ for all coordinates $j$, regardless of the number of layers or qubits. However, in general, this sparsity structure may require a case-by-case analysis. Identifying and exploiting such sparsity could significantly reduce the number of samples required for interpolation, thereby enhancing the efficiency of ICD. In \cref{app-sparse-frequency}, we briefly analyze a special case of the 3-qubit TFIM. For more general models and quantum circuits, developing systematic analytical methods remain directions for future research.

\paragraph{Convergence theory of ICD.}

The theoretical convergence of ICD remain an open question. 
Unlike standard coordinate descent (CD) algorithms, ICD does not have a direct counterpart in classical optimization literature, because existing convergence results for CD algorithms typically assume noise-free cost functions. 
While some CD algorithms update each step using argmin, they do not account for noise in the their problems. 
ICD may introduce a new optimization formulation, requiring a tailored theoretical analysis to establish its convergence properties. Investigating this aspect is an important direction for future research.

On the other hand, we would like to clarify that achieving a global minimum is generally intractable for nonconvex optimization problems due to their intrinsic complexity. 
The motivation behind using a CD scheme (where a single parameter is updated per iteration via a one-dimensional minimization) is to enhance convergence speed compared to RCD, particularly by leveraging the closed-form solution of the one-dimensional subproblem. 
While global optimality is not guaranteed in general, under certain conditions, such as when the initial point is sufficiently close to a global minimizer and the cost function satisfies the Kurdyka-\L{}ojasiewicz (K\L{}) property, convergence to a global minimum can be theoretically achievable. 
Please see \cite[Corollary 2.7]{xu2013block} for a related discussion. A rigorous analysis of this direction is beyond the current scope and will be considered in future work.
 
\paragraph{Local minimum of cost function.}

The PQCs cost function is inherently non-convex, and the landscape is often highly rugged, filled with numerous spurious local minima, which makes classical gradient-based methods prone to getting trapped in local minimum.
\cite{bittel2021training} proves that the classical optimization of PQCs is NP-hard.
On the other hand, ICD employs a randomized coordinate selection along with an argmin update scheme, which offers the potential to escape local minima.
However, in our numerical experiments, we occasionally observe that ICD can also become stuck in local minima.
As noted in \cite{shi2016primer}, CD with the argmin update scheme may fail to converge for certain cost functions.
Specifically, \cite{shi2016primer} presents a convex but non-smooth function (originally proposed in \cite{warga1963minimizing}),
$f(x, y) = |x - y| - \min(x, y),$
for which CD with the argmin update gets stuck at a even non-stationary point (see in \cite[Figure 2]{shi2016primer}).


\begin{acknowledgments}
This work was supported by the National Natural Science Foundation of China under the grant numbers 12331010 and 12288101, and the National Key R\&D Program of China under the grant number 2024YFA1012903.
DA acknowledges the support by The Fundamental Research Funds for the Central Universities, Peking University.
Part of this work was completed while JH was affiliated with UC Berkeley.
We would like to express our sincere gratitude to Liyuan Cao, Zhiyan Ding, Tianyou Li, Xiantao Li, Xiufan Li, Lin Lin, Yin Liu, and Zaiwen Wen for their valuable feedback and insightful comments on the manuscript. 
\end{acknowledgments}

\section*{Author contributions}
Z.L. conceived the idea and carried out the theoretical analysis. Z.L., T.K. and J.W. performed the numerical simulations. Z.L., J.H. and D.A. analyzed the results. All authors contributed to the preparation of the manuscript. 

\section*{Competing interests}
The authors declare no competing interests. 

\section*{Data availability}
The data and codes used in this study are available from public repositories, as described in \cref{sec-experiments}.

\newpage

\begin{figure}
    \centering
    \includegraphics[width=0.42\linewidth]{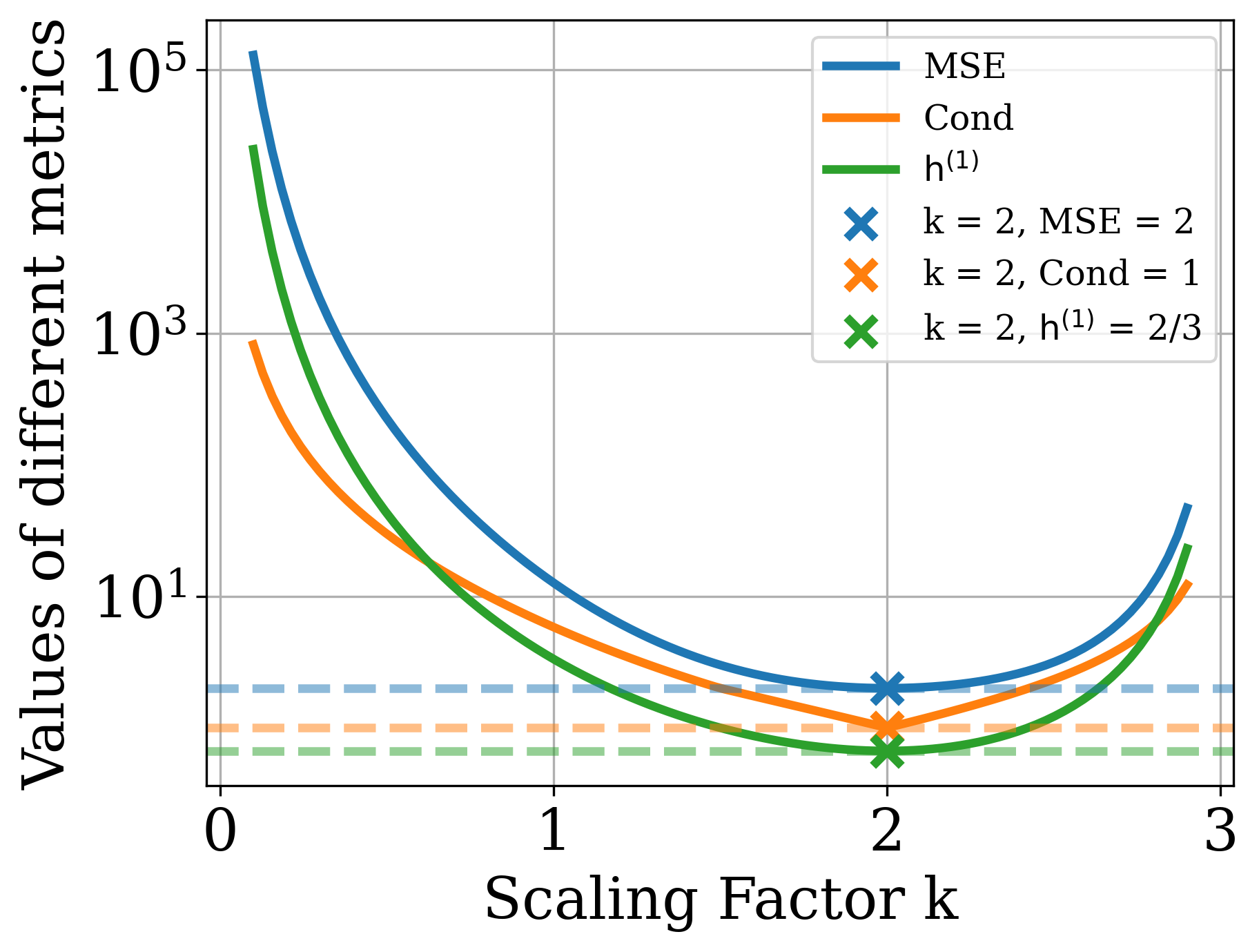}
    \caption{Comparison of the three optimality criteria --- mean squared error (MSE), condition number (Cond) and average variance of derivatives ($h^{(1)}$) (i.e., \cref{thm-first-veiw,thm-second-veiw,thm-third-veiw}, respectively) --- as functions of $k \in (0,3)$, where the interpolation nodes spaced by $k\pi/3$.}
    \label{fig:metrics_with_diff_iterp_nodes}
\end{figure}

\begin{figure}[htbp]
  \begin{subfigure}{0.45\textwidth}
    \centering
    \includegraphics[width=0.9\linewidth]{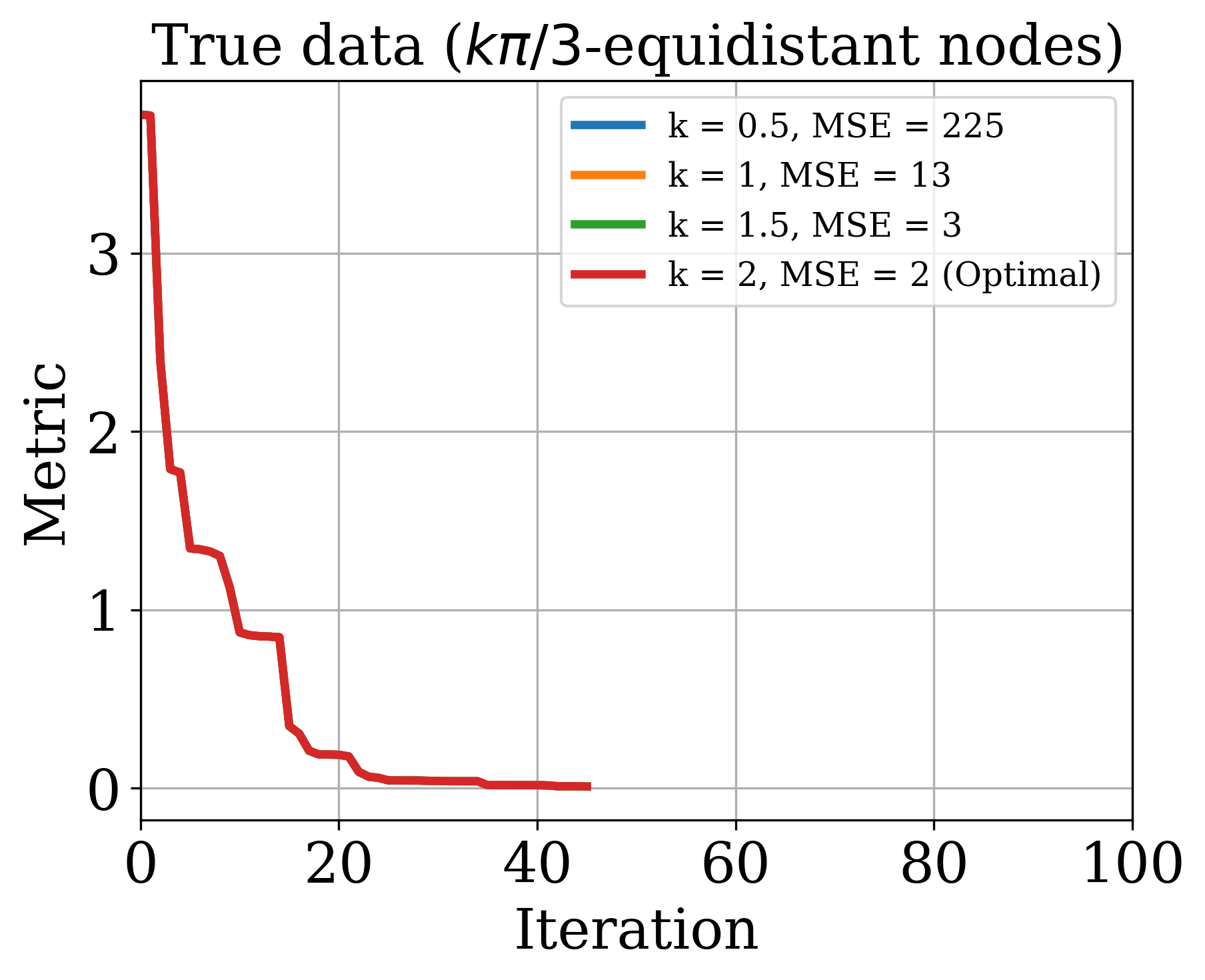} 
    \caption{}
    \label{fig:True-Data-Fixed-k-Nodes}
  \end{subfigure}
  \begin{subfigure}{0.45\textwidth}
    \includegraphics[width=0.9\linewidth]{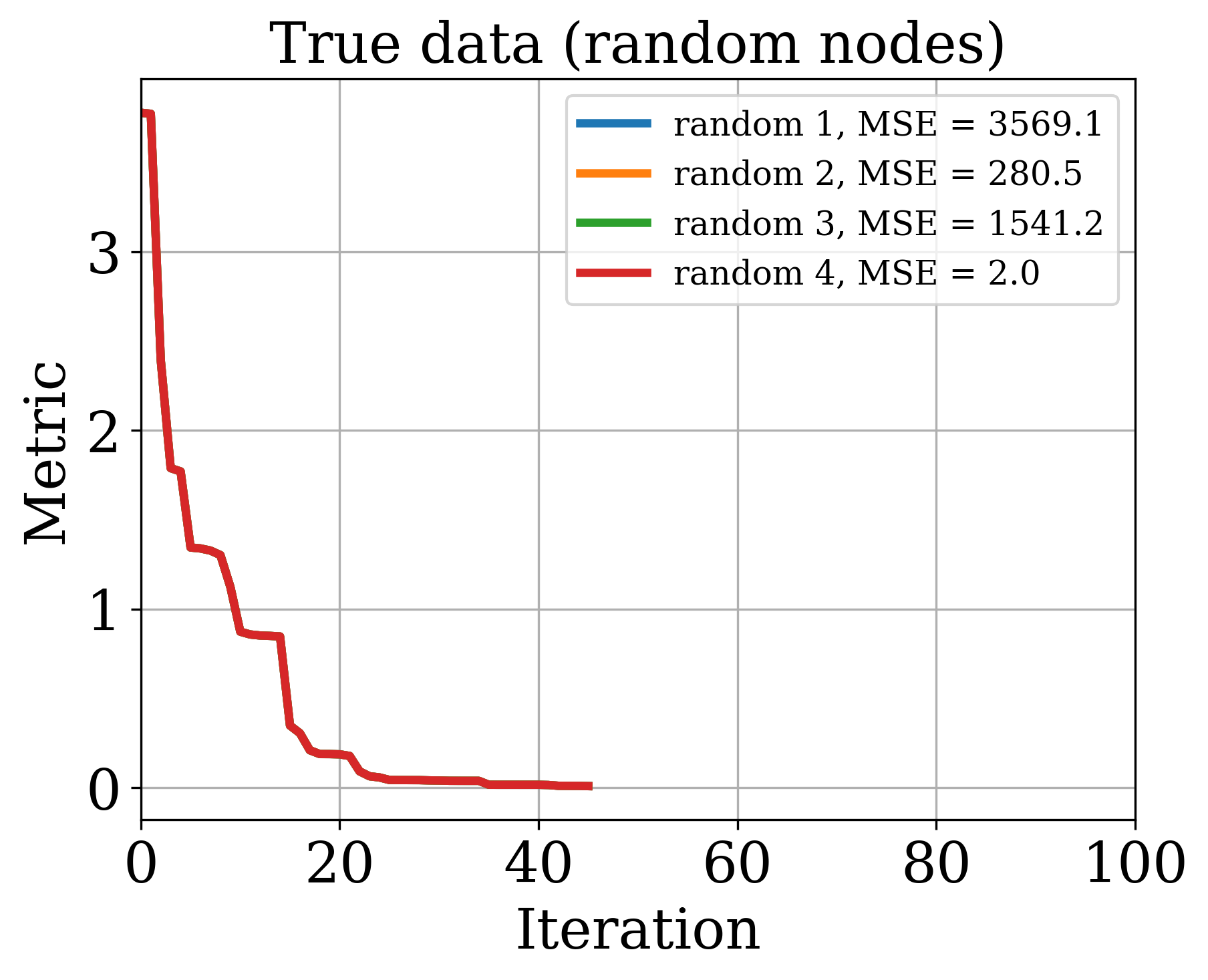}
    \caption{}
    \label{fig:True-Data-Random-Nodes}
  \end{subfigure}
  \caption{
    We test the 4-qubit MaxCut problem in \cref{sec-pro-maxcut}. The figures present the convergence of the ICD algorithm in the noiseless setting (infinite shots).
    (a) Using $k\pi/3$ equidistant interpolation nodes ($k = 0.5, 1, 1.5, 2$);
    (b) Using randomly selected interpolation nodes.
    All configurations exhibit identical convergence trajectories, confirming that in the absence of sampling noise, the choice of interpolation nodes has no impact on the optimization performance.}
\end{figure}

\begin{figure}[htbp]
  \begin{subfigure}{0.45\textwidth}
    \centering
    \includegraphics[width=0.9\linewidth]{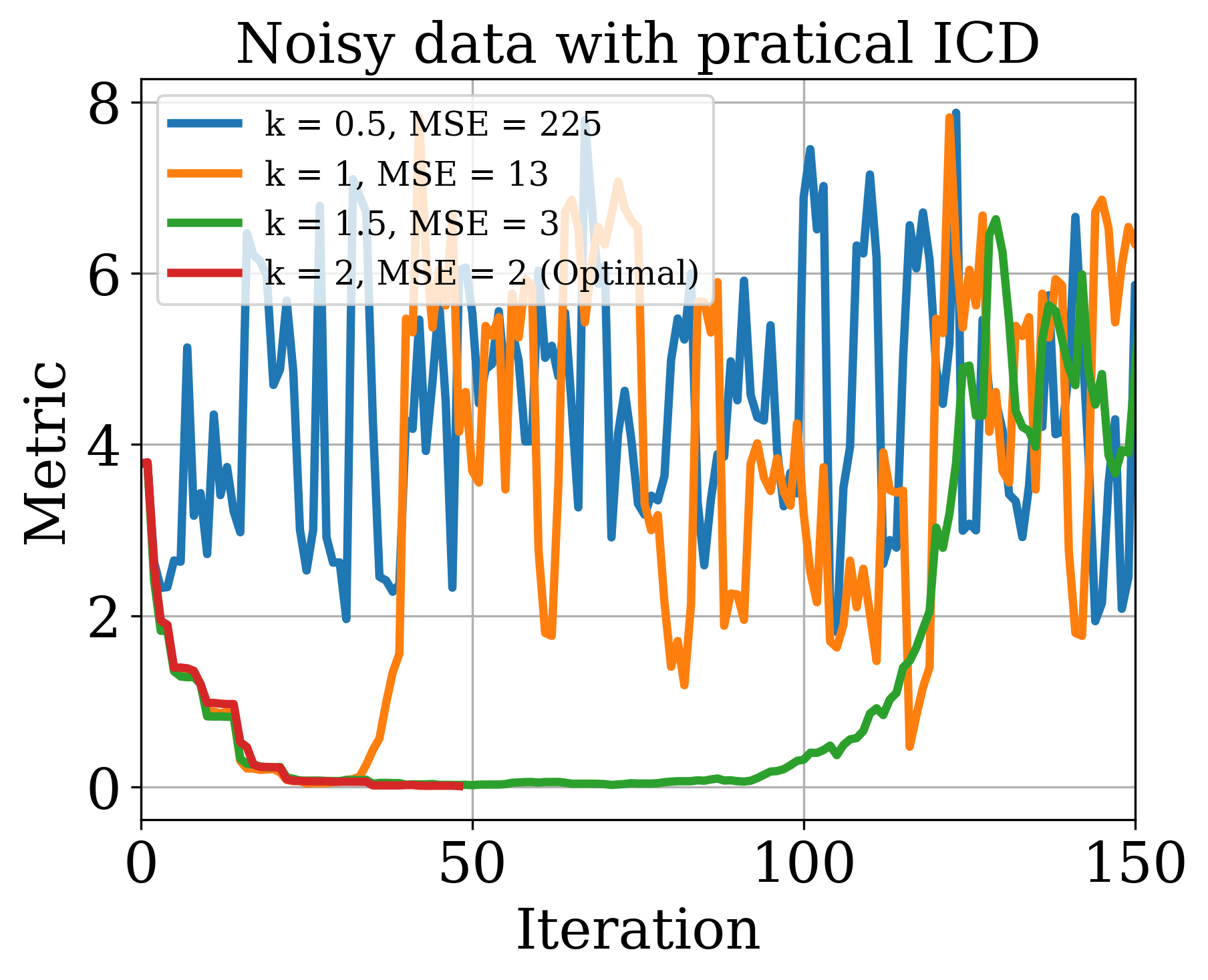} 
    \caption{}
    \label{fig:Noisy-Data-Practical-ICD}
  \end{subfigure}
  \begin{subfigure}{0.45\textwidth}
    \includegraphics[width=0.9\linewidth]{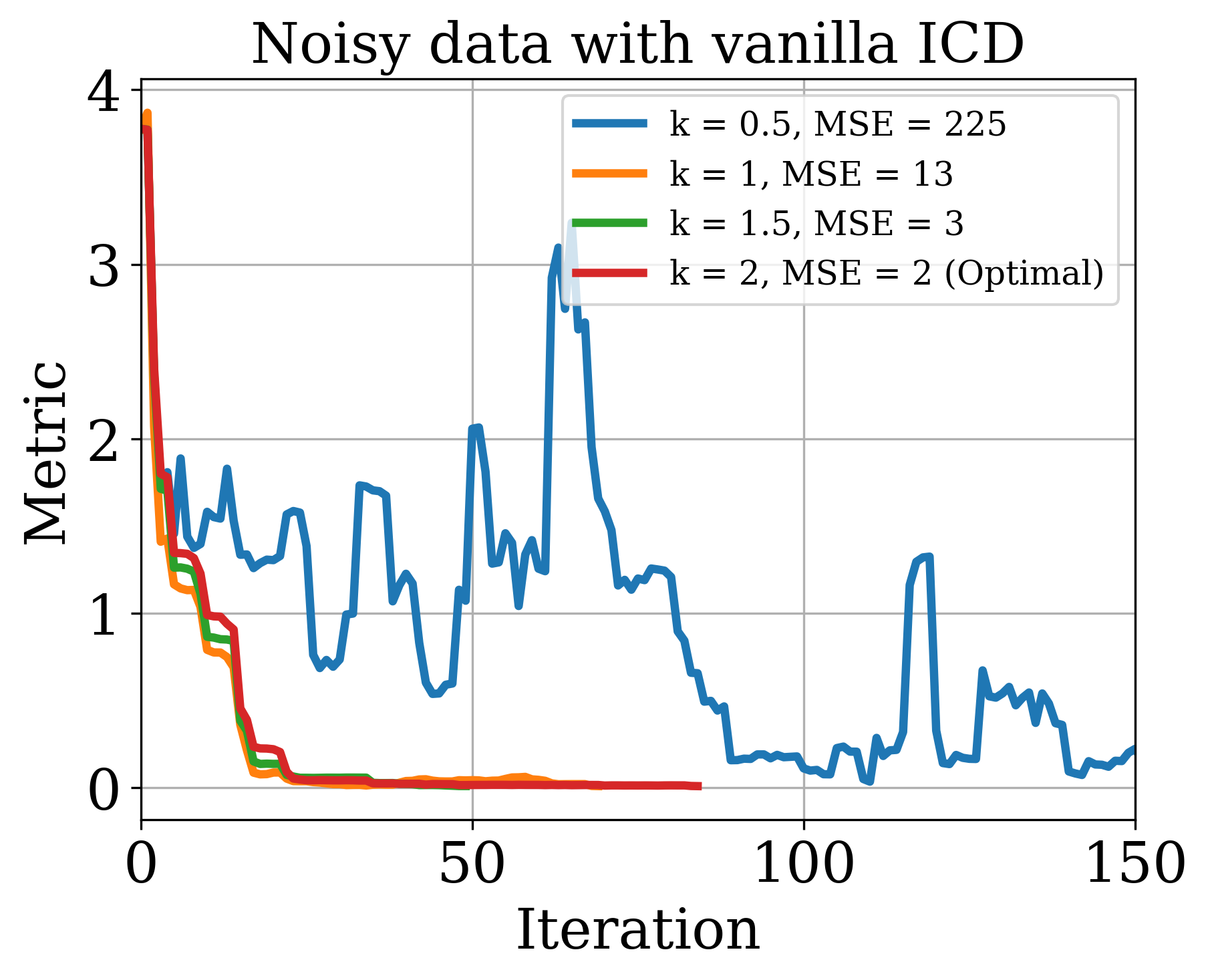}
    \caption{}
    \label{fig:Noisy-Data-Vanilla-ICD}
  \end{subfigure}
  \caption{We test the 4-qubit MaxCut problem in \cref{sec-pro-maxcut}. The figures present the convergence behavior of (a) the reduced ICD algorithm and (b) the standard ICD algorithm under noisy settings (1024 shots), using $k\pi/3$ equidistant interpolation nodes ($k = 0.5, 1, 1.5, 2$).}
\end{figure}

\begin{figure}
    \centering
    \includegraphics[width=1\linewidth]{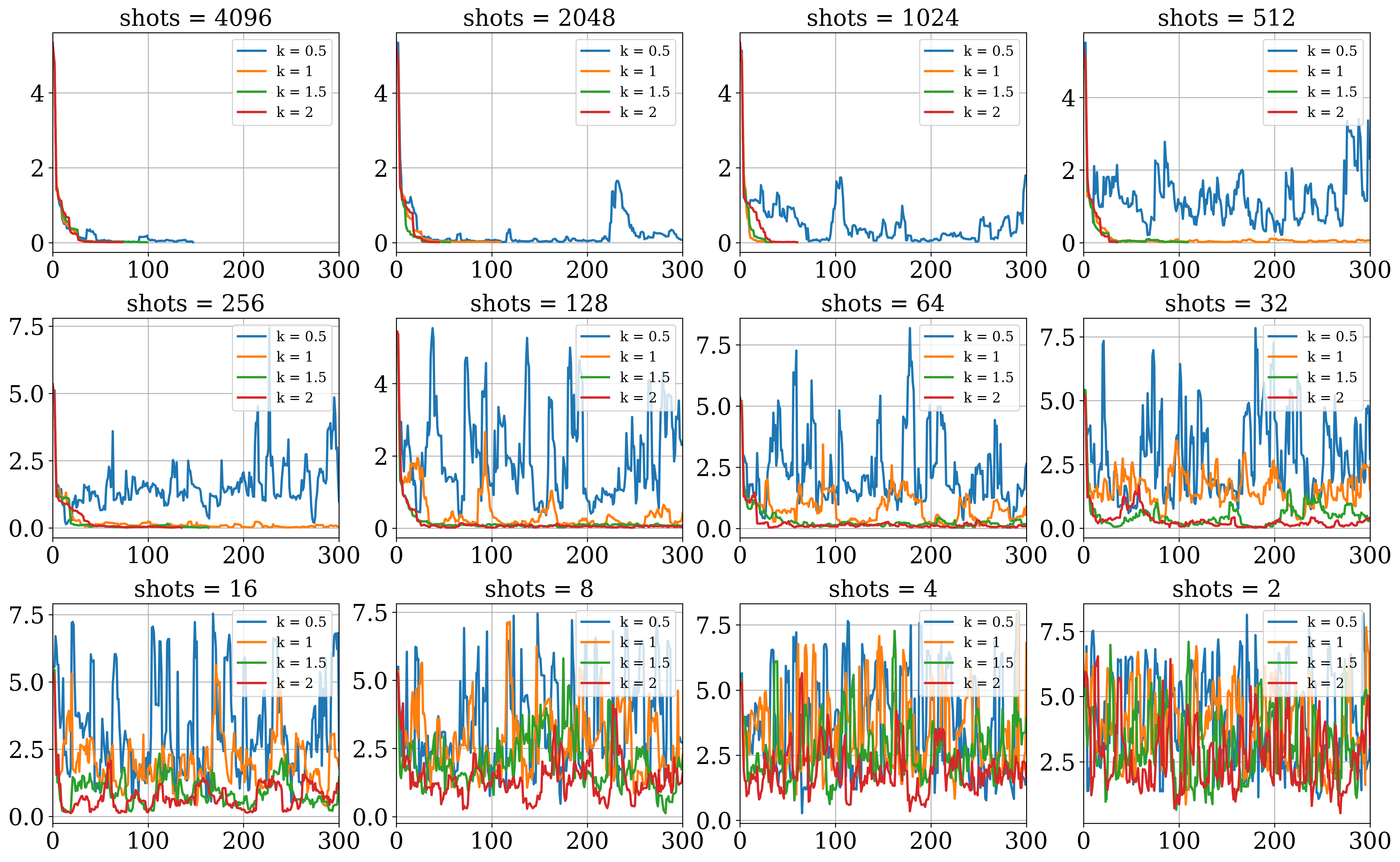}
    \caption{We test the 4-qubit MaxCut problem in \cref{sec-pro-maxcut}. The figures present the convergence behavior of standard ICD under varying shot counts for $k\pi/3$ equidistant interpolation nodes ($k = 0.5, 1, 1.5, 2$). Each curve corresponds to one interpolation node configuration. As the number of shots decreases from 4096 to 128, configurations with larger $k$ (closer to optimal spacing) maintain better convergence properties, while smaller $k$ values degrade more rapidly.}
    \label{fig:maxcut_oicd_all_nshot_subplots_1}
\end{figure}

\begin{figure}
    \centering
    \includegraphics[width=1\linewidth]{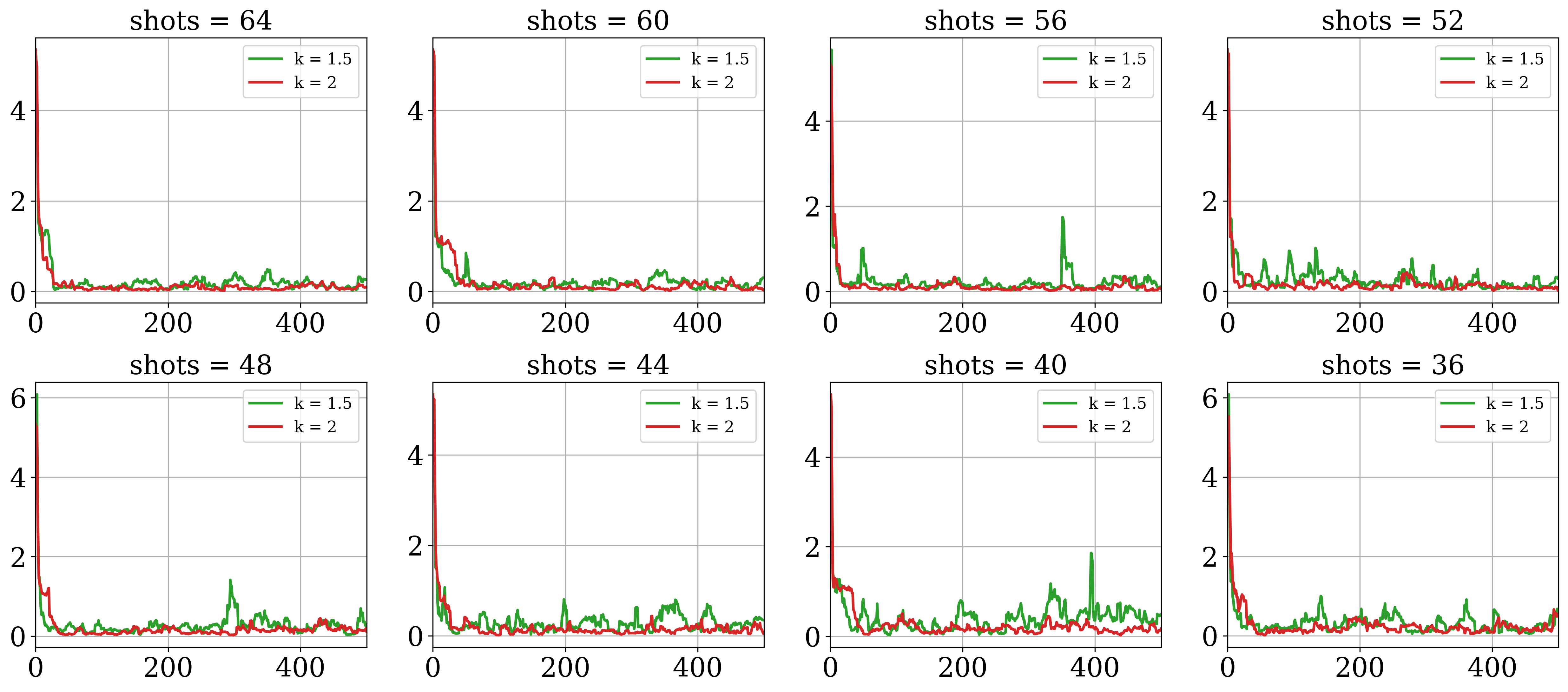}
    \caption{
    We test the 4-qubit MaxCut problem in \cref{sec-pro-maxcut}. The figures present the detailed comparison of convergence behaviors for $k\pi/3$ equidistant interpolation nodes with $\frac{\pi}{2}$ ($k = 1.5$) and $\frac{2\pi}{3}$ ($k = 2$)  under low shot counts (from 64 to 36). While both settings perform similarly overall, the $k = 1.5$ configuration shows greater sensitivity to noise.}
    \label{fig:maxcut_oicd_all_nshot_subplots_2}
\end{figure}

\begin{figure}[htbp]
    \centering
    \begin{subfigure}{0.3\textwidth}
        \includegraphics[width=\linewidth]{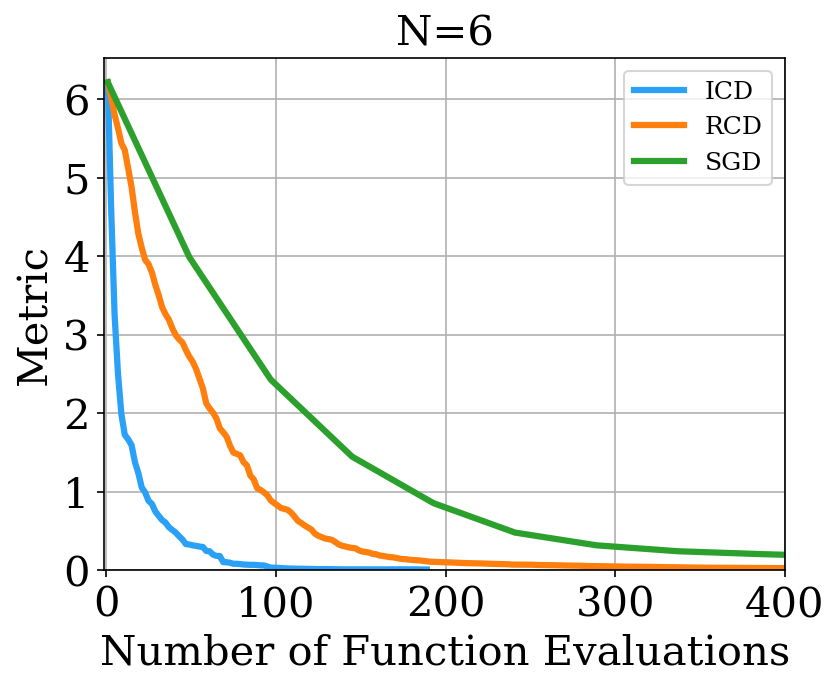}
    \end{subfigure}
    \begin{subfigure}{0.3\textwidth}
        \includegraphics[width=\linewidth]{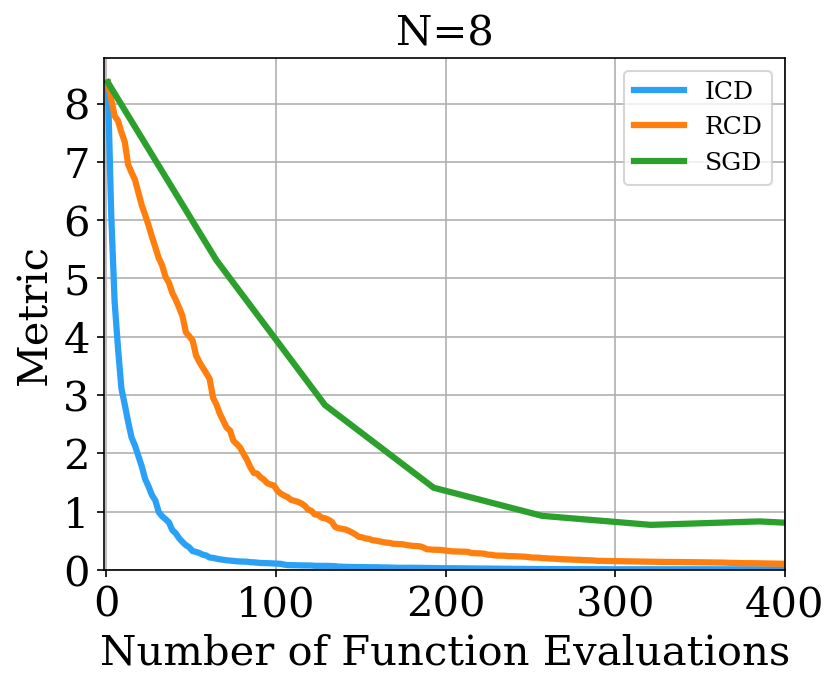}
    \end{subfigure}
    \begin{subfigure}{0.3\textwidth}
        \includegraphics[width=\linewidth]{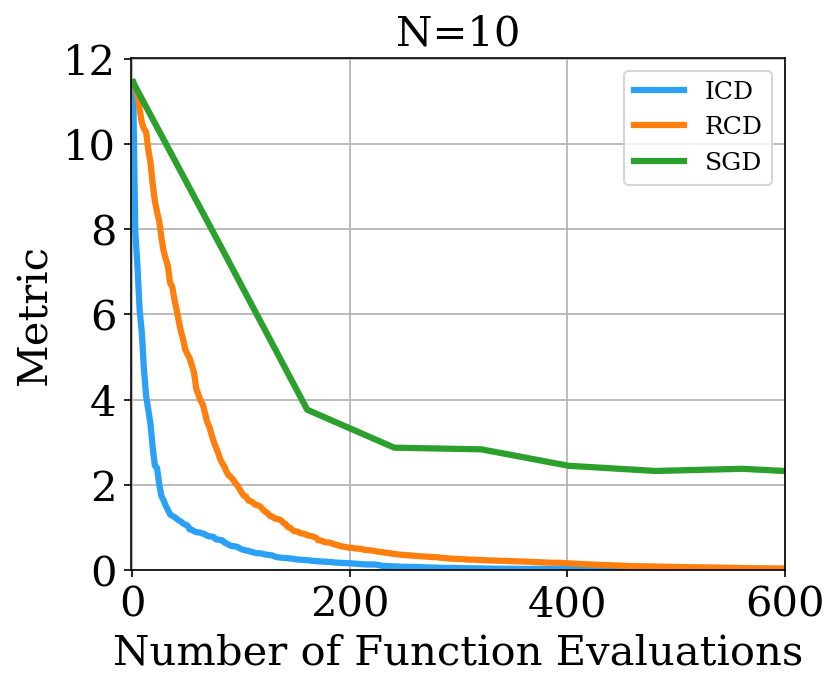}
    \end{subfigure}

    \vspace{0.5em}

    \begin{subfigure}{0.3\textwidth}
        \includegraphics[width=\linewidth]{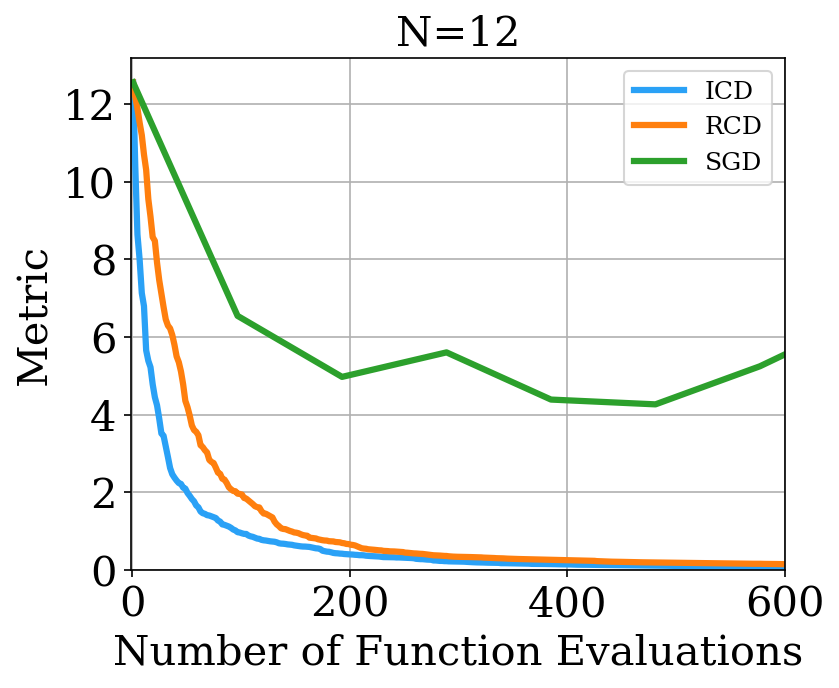}
    \end{subfigure}
    \begin{subfigure}{0.3\textwidth}
        \includegraphics[width=\linewidth]{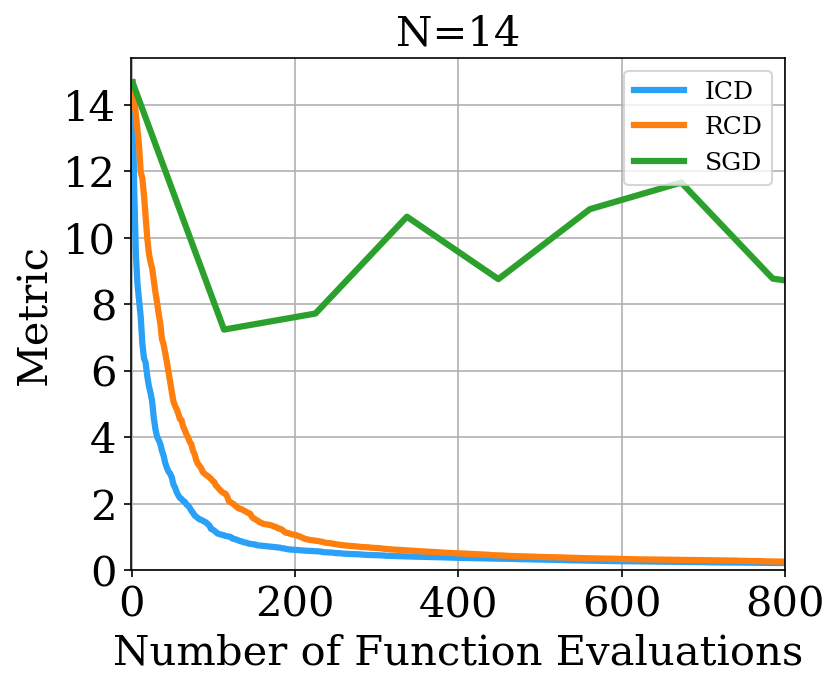}
    \end{subfigure}
    \begin{subfigure}{0.3\textwidth}
        \includegraphics[width=\linewidth]{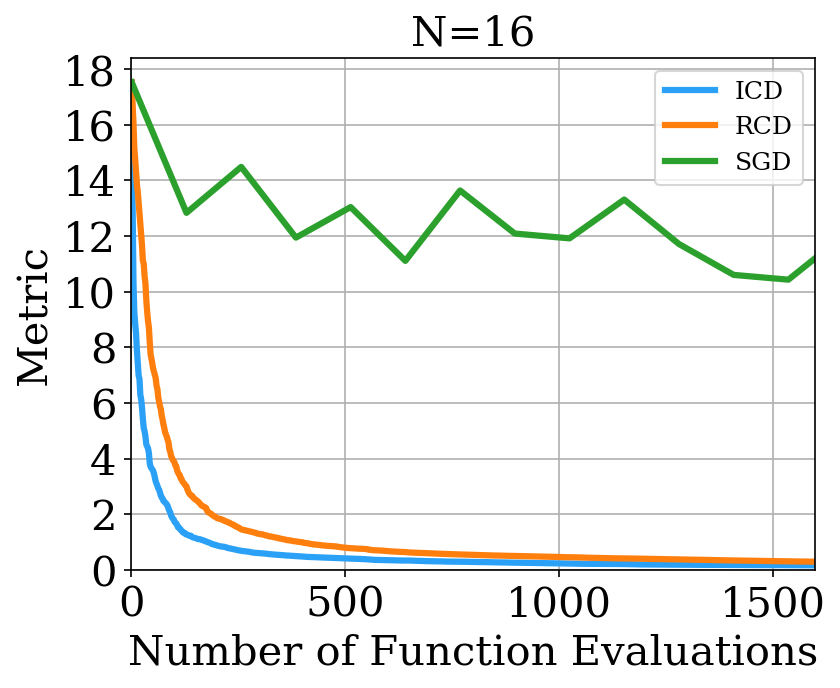}
    \end{subfigure}

    \caption{
    We test the $N$-qubit TFIM models using a $2N$-layer HVA ansatz, as described in \cref{sec-pro-tfim}, for $N=6,8,10,12,14,16$. The figures present a performance comparison among ICD, RCD, and SGD. ICD consistently outperforms RCD and SGD across various system sizes, showing faster convergence without encountering the barren plateau phenomenon.
    }
    \label{fig:tfim_hva_many}
\end{figure}

\begin{figure}[htbp]
    \centering
    \begin{subfigure}{0.3\textwidth}
        \includegraphics[width=\linewidth]{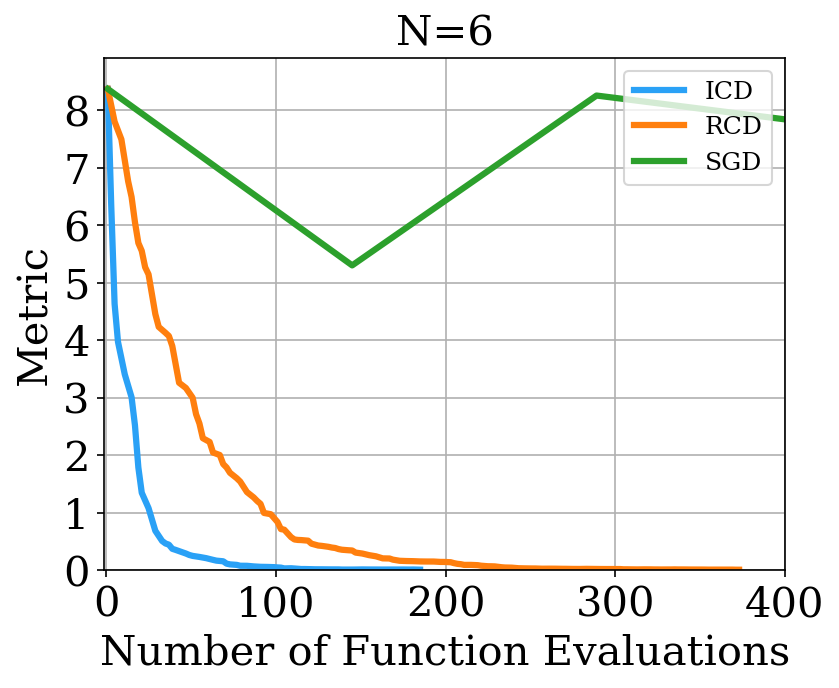}
    \end{subfigure}
    \begin{subfigure}{0.3\textwidth}
        \includegraphics[width=\linewidth]{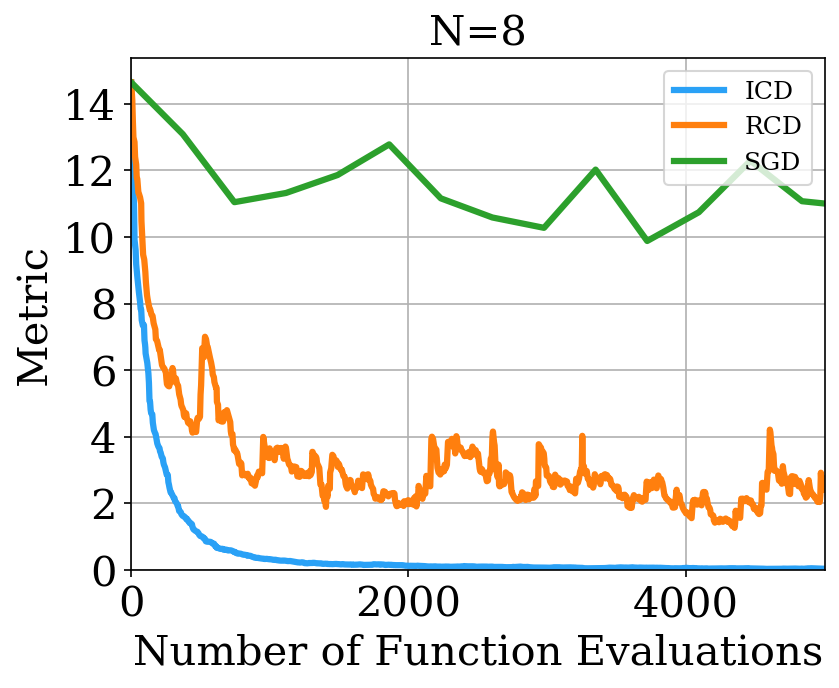}
    \end{subfigure}
    \begin{subfigure}{0.3\textwidth}
        \includegraphics[width=0.95\linewidth]{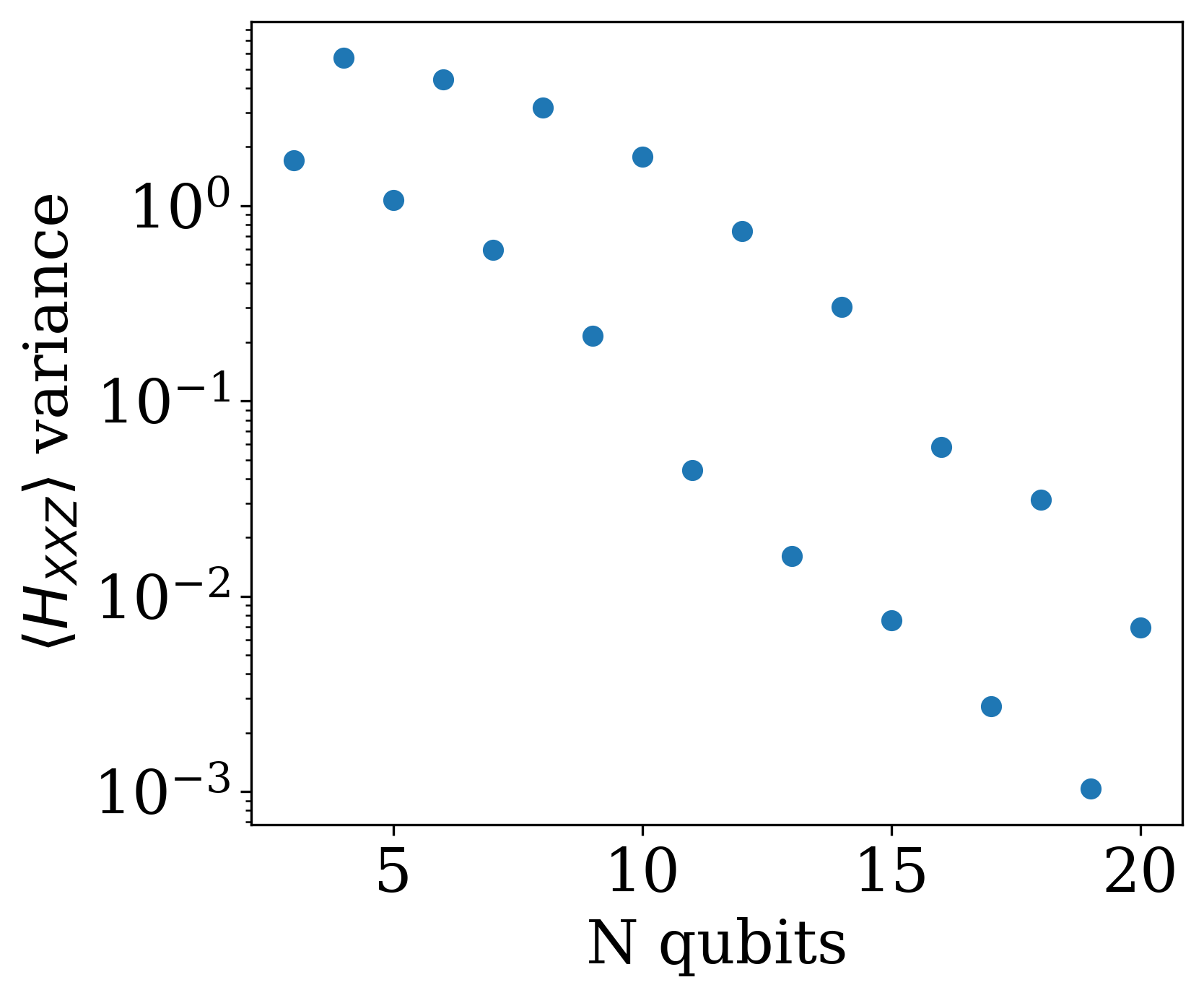}
    \end{subfigure}

    \vspace{0.5em}

    \begin{subfigure}{0.3\textwidth}
        \includegraphics[width=\linewidth]{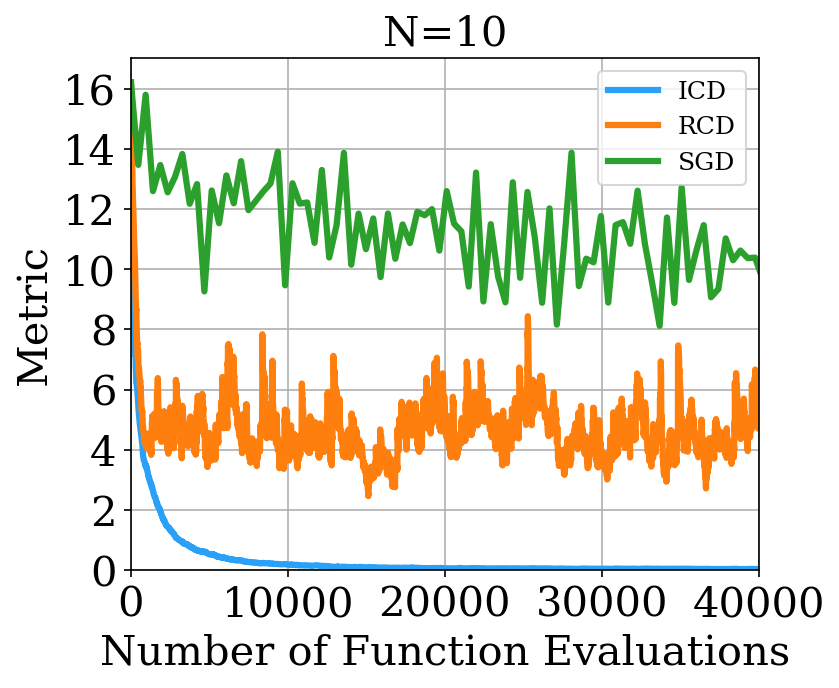}
    \end{subfigure}
    \begin{subfigure}{0.3\textwidth}
        \includegraphics[width=\linewidth]{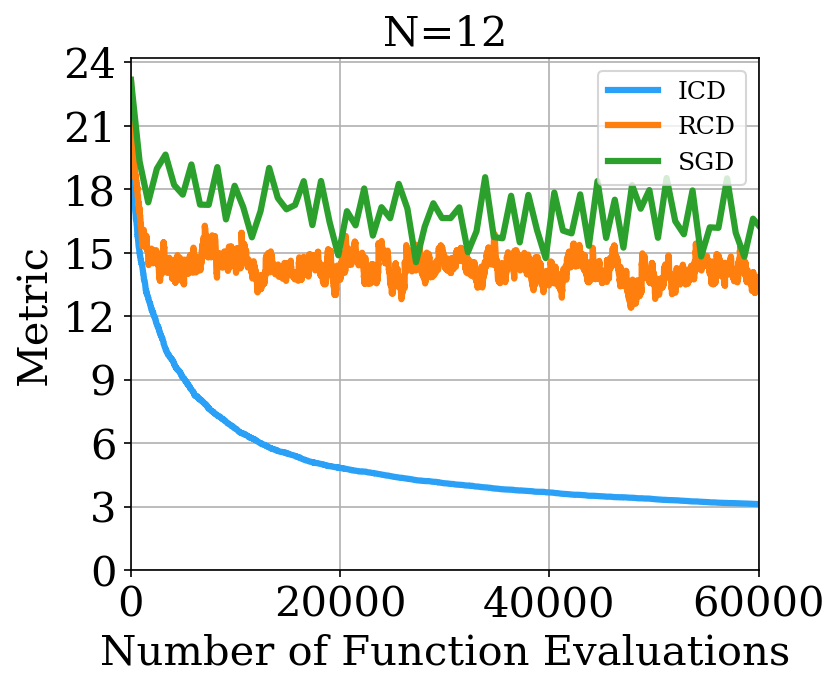}
    \end{subfigure}
    \begin{subfigure}{0.3\textwidth}
        \includegraphics[width=0.94\linewidth]{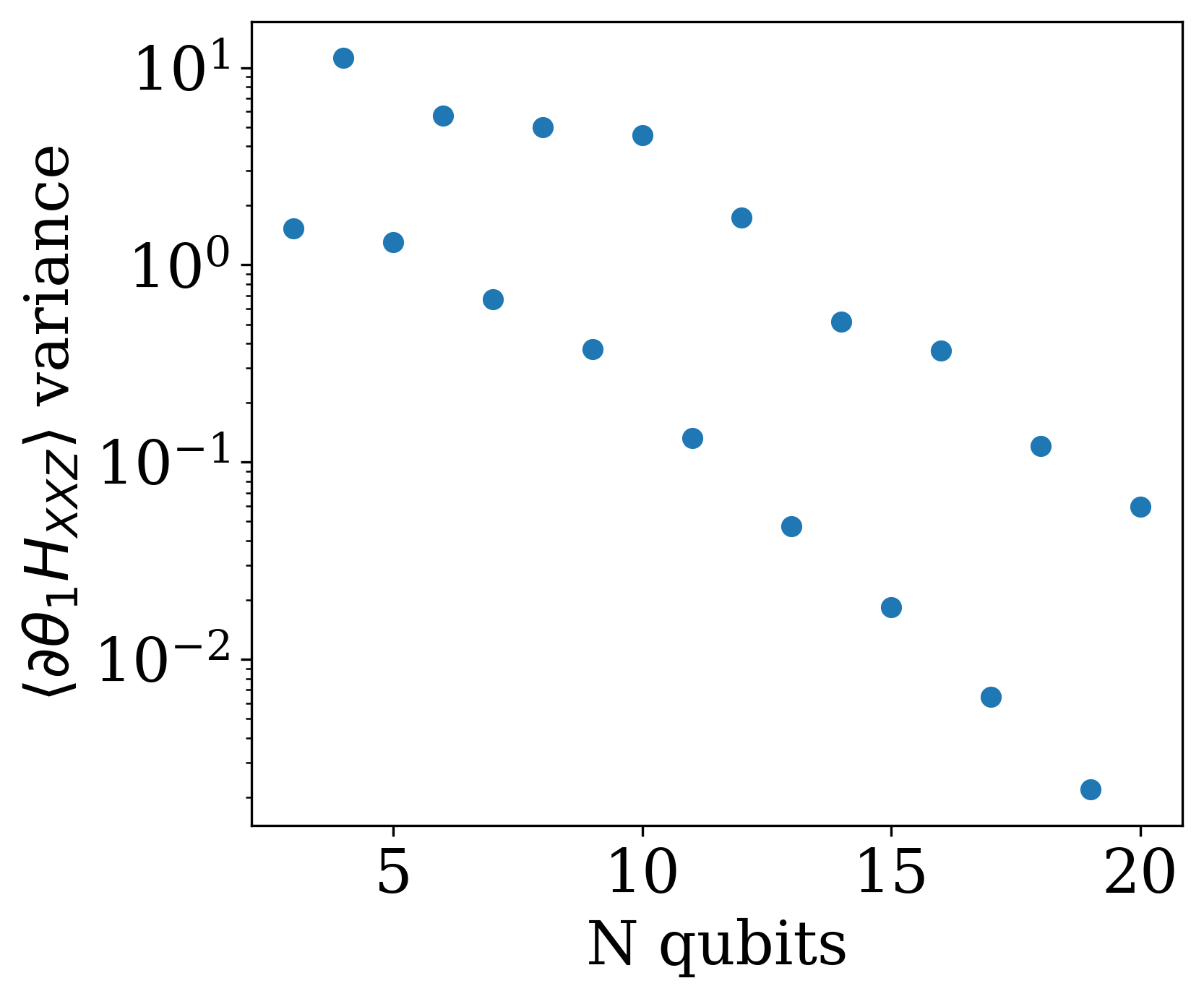}
    \end{subfigure}

    \caption{
    We test the $N$-qubit XXZ models using a $2N$-layer HVA ansatz, as described in \cref{sec-pro-xxz}, for $N=6,8,10,12$. 
The figures present a performance comparison among ICD, RCD, and SGD. The barren plateau phenomenon is evident, as indicated by the exponential decay in cost and gradient variances. ICD fails to converge as qubit number increases.
    }
    \label{fig:xxz_hva_many}
\end{figure}

\newpage
\section*{References}

\bibliography{mybib}

\appendix

\section{Trigonometric representation of PQC cost function}\label{app-trig}

The derivation of \cref{eq-trig-thetaj} comes from \cite{wierichs2022general}. 
For the sake of completeness, we provide the detailed derivation process in this appendix. 
Let $N=2^q$ for the $q$-qubit system.
We begin by formulating a PQC cost function that depends on a single parameter $x \in \mathbb{R}$.
Consider the unitary operator defined as $U(x) = \exp(i H x),$ where $H$ is a Hermitian generator. 
Let $|\psi\rangle$ denote the quantum state to which $U(x)$ is applied, and let $O$ represent the observable being measured. 
The PQC cost function is then defined by
\begin{equation}
    f(x) := \langle \psi | U(x)^{\dagger} O U(x) | \psi \rangle.
\end{equation}
Let $\{ \lambda_j \}_{j \in [N]}$ represent the eigenvalues of $H$, where $[N] := \{1, \ldots, N\}$ and the eigenvalues are arranged in non-decreasing order ($\lambda_1 \leq \dots \leq \lambda_N$). 
For any real number $x$, the set $\{ \exp(i \lambda_j x) \}_{j \in [N]}$ constitutes the eigenvalues of $U(x)$. 
Specifically, let $\{ |\phi_j\rangle \}_{j \in [N]}$ be the eigenbasis of $H$. Then, $H$ can be diagonalized as
\begin{equation}
    H = \sum_{j=1}^N \lambda_j |\phi_j\rangle \langle \phi_j |,
\end{equation}
and $U(x) = \exp(i H x)$ can also be diagonalized in the same eigenbasis as
\begin{equation}\label{eq-eigenbasis}
    U(x) = \sum_{j=1}^N \exp(i x \lambda_j) |\phi_j\rangle \langle \phi_j |. 
\end{equation}
We proceed by expanding $O$ and $|\psi\rangle$ in the eigenbasis  $\{ |\phi_j\rangle \}_{j \in [N]}$. Specifically, we define the matrix entries of $O$ and the coefficients of $|\psi\rangle$ in this eigenbasis by 
\begin{equation}
    [O]_{jk} := \langle \phi_j | O | \phi_k \rangle, \quad [\psi]_j := \langle \phi_j | \psi \rangle, \quad \forall j, k \in[N].
\end{equation}
By \cref{eq-eigenbasis}, we apply $U(x)$ to the state $|\psi\rangle$ to obtain
\begin{equation}
    U(x)|\psi\rangle = \sum_{j=1}^N \exp(i x \lambda_j) |\phi_j\rangle \langle \phi_j | \psi \rangle = \sum_{j=1}^N [\psi]_j \exp(i x \lambda_j) |\phi_j\rangle.
\end{equation}
Substituting the expanded form of $U(x)|\psi\rangle$ above into $f(x)$, we have
\begin{align}
f(x) =\langle \psi | U(x)^{\dagger} O U(x) | \psi \rangle
&= \left( \sum_{j=1}^N \overline{[\psi]_j} \exp(-i x \lambda_j) \langle \phi_j | \right)O \left( \sum_{k=1}^N [\psi]_k \exp(i x \lambda_k) |\phi_k\rangle \right) \\
&= \sum_{j,k=1}^N \overline{[\psi]_j} [\psi]_k \exp\left[ i (\lambda_k - \lambda_j) x \right] \langle \phi_j | O | \phi_k \rangle \\
&= \sum_{j,k=1}^N \overline{[\psi]_j} [\psi]_k \;[O]_{jk} \;e^{i (\lambda_k - \lambda_j) x}.
\end{align}
To further simplify $f(x)$, we separate the terms where $j \neq k$ and $j = k$ and obtain
\begin{equation}
    f(x) = \sum_{\substack{j, k=1 \\ j < k}}^N \left[ \overline{[\psi]_j} [\psi]_k b_{jk} e^{i (\lambda_k - \lambda_j) x} + [\psi]_j \overline{[\psi]_k [O]_{jk}} e^{-i (\lambda_k - \lambda_j) x} \right] + \sum_{j=1}^N |[\psi]_j|^2 [O]_{jj}.
\end{equation}
We can collect the $x$-independent terms into coefficients defined by $c_{jk} := \overline{[\psi]_j} [\psi]_k [O]_{jk}.$ Then, it becomes
\begin{equation}\label{eq-931}
f(x)
= \sum_{\substack{j, k=1 \\ j < k}}^N \left[ c_{jk} e^{i(\lambda_k - \lambda_j)x} + \overline{c_{jk}} e^{-i(\lambda_k - \lambda_j)x} \right] + \frac{1}{\sqrt{2}} a_0,
\end{equation}
where we let $a_0 := \sqrt{2} \sum_{j=1}^N |[\psi]_j|^2 [O]_{jj} \in \mathbb{R}.$ Notice that $[O]_{jj}$ must be real since $O$ is Hermitian.
Moreover, we introduce the unique positive differences (called frequencies)
\begin{equation}\label{eq-509}
    \{\Omega_\ell\}_{\ell \in [r]} := \{\lambda_k - \lambda_j \mid \forall j, k \in [N], \ \lambda_k > \lambda_j\}.
\end{equation}
Here, $r$ is the number of the unique positive differences.
For the same term $e^{i(\lambda_k - \lambda_j) x}=e^{i \Omega_{\ell} x}$, we sum the coefficients $c_{jk}$ in front of them and re-index it as $c_{\ell}$.
Consequently, \cref{eq-931} becomes
\begin{equation}\label{eq-1420}
f(x)
= \sum_{\ell=1}^r c_\ell e^{i \Omega_\ell x} + \sum_{\ell=1}^r \overline{c_\ell} e^{-i \Omega_\ell x} + \frac{1}{\sqrt{2}} a_0.
\end{equation}
This is precisely the complex form of a finite-term Fourier series.
Next, we parameterize the complex coefficients $c_\ell$ by real numbers $a_\ell$ and $b_\ell$ as
\begin{equation}
    c_\ell = \frac{1}{2}(a_\ell - i b_\ell), \quad \forall \ell \in [r].
\end{equation}
Utilizing the trigonometric identities $\cos(x) = \frac{1}{2}(e^{ix} + e^{-ix}), \sin(x) = \frac{i}{2} (-e^{ix} + e^{-ix}),$ we can rewrite $f(x)$ as the real form of a finite-term Fourier series:
\begin{align}
f(x)
&= \sum_{\ell=1}^r \frac{1}{2}(a_\ell - i b_\ell) e^{i \Omega_\ell x} + \sum_{\ell=1}^r \frac{1}{2}(a_\ell + i b_\ell) e^{-i \Omega_\ell x} + \frac{1}{\sqrt{2}} a_0 \\
&= \sum_{\ell=1}^r a_\ell \left( \frac{e^{i \Omega_\ell x} + e^{-i \Omega_\ell x}}{2} \right) + \sum_{\ell=1}^r i b_\ell \left( \frac{-e^{i \Omega_\ell x} + e^{-i \Omega_\ell x}}{2} \right) + \frac{1}{\sqrt{2}} a_0 \\
&= \frac{1}{\sqrt{2}} a_0 + \sum_{\ell=1}^r \left[  a_\ell \cos(\Omega_\ell x) + b_\ell \sin(\Omega_\ell x)\right], \label{eq-1421}
\end{align}
where the frequencies are given by \cref{eq-509}. We complete the proof for \cref{eq-trig-thetaj}.

\begin{remark}[Actual frequency is more important]\label{re-sparse-freq}
From the above construction, it can be observed that if the coefficient $c_{jk}$ equals zero in \cref{eq-931}, then the corresponding $\lambda_k - \lambda_j$ can be discarded. Therefore, the actual frequencies should consist only of the terms where both $[O]_{jk}$ and $[\psi]_k$ are non-zero. In other words, the frequency set defined in \cref{eq-509} (which depends solely on $H$) provides an upper bound for the actual frequency set. In the worst-case scenario, the number of frequencies is $O (4^q)$ for q-qubit system.
However, in practical cases, such as HEA circuits where $H$ consists of Pauli strings, the frequencies $\Omega$ are singletons, resulting in $O(1)$ frequencies. 
Additionally, for HVA circuits, the generator $H$ is highly correlated with the observable $O$, leading to a large number of zero coefficients $[O]_{jk}$. As a result, the actual frequency is further reduced. \cref{app-sparse-frequency} provides a detailed analysis for the TFIM model with HVA, proving that their frequency count is also $O (1)$. 
\end{remark}

\section{Recover the complex coefficients of restricted univariate functions}\label{app-recover-complex}

In \cref{sec-interp-real}, we discussed the interpolation method to recover $n = 2r + 1$ real coefficients in the restricted univariate function $f(x)$ in \cref{eq-trig-poly-real}. In this appendix, we discuss an interpolation method for recovering complex coefficients of $f(x)$. It will not appear in our actual ICD algorithms, but it is very useful for explaining the theoretical phenomena behind the ICD.
For the moment, assume that we can access to $ f(x)$ without noise. 

As we saw in \cref{eq-1420,eq-1421} in \cref{app-trig}, there are two equivalent expressions for $f(x)$, since we observe that
\begin{align}
    f(x)
     &= \frac{1}{\sqrt{2}} a_0 + \sum_{k=1}^r \left[ a_k \cos(\Omega_k x) + b_k \sin(\Omega_k x) \right]  \label{eq-trig-poly-real000}\\
    &= \sum_{k=-r}^r c_k e^{i\Omega_k x}, \label{eq-trig-poly-complex}
\end{align}
where we define $\Omega_{-k} := -\Omega_{k}$ for $k=1, \cdots, r$ and $\Omega_0 := 0$; and we convert the real coefficients $\mathbf{z} = \left[ a_0, a_1, b_1, \cdots, a_r, b_r \right]^{\dagger} \in \mathbb{R}^n$ to the complex coefficients 
\begin{equation}
    \mathbf{z}^{c} := \left[ c_{-r}, \cdots, c_0, \cdots, c_r \right]^{\dagger} \in \mathbb{C}^n
\end{equation}
by using the a linear transformation, i.e.,
\begin{equation}
    \begin{cases}
    c_0 := \frac{1}{\sqrt{2}} a_0 \\
    c_k := \frac{a_k}{2} - \frac{b_k}{2}i,   \;\forall k = 1, \dots, r \\
    c_{-k} := \frac{a_k}{2} + \frac{b_k}{2}i = c_k^*,   \;\forall k = 1, \dots, r
    \end{cases}
    \Longleftrightarrow
    \begin{cases}
    a_0 = \sqrt{2} c_0 \\
    a_k = c_k + c_{-k},  \;\forall k = 1, \dots, r \\
    b_k = i\left(c_{k} - c_{-k}\right),   \;\forall k = 1, \dots, r.
    \end{cases}
\end{equation}
The invertibility of this transformation establishes the equivalence between the real-coefficient expression \cref{eq-trig-poly-real000} and its complex-coefficient counterpart \cref{eq-trig-poly-complex}. If we define the constant matrix
\begin{equation}\label{eq-C}
    C := \left[\begin{array}{cccccc}
    & & & & \frac{1}{2} & \frac{i}{2} \\
    & & & \ddots & & \\
    & \frac{1}{2} & \frac{i}{2} & & & \\
    \frac{1}{\sqrt{2}} & & & & & \\
    & \frac{1}{2} & -\frac{i}{2} & & & \\
    & & & \ddots & & \\
    & & & & \frac{1}{2} & -\frac{i}{2}
    \end{array}\right] \in \mathbb{C}^{n \times n},
\end{equation}
we can rewrite this transformation as $\mathbf{z}^{c} = C \mathbf{z}$. It is easily to see that $C C^{\dagger} = C^{\dagger} C = \frac{1}{2} I$.

Next, like interpolation \cref{eq-yAz}, we also attempt to recover the complex coefficients $\mathbf{z}^{c}$ by interpolation. Similarly, after selecting some interpolation nodes $\mathbf{x} = [x_0, x_1, \cdots, x_{2r}]^{\dagger} \in \mathbb{R}^n$ with distinct entries, we construct the complex interpolation matrix
\begin{equation}
    A^{c}_{\mathbf{x}} :=
    \left[\begin{array}{ccccccc}
    \omega_0^{-\Omega_r} & \cdots & \omega_0^{-\Omega_1} & 1 & \omega_0^{\Omega_1} & \cdots & \omega_0^{\Omega_r} \\
    \omega_1^{-\Omega_r} & \cdots & \omega_1^{-\Omega_1} & 1 & \omega_1^{\Omega_1} & \cdots & \omega_1^{\Omega_r} \\
    \vdots & \ddots & \vdots & \vdots & \vdots & \ddots & \vdots \\
    \omega_{2r}^{-\Omega_r} & \cdots & \omega_{2r}^{-\Omega_1} & 1 & \omega_{2r}^{\Omega_1} & \cdots & \omega_{2r}^{\Omega_r}
    \end{array}\right]\in \mathbb{C}^{n \times n},
\end{equation}
where $\omega_k := e^{ix_k}$ for $k = 0, 1, \ldots, 2r$. Then, similarly, we can solve the linear equation
\begin{equation}\label{eq-yAz-prime}
   A^{c}_{\mathbf{x}} \mathbf{z}^{c}= \mathbf{y}_{\mathbf{x}}
\end{equation}
to obtain $\mathbf{z}^{c}$. Notice that $\mathbf{y}_{\mathbf{x}}$ above is again the true data vector given in \cref{eq-yx}.

A close relationship exists between real-coefficient and complex-coefficient interpolation schemes. In fact, we have the decomposition $A^{c}_{\mathbf{x}} = D_{\mathbf{x}}V_{\mathbf{x}}$, where the diagonal matrix (with $\omega_k = e^{i x_k}$)
\begin{equation}
    D _{\mathbf{x}} :=\operatorname{diag}\left(\omega_0^{-\Omega_r}, \omega_1^{-\Omega_r},\ldots, \omega_{2r}^{-\Omega_r}\right) \in \mathbb{C}^{n \times n}
\end{equation}
is unitary, and
\begin{equation}\label{eq-Vx}
    V _{\mathbf{x}}:= \left[\begin{array}{cccccccc}
    1 & \omega_0^{\Omega_r-\Omega_{r-1}} & \cdots & \omega_0^{\Omega_r-\Omega_{1}} &\omega_0^{\Omega_r} & \omega_0^{\Omega_r+\Omega_{1}} & \cdots& \omega_0^{2\Omega_r} \\
    1 & \omega_1^{\Omega_r-\Omega_{r-1}} & \cdots &  \omega_1^{\Omega_r-\Omega_{1}} &\omega_1^{\Omega_r} & \omega_1^{\Omega_r+\Omega_{1}} & \cdots& \omega_1^{2\Omega_r} \\
    \vdots & \vdots & \ddots & \vdots & \vdots & \ddots & \vdots\\
    1 & \omega_{2r}^{\Omega_r-\Omega_{r-1}} & \cdots & \omega_{2r}^{\Omega_r-\Omega_{1}} &\omega_{2r}^{\Omega_r}  & \omega_{2r}^{\Omega_r+\Omega_{1}}  & \cdots& \omega_{2r}^{2\Omega_r}
    \end{array}\right]\in \mathbb{C}^{n \times n},
\end{equation}
which becomes the Vandermonde matrix for equidistant frequencies $\Omega_k=k$ for $k=1,\ldots,r$. The relationship between the real interpolation matrix $A_{\mathbf{x}}$ in \cref{eq-A} and all the previously defined matrices is as follows:
\begin{equation}\label{eq-matrix-relations}
    A_{\mathbf{x}} = A^{c}_{\mathbf{x}} C = D_{\mathbf{x}} V_{\mathbf{x}} C.
\end{equation}
Since $D_{\mathbf{x}}$ and $C$ are always invertible, $A_{\mathbf{x}}$ is invertible if and only if $V_{\mathbf{x}}$ is invertible. When equidistant frequencies $\Omega_k=k$ for $k=1,\ldots,r$, hold, we can derive the well-known Vandermonde determinant:
\begin{equation}
    \operatorname{det}(V_{\mathbf{x}})
    = \prod_{0 \leq j < k \leq n} \left( \omega_k - \omega_j \right)
    = \prod_{0 \leq j < k \leq n} \left( e^{ix_k} - e^{ix_j}\right),
\end{equation}
which is non-zero if and only if all $x_i$ are distinct modulo $2\pi$. Therefore, the \cref{eq-yAz} (also \cref{eq-yAz-prime}) always has the unique solution. However, for general $\Omega_k$, it is difficult to find a similar conclusion. 

In our ICD algorithms, we only need to use the matrix $A_{\mathbf{x}}$ to recover the real coefficients $\mathbf{z}$, but the complex coefficient counterpart is very convenient for theoretical analysis. For example, \cref{lem-shift-invariant} in \cref{app-pratical-ICD} is a typical instance.

\section{Reduced ICD algorithm}\label{app-pratical-ICD}

In the standard ICD \cref{alg-standard-ICD}, each interpolation requires $2r_j+1$ function evaluations. 
In the following, we propose a reduced ICD \cref{alg-reduced-ICD}, which is the almost same as the standard ICD but only requires $2r_j$ function evaluations per iteration (the same as in RCD).
The differences from \cref{alg-standard-ICD} are marked with $\blacktriangleright$.

\begin{algorithm}[H]
\caption{Reduced ICD Method for \cref{Problem}}
\label{alg-reduced-ICD}
\SetAlgoLined
\SetKwInOut{Input}{Input}
\SetKwInOut{Output}{Output}
\Input{Initial parameters $\boldsymbol{\theta}^0 = [\theta_1^0, \ldots, \theta_m^0]^{\dagger}$, and the number of iterations $\mathsf{T}$.}
\Output{Optimized parameters $\boldsymbol{\theta}^{\mathsf{T}}$ after $\mathsf{T}$ iterations.}
Obtain the optimal interpolation schemes $\{(\mathbf{x}^{j,*},A_{\mathbf{x}^{j,*}}^{-1})\}_{j=1}^m$ using \cref{alg-opt-interpolation}\; 
$\blacktriangleright$ Compute initial function value $\hat{f}^{0} := \tilde{f}(\boldsymbol{\theta}^{0})$\; 
\For{$t = 0$ \KwTo $\mathsf{T}$}{
    Select a coordinate $j \in \{1, \ldots, m\}$, either sequentially or uniformly at random\; 
    Fix all parameters of $\boldsymbol{\theta}^t$ except for $\theta_j^t$, and consider the restricted univariate function $\theta_j \mapsto f(\theta_j)$\;  
    $\blacktriangleright$ (Quantum burden) Construct the observed data vector at $\mathbf{x}^{j,*}$ but starting from $\theta_j^t$, i.e., 
\begin{equation}
 \mathbf{y}_{\text{obs}} := [\tilde{f}(\theta_j^t), \tilde{f}(\theta_j^t+ (x_1^{j,*}- x^{j,*}_0)), \ldots, \tilde{f}(\theta_j^t+ (x_{2r}^{j,*}- x^{j,*}_0))]^{\dagger},
\end{equation}
and $\tilde{f}(\theta_j^t)$ is replaced by $\hat{f}^{t}$\;
    $\blacktriangleright$ Construct the matrix $E_{s}^{-1}$ in \cref{eq-defn-Bi} with $s:=\theta_j^t-x^{j,*}_0$\; 
    $\blacktriangleright$ Compute the estimated coefficients $\hat{\mathbf{z}} := E_{s}^{-1} A_{\mathbf{x}^{j,*}}^{-1}  \mathbf{y}_{\text{obs}}$ and recover the estimated function $\hat{f}(\theta_j)$ as in \cref{eq-inter-hat-f}\; 
    Let $\theta_j^{t+1} := \underset{\theta_j \in \mathbb{R}}{\operatorname{argmin}\;} \hat{f}(\theta_j)$\; 
    Let $\theta_i^{t+1} := \theta_i^{t}$ for all $i \neq j$\; 
    $\blacktriangleright$ Record $\hat{f}^{t+1} := \hat{f}(\theta_j^{t+1})$\; 
}
\end{algorithm}

We now provide a detailed explanation of the reduced ICD described in \cref{alg-reduced-ICD}. 
The improvement relies on the following lemma concerning the shift invariance of interpolation. 
We return to the notation introduced in \cref{sec-interp-real}, omitting the specific index $j$ for clarity.

\begin{lemma}[Shift invariance of interpolation nodes]\label{lem-shift-invariant}
Fix some nodes $\mathbf{x} \in \mathbb{R}^n$ and let $\mathbf{1} \in \mathbb{R}^n$ denotes the all-ones vector.
Then, for any shift value $s \in \mathbb{R}$, we have:
\begin{enumerate}
    \item \label{item-1} 
    $A_{\mathbf{x} + s \mathbf{1} }^{-1} = E_s^{-1} A_{\mathbf{x}}^{-1} $ where the block diagonal matrix $E_s^{-1}=\operatorname{diag}\left(1, B_1^{\dagger}, \ldots, B_r^{\dagger}\right) \in \mathbb{R}^{n \times n}$ and $B_i \in \mathbb{R}^{2 \times 2}$ is given by rotation matrices
\begin{equation}\label{eq-defn-Bi}
B_i :=\left[\begin{array}{ll}
\cos \left( \Omega_i s\right) & \sin \left(\Omega_i  s\right) \\
- \sin \left(\Omega_i s\right) & \cos \left(\Omega_i  s\right)
\end{array}\right];
\end{equation}
\item \label{item-2} $\|A_{\mathbf{x} + s \mathbf{1} }^{-1}\|_F^2 = \|A_{\mathbf{x}}^{-1}\|_F^2$.
\end{enumerate}
\end{lemma}

\begin{proof}
Let us now construct the matrices (introduced in \cref{app-recover-complex}) corresponding to the shifted interpolation nodes $\mathbf{x} + s \mathbf{1} $. First, we have (with $\omega_k = e^{i x_k}$)
\begin{equation}
    V_{\mathbf{x} + s \mathbf{1} }=
    \left[\begin{array}{ccccc}
    1 & e^{i s (\Omega_r-\Omega_{r-1})} \omega_0^{\Omega_r-\Omega_{r-1}} & e^{i s (\Omega_r-\Omega_{r-2})} \omega_0^{\Omega_r-\Omega_{r-2}} & \cdots & e^{i s 2\Omega_r} \omega_0^{2 \Omega_r} \\
    1 & e^{i s (\Omega_r-\Omega_{r-1})} \omega_1^{\Omega_r-\Omega_{r-1}} & e^{i s (\Omega_r-\Omega_{r-2})} \omega_1^{\Omega_r-\Omega_{r-2}} & \cdots & e^{i s 2\Omega_r} \omega_1^{2 \Omega_r} \\
    \vdots & \vdots & \vdots & \ddots & \vdots \\
    1 & e^{i s (\Omega_r-\Omega_{r-1})} \omega_{2 r}^{\Omega_r-\Omega_{r-1}} & e^{i s (\Omega_r-\Omega_{r-2})} \omega_{2 r}^{\Omega_r-\Omega_{r-2}} & \cdots & e^{i s 2\Omega_r} \omega_{2 r}^{2 \Omega_r} \\
    \end{array}\right].
\end{equation}
We can observe that $ V_{\mathbf{x} + s \mathbf{1} } =  V_{\mathbf{x}} S$, where $S \in \mathbb{C}^{n \times n}$ is the diagonal matrix
\begin{equation}
    S :=\operatorname{diag}\left(1, e^{i s (\Omega_r-\Omega_{r-1})}, e^{i s (\Omega_r-\Omega_{r-2})},\ldots, e^{i s 2\Omega_r}\right).
\end{equation}
Next, we have
\begin{equation}
    D_{\mathbf{x} + s \mathbf{1} }=\operatorname{diag}\left(e^{-i s \Omega_ r} \omega_0^{-\Omega_r}, e^{-i s \Omega_ r} \omega_1^{-\Omega_r},\ldots, e^{-i s \Omega_ r} \omega_{2r}^{-\Omega_r}\right)= e^{-i s \Omega_ r} \cdot D_{\mathbf{x}}.
\end{equation}
Finally, using the identity $C C^\dagger = \frac{1}{2} I$, we obtain
\begin{equation}
    A_{\mathbf{x} + s \mathbf{1}} = D_{\mathbf{x} + s \mathbf{1}} V_{\mathbf{x} + s \mathbf{1}} C = (e^{-i s \Omega_r} \cdot D_{\mathbf{x}}) (V_{\mathbf{x}} S) C = 2 e^{-i s \Omega_r} \cdot (D_{\mathbf{x}} V_{\mathbf{x}} C) C^\dagger S C.
\end{equation}
Since $A_{\mathbf{x}} = D_{\mathbf{x}} V_{\mathbf{x}} C$, we get
\begin{equation}
    A_{\mathbf{x} + s \mathbf{1}} = A_{\mathbf{x}} (2 e^{-i s \Omega_r} \cdot C^\dagger S C) = A_{\mathbf{x}} E_s,
\end{equation}
where $S^{\prime} := e^{-i s \Omega_ r} \cdot S = \operatorname{diag}(e^{-i s \Omega_ r}, \ldots, e^{-i s \Omega_1}, 1, e^{i s \Omega_1}, \ldots, e^{i s \Omega_r})$ is unitary, and
\begin{equation}
    E_s := 2 C^{\dagger} S^{\prime} C.
\end{equation}
In fact, after substituting the specific expression \cref{eq-C} for $C$, $E_s$ can be expressed as a block diagonal matrix given by
\begin{equation}
    E_s = \operatorname{diag}\left(1, B_1, \ldots, B_r\right) \in \mathbb{R}^{n \times n}
\end{equation}
where $B_i$ is defined as in \cref{eq-defn-Bi}. Since each $B_i$ is a rotation matrix, we know that $B_i^{-1} = B_i^{\dagger}$, so
\begin{equation}
    E_s^{-1} = 2 C^\dagger (S')^{-1} C = \operatorname{diag}\left(1, B_1^{\dagger}, \ldots, B_r^{\dagger}\right).
\end{equation}
This proves \cref{item-1} of \cref{lem-shift-invariant}. On the other hand,
\begin{equation}
    \begin{aligned}
    \left\|A_{\mathbf{x}+s \mathbf{1}}^{-1}\right\|_F^2
    & = \operatorname{tr}\left[\left(A_{\mathbf{x}+s \mathbf{1}}^{-1}\right)^{\dagger} A_{\mathbf{x}+s \mathbf{1}}^{-1}\right] \\
    & = \operatorname{tr}\left[\left(A_{\mathbf{x}}^{-1} \right)^{\dagger} \left(E_s^{-1}\right)^{\dagger} E_s^{-1} A_{\mathbf{x}}^{-1} \right] \quad \left(\text{using } A_{\mathbf{x} + s \mathbf{1}} = A_{\mathbf{x}} E_s \right)\\
    & = \operatorname{tr}\left[\left(A_{\mathbf{x}}^{-1} \right)^{\dagger} \left(2 C^{\dagger} \left(S^{\prime} \right)^{-1} C\right)^{\dagger} 2 C^{\dagger} \left(S^{\prime} \right)^{-1} C A_{\mathbf{x}}^{-1} \right] \\
    & = \operatorname{tr}\left[\left(A_{\mathbf{x}}^{-1} \right)^{\dagger} \left(4 C^{\dagger} S^{\prime} C C^{\dagger} \left(S^{\prime} \right)^{-1} C \right) A_{\mathbf{x}}^{-1} \right] \\
    & = \operatorname{tr}\left[\left(A_{\mathbf{x}}^{-1} \right)^{\dagger} A_{\mathbf{x}}^{-1} \right] = \left\|A_{\mathbf{x}}^{-1}\right\|_F^2. \quad (\text{since } S' \text{ is unitary, } C C^{\dagger} = C^{\dagger} C = \frac{1}{2} I )
    \end{aligned}
\end{equation}
This completes the proof of \cref{item-2} of \cref{lem-shift-invariant}.
\end{proof}

The \cref{item-2} in \cref{lem-shift-invariant} states that the value of $\|A_{\mathbf{x}}^{-1}\|_F^2$ is invariant under any shift of the fixed vector $\mathbf{x}$. 
As a result, if a solution $\mathbf{x}^*=[x_0^*,x_1^*,\ldots,x_{2r}^*]^{\dagger} \in \mathbb{R}^n$ solves \cref{pro-min-mse}, then any shifted vector $\mathbf{x}^* + s\mathbf{1} \in \mathbb{R}^n$ is also a solution. 
This implies that for $\operatorname{MSE}(\hat{\mathbf{z}}_{\mathbf{x}})$, the key is the spacing between the interpolation nodes, rather than their absolute positions. 
Consequently, we are free to choose the position of the first point, and then determine the subsequent $2r$ points based on the optimal spacing of $\mathbf{x}^*$. 
For example, given any specific $\theta \in \mathbb{R}$, we can choose $s:= \theta - x^*_0$, then
\begin{align}
    \mathbf{x}_{\text{new}}^*  = \mathbf{x}^* + s\mathbf{1} &= [x_0^* + s,x_1^*+ s,\ldots,x_{2r}^*+ s]^{\dagger} \\
& = [\theta,\theta + (x_1^*- x^*_0),\ldots,\theta +(x_{2r}^*- x^*_0)]^{\dagger}
\end{align}
is another valid solution and preserve the minimization of the $\operatorname{MSE}(\hat{\mathbf{z}}_{\mathbf{x}})$ in \cref{pro-min-mse}. 
Moreover, to obtain the formulation for $A_{\mathbf{x}_{\text{new}}^*}^{-1}$, we do not need to naively compute the inverse matrix again. The  \cref{item-1} in \cref{lem-shift-invariant} shows that we only need to multiply the original inverse matrix $A_{\mathbf{x}^*}^{-1}$ by a simple block diagonal matrix $E_s^{-1}$ that depends on the shift value $s$, whose diagonal blocks are rotation matrices $B_i$. 

Return to the $t$-th iteration of \cref{alg-reduced-ICD}, after selecting coordinate $j$, we deliberately set the first node to be the current value at $j$, i.e., $\theta_j^t$. The following reasoning motivates this choice.
Consider the previous iteration $\boldsymbol{\theta}^{t-1} = [\theta_1^{t-1}, \ldots, \theta_k^{t-1}, \ldots, \theta_m^{t-1}]^{\dagger}$. After selecting coordinate $k$, we update $\theta_k^{t-1} \to \theta_k^{t} := \operatorname{argmin}\; \hat{f}(\theta_k)$, and the new point becomes $\boldsymbol{\theta}^{t} = [\theta_1^{t-1}, \ldots, \theta_k^{t}, \ldots, \theta_m^{t-1}]^{\dagger}$. From the unbiasedness of the approximated function, as shown in \cref{eq-unbaised-hatf}, we know that
\begin{equation}\label{eq-Ef1}
    \operatorname{E}[\hat{f}(\theta_k^t)] = f(\boldsymbol{\theta}^t).
\end{equation}
In the $t$-th iteration, we select another coordinate $j$. Due to the properties of $\tilde{f}$, we also have
\begin{equation}\label{eq-Ef2}
    \operatorname{E}[\tilde{f}(\theta_j^t)] = f(\boldsymbol{\theta}^t).
\end{equation}
Thus, both $\hat{f}(\theta_k^t)$ and $\tilde{f}(\theta_j^t)$ are unbiased estimates of $f(\boldsymbol{\theta}^t)$. This means that $\hat{f}(\theta_k^t)$ can replace $\tilde{f}(\theta_j^t)$, allowing us to reduce one function evaluation when we set the first interpolation node to $\theta_j^t$. 
Combining these considerations, we describe the reduced ICD algorithm in \cref{alg-reduced-ICD}.

However, we did not previously analyze their variances in \cref{eq-Ef1} and \cref{eq-Ef2}. Specifically, the variance of $\tilde{f}(\theta_j^t)$ is $\sigma^2$ according to \cref{assm-tilde-f}, whereas the variance of $\hat{f}(\theta_k^t)$ depends on $\theta_k^t$ and is generally difficult to determine. It is typically not exactly equal to $\sigma^2$, implying that their probability distributions are not strictly identical. Nevertheless, when the PQC problem involves equidistant frequencies and $\frac{2\pi}{n}$-equidistant nodes are adopted, as shown in \cref{eq-1101}, the variance at any point remains $\sigma^2$, matching that of $\tilde{f}(\theta_j^t)$. Consequently, we may theoretically regard $\hat{f}(\theta_k^t)$ as a genuine sample from $\tilde{f}(\theta_j^t) \sim \mathcal{N}(f(\theta_j^t), \sigma^2)$. At this point, the theoretical foundation of the reduced ICD in \cref{alg-reduced-ICD} is complete

\begin{remark}[Beware of error accumulation!]
Although reduced ICD reduces the number of function evaluations by one, it sacrifices sample independence, leading to cumulative errors that may ultimately cause the algorithm to fail. In experiments of \cref{subsec-num-re-1}, we clearly observed this phenomenon. Consequently, reduced ICD is less robust than standard ICD. To mitigate this issue, one can combine the strengths of both methods: for instance, after M iterations of reduced ICD, perform a single standard ICD step to reset the accumulated error. Of course, with sufficient quantum resources, directly applying standard ICD remains the most reliable approach.
\end{remark}

\section{Proofs of optimality of \texorpdfstring{$\frac{2 \pi}{n}$}{}-equidistant interpolation nodes} \label{app-proof-equidistnt}

In \cref{sec-equ-opt}, we asserted that $\frac{2 \pi}{n}$-equidistant interpolation nodes are optimal from three perspectives, as formalized in \cref{thm-first-veiw,thm-second-veiw,thm-third-veiw}. The purpose of this appendix is to prove these three theorems. We begin by establishing several auxiliary results. 

\subsection{Auxiliary results for Vandermonde matrix}

\begin{lemma}\label{lem-complex-1}
Let $n \geq 2$ be an integer, and let $m$ be any nonzero integer such that $-n < m < n$. Consider the $n$-th roots of unity (the solutions to the equation $z^n = 1$) $\omega_k = e^{2 \pi i \frac{k}{n}}$ for $k = 0, 1, \ldots, n-1$. Then,
\begin{equation}
    \sum_{k=0}^{n-1} \omega_k^{m} = 0.
\end{equation}
\end{lemma}

\begin{proof}
We need to show $S := \sum_{k=0}^{n-1} e^{2\pi i \frac{mk}{n}} = 0.$ This sum is a geometric series with the common ratio $q := e^{2\pi i \frac{m}{n}} \neq 1$ since $m$ is a nonzero integer and $-n < m < n$. Then,
\begin{equation}
    S = \sum_{k=0}^{n-1} q^k = \frac{1 - q^n}{1 - q}.
\end{equation}
Since $q^n = \left(e^{2\pi i \frac{m}{n}}\right)^n = e^{2\pi i m} = 1$, we have $S = 0.$
\end{proof}

\begin{lemma}\label{lem-complex-2}
Let $n \geq 2$ be an integer. Consider the $n$-th roots of unity $\omega_k = e^{2 \pi i \frac{k}{n}}$ for $k = 0, 1, \ldots, n-1$. Then, for any pair $k, j$ such that $k < j$, and $k, j = 0, 1, \ldots, n-1$, we have
\begin{equation}
    1+\frac{\omega_k}{\omega_j}+\left(\frac{\omega_k}{\omega_j}\right)^2+\cdots+\left(\frac{\omega_k}{\omega_j}\right)^{n-1}=0.
\end{equation}
\end{lemma}
\begin{proof}
Let $\omega$ be an arbitrary $n$-th root of unity and $\omega \neq 1$. Then, $\omega^n = 1$. Note that
\begin{equation}
    (\omega-1)\left(\omega^{n-1}+\omega^{n-2}+\cdots+1\right) = \omega^n-1=0.
\end{equation}
Since $\omega-1 \neq 0$, we conclude that
\begin{equation}\label{eq-sum-roots}
    1+\omega+\omega^2+\cdots+\omega^{n-1}=0.
\end{equation}
Given $k < j$, and $k, j = 0, 1, \ldots, n-1$, then
\begin{equation}
    \frac{\omega_k}{\omega_j} = \frac{e^{2 \pi i \frac{k}{n}}}{e^{2 \pi i \frac{j}{n}}}
    = e^{2 \pi i \frac{k-j}{n}}
    = e^{2 \pi i \frac{n + k - j}{n}}.
\end{equation}
Note that all possible values of $n + k - j$ are $1, 2, \ldots, n-1$, except for $0$ and $n$. Therefore, $\frac{\omega_k}{\omega_j}$ is again one of the $n$-th roots of unity, but is not equal to one. Applying \cref{eq-sum-roots} completes the lemma.
\end{proof}

\begin{lemma}\label{lem-Vandermonde}
Let $n \geq 2$ be an integer. Consider the $n$-th roots of unity $\omega_k = e^{2 \pi i \frac{k}{n}}$ for $k = 0, 1, \ldots, n-1$. Consider the Vandermonde matrix
\begin{equation}
    V := \left[\begin{array}{ccccc}
    1 & \omega_0 & \omega_0^2 & \cdots & \omega_0^{n-1}\\
    1 & \omega_1 & \omega_1^2 & \cdots & \omega_1^{n-1}\\
    \vdots & \vdots & \vdots & \ddots & \vdots \\
    1 & \omega_{n-1} & \omega_{n-1}^2 & \cdots & \omega_{n-1}^{n-1}
    \end{array}\right]\in \mathbb{C}^{n\times n}.
\end{equation}
Then, $V^{\dagger} V = V V^{\dagger} = n I$.
\end{lemma}

\begin{proof}
First, consider the product
\begin{equation}
    V^{\dagger} V = \left[\begin{array}{cccc}
    1 & 1 & \cdots & 1 \\
    \omega_0^{-1} & \omega_1^{-1} & \cdots & \omega_{n-1}^{-1} \\
    \omega_0^{-2} & \omega_1^{-2} & \cdots & \omega_{n-1}^{-2} \\
    \vdots & \vdots & \ddots & \vdots \\
    \omega_0^{-(n-1)} & \omega_1^{-(n-1)} & \cdots & \omega_{n-1}^{-(n-1)}
    \end{array}\right]
    \left[\begin{array}{ccccc}
    1 & \omega_0 & \omega_0^2 & \cdots & \omega_0^{n-1}\\
    1 & \omega_1 & \omega_1^2 & \cdots & \omega_1^{n-1}\\
    \vdots & \vdots & \vdots & \ddots & \vdots \\
    1 & \omega_{n-1} & \omega_{n-1}^2 & \cdots & \omega_{n-1}^{n-1}
    \end{array}\right].
\end{equation}
For indices $j, l = 1, \ldots, n$, the component of $V^{\dagger} V$ is
\begin{equation}
    [V^{\dagger} V]_{jl} = \sum_{k=0}^{n-1} \omega_{k}^{l-j}.
\end{equation}
If $j = l$, then $[V^{\dagger} V]_{jj} = \sum_{k=0}^{n-1} \omega_{k}^{0} = n.$ Consider $j \neq l$. Note that all possible values of $l-j$ are $\pm 1, \pm 2, \ldots, \pm (n-1)$. By \cref{lem-complex-1}, we have $[V^{\dagger} V]_{jl} = 0$. Thus, $V^{\dagger} V = n I$. Now consider the product
\begin{equation}
    VV^{\dagger} =
    \left[\begin{array}{ccccc}
    1 & \omega_0 & \omega_0^2 & \cdots & \omega_0^{n-1}\\
    1 & \omega_1 & \omega_1^2 & \cdots & \omega_1^{n-1}\\
    \vdots & \vdots & \vdots & \ddots & \vdots \\
    1 & \omega_{n-1} & \omega_{n-1}^2 & \cdots & \omega_{n-1}^{n-1}
    \end{array}\right]
    \left[\begin{array}{cccc}
    1 & 1 & \cdots & 1 \\
    \omega_0^{-1} & \omega_1^{-1} & \cdots & \omega_{n-1}^{-1} \\
    \omega_0^{-2} & \omega_1^{-2} & \cdots & \omega_{n-1}^{-2} \\
    \vdots & \vdots & \ddots & \vdots \\
    \omega_0^{-(n-1)} & \omega_1^{-(n-1)} & \cdots & \omega_{n-1}^{-(n-1)}
    \end{array}\right].
\end{equation}
For indices $j, l = 1, \ldots, n$, the component of $VV^{\dagger}$ is
\begin{equation}
    [V V^{\dagger}]_{jl} = \sum_{k=0}^{n-1}\left(\frac{\omega_j}{\omega_l}\right)^{k}.
\end{equation}
If $j = l$, then $[V V^{\dagger}]_{jj} = \sum_{k=0}^{n-1} 1^{k} = n.$ For $j \neq l$ and $j < l$, by \cref{lem-complex-2}, the elements above the diagonal are zeros. By symmetry, we have $V^{\dagger} V = n I$.
\end{proof}

\subsection{Other auxiliary results}

Let $S_{++}^n$ denote the set of $n \times n$ symmetric positive definite matrices. The following two properties about positive definite matrices can be found in many matrix textbooks. 

\begin{lemma}\label{lem-pd-1}
For any $X \in S_{++}^n$, we have the inequality $\operatorname{tr}(X^{-1}) \geqslant \frac{n^2}{\operatorname{tr}(X)}.$ The equality holds if and only if $X = \lambda I$ for some $\lambda > 0$.
\end{lemma}
\begin{lemma}\label{lem-pd-2}
For any $X \in S_{++}^n$, we have $(X^{-1})_{ii} X_{ii} \geq 1$ for all $i=1,\ldots,n$
\end{lemma}

\begin{lemma}\label{lem:Cauchy-1}
For $x_i > 0, \forall i=1, \ldots, n$, we have $\left(\sum_{i=1}^n x_i\right)\left(\sum_{i=1}^n \frac{1}{x_i}\right) \geq n^2,$ with equality holding if and only if all $x_i$ are equal.
\end{lemma}
\begin{proof}
By Cauchy-Schwarz inequality, we have
\begin{equation}
\left(\sum_{i=1}^n \left(\sqrt{x_i}\right)^2\right) \left(\sum_{i=1}^n \left(\frac{1}{\sqrt{x_i}}\right)^2\right)
    \geq \left(\sum_{i=1}^n \left(\sqrt{x_i} \cdot \frac{1}{\sqrt{x_i}}\right)\right)^2 = n^2.
\end{equation}
Equality holds if and only if there exists a constant $\lambda > 0$ such that $\lambda=\frac{\sqrt{x_i}}{\frac{1}{\sqrt{x_i}}}=x_i$, for all $i=1, \ldots, n$, i.e., when all $x_i$ are equal.
\end{proof}

\begin{lemma}\label{lem:Cauchy-2}
For $x_i > 0$ and $a_i > 0, \forall i = 1, \ldots, n$, we have $\left(\sum_{i=1}^n a_i x_i\right)\left(\sum_{i=1}^n \frac{a_i}{x_i}\right) \geq \left(\sum_{i=1}^n a_i\right)^2,$ with equality holding if and only if all $x_i$ are equal.
\end{lemma}
\begin{proof}
By Cauchy-Schwarz inequality, we have
\begin{equation}
\left(\sum_{i=1}^n \left(\sqrt{a_i x_i}\right)^2\right) \left(\sum_{i=1}^n \left(\sqrt{\frac{a_i}{x_i}}\right)^2\right)
\geq \left(\sum_{i=1}^n \left( \sqrt{a_i x_i} \cdot \sqrt{\frac{a_i}{x_i}}\right)\right)^2 
= \left(\sum_{i=1}^n a_i\right)^2.
\end{equation}
Equality holds if and only if there exists a constant $\lambda > 0$ such that $\lambda=\frac{\sqrt{a_i x_i}}{\sqrt{\frac{a_i}{x_i}}}=x_i$, for all $i=1, \ldots, n$, i.e., when all $x_i$ are equal.
\end{proof}

\subsection{Proof of \cref{thm-first-veiw}}\label{sec-proof-thm-1}

\begin{proof}[Proof of \cref{thm-first-veiw}]
We first demonstrate that $2\sigma^2$ is the (global) lower bound of cost function $\operatorname{MSE}(\hat{\mathbf{z}}_{\mathbf{x}})$ in \cref{eq-min-mse}, and then we prove that this lower bound can be achieved when taking $\frac{2 \pi}{n}$-equidistant nodes defined in \cref{eq-2pi/n-equidistant}. These two steps will complete the proofs.

For any $\mathbf{x} \in \mathbb{R}^n$ with $x_i$ distinct modulo $2 \pi$, consider the matrix $V_{\mathbf{x}}$ defined in \cref{eq-Vx}. It is easy to see that $[V_{\mathbf{x}}^{\dagger} V_{\mathbf{x}}]_{kk} = n$ for all $k=1,\ldots,n$; thus, $\operatorname{tr}(V_{\mathbf{x}}^{\dagger} V_{\mathbf{x}}) = n^2$. On the other hand, from the matrix relations \cref{eq-matrix-relations} and $C C^{\dagger}=\frac{1}{2} I$ in \cref{eq-C}, one has
\begin{equation}
\operatorname{tr}(A_{\mathbf{x}}^{\dagger} A_{\mathbf{x}}) = \operatorname{tr}( C^{\dagger} V_{\mathbf{x}}^{\dagger} D_{\mathbf{x}}^{\dagger}D_{\mathbf{x}} V_{\mathbf{x}} C) 
= \frac{1}{2} \operatorname{tr}(V_{\mathbf{x}}^{\dagger}V_{\mathbf{x}}) 
= \frac{n^2}{2} . 
\end{equation}
By \cref{lem-pd-1}, we obtain
\begin{equation}
\operatorname{MSE}(\hat{\mathbf{z}}_{\mathbf{x}})
=\sigma^2 \|A_{\mathbf{x}}^{-1}\|_F^2
=\sigma^2 \operatorname{tr}([A_{\mathbf{x}}^{\dagger} A_{\mathbf{x}}]^{-1})
\geq
\sigma^2 \frac{n^2 }{\operatorname{tr}(A_{\mathbf{x}}^{\dagger} A_{\mathbf{x}})} 
=\sigma^2 \frac{n^2 }{\frac{n^2}{2}}
=2 \sigma^2.
\end{equation}
Hence, $2 \sigma^2$ is a lower bound for $\operatorname{MSE}(\hat{\mathbf{z}}_{\mathbf{x}})$ for all $\mathbf{x} \in \mathbb{R}^n$ with $x_i$ distinct modulo $2 \pi$.

By \cref{lem-shift-invariant}, the value of $\operatorname{MSE}(\hat{\mathbf{z}}_{\mathbf{x}})$ is invariant under any shift of $\mathbf{x}$. Therefore, it is sufficient to consider $\frac{2\pi}{n}$-equidistant nodes $\mathbf{x}^*$ without shift, i.e.,
\begin{equation}
    x_k^* = \frac{2\pi}{n} k \quad \text{for } k = 0, 1, \ldots, 2r.
\end{equation}
This implies that $\omega_k \equiv e^{i x_k^*} = e^{2\pi i \frac{k}{n}}$ are the $n$-th roots of unity. In this case, by \cref{lem-Vandermonde}, we have $ V_{\mathbf{x}^*}^{\dagger} V_{\mathbf{x}^*} = V _{\mathbf{x}^*}V_{\mathbf{x}^*}^{\dagger} = n I.$ Hence, 
\begin{equation}
    A_{\mathbf{x}^*}^{\dagger}  A_{\mathbf{x}^*} 
    = C^{\dagger} V_{\mathbf{x}}^{\dagger} D_{\mathbf{x}}^{\dagger}  D_{\mathbf{x}} V_{\mathbf{x}} C
    = C^{\dagger} V_{\mathbf{x}}^{\dagger} V_{\mathbf{x}} C
    = \frac{n}{2} I.
\end{equation}
Then, $\operatorname{MSE}(\hat{\mathbf{z}}_{\mathbf{x}^*})
=\sigma^2 \operatorname{tr}([A_{\mathbf{x}}^{\dagger} A_{\mathbf{x}}]^{-1})
=\frac{2}{n} \sigma^2 \operatorname{tr}( I)=2\sigma^2.$ Therefore, the lower bound $2 \sigma^2$ is achieved at $\frac{2\pi}{n}$-equidistant nodes $\mathbf{x}^*$.
\end{proof}

\subsection{Proof of \cref{thm-second-veiw}}\label{sec-proof-thm-2}

Given an arbitrary invertible complex matrix $X \in \mathbb{C}^{n \times n}$, the 2-norm condition numbers defined as
\begin{equation}
\kappa_2(X):=\|X\|_2\|X^{-1}\|_2=\frac{\sigma_{\max }(X)}{\sigma_{\min }(X)} \geq 1.
\end{equation}
where $\|X\|_2$ is spectral norm, and $\sigma_{\max }(X)$ ($\sigma_{\min }(X)$) is the largest (smallest) singular value of $X$. The lower bound 1 is attained if and only if all singular values of $X$ are equal. The next lemma implies that the condition number $\kappa_2$ is invariant under multiplication by a scaled unitary matrix.

\begin{lemma}\label{lem-cond2}
Given any nonsingular complex matrix $A$ and any complex matrix $B$ with $B^{\dagger}B = \lambda I$ for some constant $\lambda >0$, we have
\begin{equation}
    \kappa_2( A B ) = \kappa_2( B A ) = \kappa_2( A ).
\end{equation}
\end{lemma}
\begin{proof}
We have
\begin{align}
    \| B A \|_2^2
    &= \sup _{\mathbf{x} \in \mathbb{R}^n \backslash\{\mathbf{0}\}} \frac{\langle B A \mathbf{x}, B A \mathbf{x}\rangle}{\|\mathbf{x}\|_2^2}
    = \sup _{\mathbf{x} \in \mathbb{R}^n \backslash\{\mathbf{0}\}} \frac{\langle B ^{\dagger} B A \mathbf{x}, A \mathbf{x}\rangle}{\|\mathbf{x}\|_2^2} \\
    &= \sup _{\mathbf{x} \in \mathbb{R}^n \backslash\{\mathbf{0}\}} \frac{\lambda \langle A \mathbf{x}, A \mathbf{x}\rangle}{\|\mathbf{x}\|_2^2}
    = \lambda\| A \|_2^2,
\end{align}
and note that $B^{-1} = \frac{1}{\lambda} B^{\dagger}$, so
\begin{align}
    \| A^{-1} B^{-1} \|_2^2
    &= \sup _{\mathbf{x} \in \mathbb{R}^n \backslash\{\mathbf{0}\}} \frac{\left\langle A^{-1} B^{-1} \mathbf{x}, A^{-1} B^{-1} \mathbf{x}\right\rangle}{\|\mathbf{x}\|_2^2}
    = \sup _{\mathbf{y} \in \mathbb{R}^n \backslash\{\mathbf{0}\}} \frac{\left\langle A^{-1} \mathbf{y}, A^{-1} \mathbf{y}\right\rangle}{\| B \mathbf{y}\|_2^2} \\
    &= \sup _{\mathbf{y} \in \mathbb{R}^n \backslash\{\mathbf{0}\}} \frac{\left\langle A^{-1} \mathbf{y}, A^{-1} \mathbf{y}\right\rangle}{\lambda\|\mathbf{y}\|_2^2} = \frac{1}{\lambda} \| A^{-1}\|_2^2.
\end{align}
Hence, we have
\begin{equation}
    \kappa_2( B A ) = \| B A \|_2 \| A^{-1} B^{-1}\|_2 = \| A \|_2 \| A^{-1}\|_2 = \kappa_2( A ).
\end{equation}
Similarly, one can show that $\kappa_2(A B) = \kappa_2( A )$.
\end{proof}

Now we are ready to provide the proof
of \cref{thm-second-veiw}.

\begin{proof}[Proof of \cref{thm-second-veiw}]
We know that $1$ is the global lower bound of the cost function $\kappa_2(A_{\mathbf{x}})$ in \cref{eq-min-cn}. Therefore, it remains to show that $\kappa_2(A_{\mathbf{x}^*}) = 1$ for the $\frac{2\pi}{n}$-equidistant nodes $\mathbf{x}^*$ defined in \cref{eq-2pi/n-equidistant}.

We first establish that the value of $\kappa_2(A_{\mathbf{x}})$ is invariant under any shift of $\mathbf{x}$. By \cref{lem-shift-invariant}, we know that for any shift $s$, we have
\begin{equation}
    A_{\mathbf{x} + s \mathbf{1}} = A_{\mathbf{x}} E_s,
\end{equation}
where $E_s = \operatorname{diag}(1, B_1, \ldots, B_r) \in \mathbb{R}^{n \times n}$ and $B_i$ are rotation matrices (hence unitary), as defined in \cref{eq-defn-Bi}. It is straightforward to verify that $E_s$ is unitary. So, by applying \cref{lem-cond2}, we obtain:
\begin{equation}
    \kappa_2(A_{\mathbf{x} + s \mathbf{1}}) = \kappa_2(A_{\mathbf{x}}) \quad \text{for any shift value } s \in \mathbb{R}.
\end{equation}
Thus, we can consider the $\frac{2\pi}{n}$-equidistant nodes $\mathbf{x}^*$ without any shift, i.e., $ x_k^* = \frac{2\pi}{n} k$ for $k = 0, 1, \ldots, 2r.$ This implies that $\omega_k \equiv e^{i x_k^*} = e^{2\pi i \frac{k}{n}}$. In this case, by \cref{lem-Vandermonde}, we have $V_{\mathbf{x}^*}^{\dagger} V_{\mathbf{x}^*} = n I.$

Using the relations $A_{\mathbf{x}^*} = D_{\mathbf{x}^*} V_{\mathbf{x}^*} C$ from \cref{eq-matrix-relations} and $C^{\dagger} C = \frac{1}{2} I$, we can apply \cref{lem-cond2} twice to obtain
\begin{equation}
    \kappa_2(A_{\mathbf{x}^*}) = \kappa_2(V_{\mathbf{x}^*}).
\end{equation}
Thus, it suffices to show that $\kappa_2(V_{\mathbf{x}^*}) = 1$. Since $V_{\mathbf{x}^*}^{\dagger} V_{\mathbf{x}^*} = n I$, we set $B = V_{\mathbf{x}^*}$ and $A = I$ in \cref{lem-cond2}, yielding
\begin{equation}
    \kappa_2(V_{\mathbf{x}^*}) = \kappa_2(I) = 1.
\end{equation}
This completes the proof.
\end{proof}

\subsection{Proof of \cref{thm-third-veiw}}\label{sec-proof-thm-3}

Let us start with an auxiliary result in probability theory.

\begin{lemma}\label{lem-quadratic-form}
Consider the random vector $\mathbf{x}$ with $\operatorname{E}[\mathbf{x}] = \mu,\operatorname{VAR}[ \mathbf{x}] = \Sigma.$ If $A = yz^{\dagger}$ for some constant vectors $y, z$, then $\operatorname{E}[ \mathbf{x}^{\dagger} A  \mathbf{x}] = \mu^{\dagger} A \mu + \text{tr}(A \Sigma).$
\end{lemma}

\begin{lemma}\label{lem-cov-var}
For the $d$-th order derivative of $\hat{f}(x)$ in \cref{eq-app-fx}, i.e, $\hat{f}^{(d)}(x)=\mathbf{t}^{(d)}(x)^{\dagger} \hat{\mathbf{z}}_{\mathbf{x}}$ with $\mathbf{t}^{(d)}(x)$ given in \cref{eq-tdx}. We have
\begin{align}
\operatorname{Cov}[\hat{f}^{(d)}(x_1), \hat{f}^{(d)}(x_2)]
&= \mathbf{t}^{(d) }(x_1) ^{\dagger} \operatorname{VAR}[\hat{\mathbf{z}}_{\mathbf{x}}] \mathbf{t}^{(d)}(x_2), \quad \forall x_1, x_2 \in \mathbb{R}, \\
\operatorname{Var}[\hat{f}^{(d)}(x)]
&=\mathbf{t}^{(d)}(x)^{\dagger} \operatorname{VAR}[\hat{\mathbf{z}}_{\mathbf{x}}] \mathbf{t}^{(d)}(x), \quad \forall x \in \mathbb{R}.
\end{align}
\end{lemma}
\begin{proof}
To simplify notation, let $\hat{f}^{(d)}_1 := \hat{f}^{(d)}(x_1)$ and $\hat{f}^{(d)}_2 := \hat{f}^{(d)}(x_2)$; similarly, let $\mathbf{t}_1^{(d)} := \mathbf{t}^{(d)}(x_1)$ and $\mathbf{t}_2^{(d)} := \mathbf{t}^{(d)}(x_2)$. By \cref{lem-quadratic-form}, we can proceed as follows:
\begin{align}
\operatorname{Cov}[\hat{f}^{(d)}_1, \hat{f}^{(d)}_2] 
&= \operatorname{E}[\hat{f}^{(d)}_1 \hat{f}^{(d)}_2] - \operatorname{E}[\hat{f}^{(d)}_1] \operatorname{E}[\hat{f}^{(d)}_2] \\
&= \operatorname{E}[\hat{\mathbf{z}}_{\mathbf{x}}^{\dagger} \mathbf{t}_1^{(d)} \mathbf{t}_2^{(d) T} \hat{\mathbf{z}}_{\mathbf{x}}] - \mathbf{t}_1^{(d) T} \mathbf{z} \cdot \mathbf{t}_2^{(d) T} \mathbf{z} \\
&= \mathbf{z}^{\dagger} \mathbf{t}_1^{(d)} \mathbf{t}_2^{(d) T} \mathbf{z} + \operatorname{tr}\left( \mathbf{t}_1^{(d)} \mathbf{t}_2^{(d) T} \operatorname{VAR}[\hat{\mathbf{z}}_{\mathbf{x}}] \right) - \mathbf{z}^{\dagger} \mathbf{t}_1^{(d)} \cdot \mathbf{t}_2^{(d) T} \mathbf{z} \\
&= \mathbf{t}_1^{(d) T} \operatorname{VAR}[\hat{\mathbf{z}}_{\mathbf{x}}] \mathbf{t}_2^{(d)}, \label{eq-cov-x1x2}
\end{align}
which completes the proof.
\end{proof}

Now we are ready to provide the proof of \cref{thm-third-veiw}.

\begin{proof}[Proof of \cref{thm-third-veiw}]
For $d = 0$, as shown in \cref{eq-case-d=0}, \cref{thm-third-veiw} reduces to \cref{thm-first-veiw}, so we focus on the cases where $d \geq 1$. We first establish the lower bound of the cost function $h^{(d)}(\mathbf{x})$ in \cref{eq-min-ave}, and then prove that this lower bound is achieved when the $\frac{2\pi}{n}$-equidistant nodes defined in \cref{eq-2pi/n-equidistant} are used. These two steps complete the proof. 

For $d \geq 1$, from \cref{eq-var} and \cref{eq-int-ttt}, for all $\mathbf{x} \in \mathbb{R}^n$ (with distinct $x_i$ modulo $2\pi$), we have the following expression for $h^{(d)}(\mathbf{x})$:
\begin{equation}\label{eq-1043}
    h^{(d)}(\mathbf{x}) = \frac{1}{2} \langle \operatorname{VAR}[\hat{\mathbf{z}}_{\mathbf{x}}], \operatorname{diag}(p) \rangle = \frac{1}{2} \sigma^2 \langle B^{-1}, \operatorname{diag}(p) \rangle,
\end{equation}
where $p = [0, 1, 1, 2^{2d}, 2^{2d}, \dots, r^{2d}, r^{2d}]^{\dagger} \in \mathbb{R}^n$, and $B := A_{\mathbf{x}}^{\dagger} A_{\mathbf{x}} \in \mathbb{R}^{n \times n}$ is positive definite since $A_{\mathbf{x}}$ is invertible. 

We then claim the following equality for any $\mathbf{x}$. Note that $n = 2r + 1$ and $p_1 = 0$, so we have
\begin{align}
\sum_{i=1}^n p_i B_{i i}
&=\sum_{k=1}^r k^{2 d} \left([A_{\mathbf{x}}^{\dagger} A_{\mathbf{x}}]_{2 k, 2 k} +
[A_{\mathbf{x}}^{\dagger} A_{\mathbf{x}}]_{2 k+1, 2k+1}\right)\\
&=\sum_{k=1}^r k^{2 d}\left(\sum_{i=0}^{2 r}\left[\cos ^2\left(k x_i\right)+\sin ^2\left(k x_i\right)\right]\right) \quad \text{ (by definition of $A_{\mathbf{x}}$ in \cref{eq-A}) }  \\
&= \sum_{k=1}^r k^{2 d}n. \label{eq-pBB}
\end{align}
Now, turn to \cref{eq-1043}, we have
\begin{align}
h^{(d)}(\mathbf{x}) 
& =\frac{1}{2} \sigma^2\sum_{i=1}^n p_i \left[B^{-1}\right]_{i i}
\geq \frac{1}{2}\sigma^2  \left(\sum_{i=1}^n p_i\right)^2 \left[\sum_{i=1}^n \frac{p_i}{\left[B^{-1}\right]_{i i}}\right]^{-1} \quad \text{ (by \cref{lem:Cauchy-2}) }  \\
& \geq \frac{1}{2}\sigma^2  \left(\sum_{i=1}^n p_i\right)^2 \left[\sum_{i=1}^n p_i B_{i i}\right]^{-1} \quad  \text{ (by \cref{lem-pd-2}) } \\
& =\frac{1}{2}\sigma^2  \left(\sum_{i=1}^n p_i\right)^2 \left(n \sum_{k=1}^r k^{2d}\right)^{-1} \quad \text{ (by \cref{eq-pBB})} \\
& =\frac{2\sigma^2}{n} \sum_{k=1}^r k^{2d} .
\end{align}
Thus, $\frac{2 \sigma^2}{n} \sum_{k=1}^r k^{2d}$ is a lower bound for $h^{(d)}(\mathbf{x})$ for all $\mathbf{x}$.

Next, let us consider the $\frac{2 \pi}{n}$-equidistant nodes $\mathbf{x}^*$ without shift, i.e., $x_k^* = \frac{2 \pi}{n} k$ for $k = 0, 1, \dots, 2r$. By \cref{lem-Vandermonde}, we have
\begin{equation}
    B = A_{\mathbf{x}^*}^{\dagger} A_{\mathbf{x}^*} = C^{\dagger} V_{\mathbf{x}^*}^{\dagger} D_{\mathbf{x}^*}^{\dagger} D_{\mathbf{x}^*} V_{\mathbf{x}^*} C = \frac{n}{2} I.
\end{equation}
Thus, we obtain
\begin{equation}
    h^{(d)}(\mathbf{x}^*) = \frac{1}{2} \sigma^2 \sum_{i=1}^n p_i \left( B^{-1} \right)_{ii} = \frac{2 \sigma^2}{n} \sum_{k=1}^r k^{2d}.
\end{equation}
Therefore, the lower bound is achieved at $\mathbf{x}^*$ without any shift.

Finally, it remains to establish that the value of $h^{(d)}(\mathbf{x})$ is invariant under any shift of $\mathbf{x}$, particularly for $\mathbf{x}^*$. By \cref{lem-shift-invariant}, we know that for any shift $s$, we have
\begin{equation}
    A_{\mathbf{x} + s \mathbf{1}} = A_{\mathbf{x}} E_s,
\end{equation}
where $E_s = \operatorname{diag}(1, B_1, \ldots, B_r) \in \mathbb{R}^{n \times n}$, and $B_i \in \mathbb{R}^{2 \times 2}$ are rotation matrices, as defined in \cref{eq-defn-Bi}. Using this, we can express $h^{(d)}(\mathbf{x} + s\mathbf{1})$ as follows:
\begin{align}
h^{(d)}(\mathbf{x} + s \mathbf{1}) 
&= \frac{1}{2} \left\langle \operatorname{VAR}[\hat{\mathbf{z}}_{\mathbf{x} + s \mathbf{1}}], \operatorname{diag}(p) \right\rangle \\
&= \frac{1}{2} \sigma^2 \left\langle \left(A_{\mathbf{x} + s \mathbf{1}}^{\dagger} A_{\mathbf{x} + s \mathbf{1}}\right)^{-1}, \operatorname{diag}(p) \right\rangle \\
&= \frac{1}{2} \sigma^2 \left\langle E_s^{\dagger} \left(A_{\mathbf{x}}^{\dagger} A_{\mathbf{x}}\right)^{-1} E_s, \operatorname{diag}(p) \right\rangle \\
&= \frac{1}{2} \left\langle \operatorname{VAR}[\hat{\mathbf{z}}_{\mathbf{x}}], E_s \operatorname{diag}(p) E_s^{\dagger} \right\rangle.
\end{align}
On the other hand, we know that $E_s^{\dagger} = \operatorname{diag}(1, B_1^{\dagger}, \ldots, B_r^{\dagger})$ and that 
\begin{equation}
    \operatorname{diag}(p) = \operatorname{diag}(0, I_2, 2^{2d} I_2, \ldots, r^{2d} I_2),
\end{equation}
where $I_2$ is the $2 \times 2$ identity matrix. Therefore, we have
\begin{equation}
    E_s \operatorname{diag}(p) E_s^{\dagger}
    = \operatorname{diag}(0, B_1 I_2 B_1^{\dagger}, 2^{2d} B_2 I_2 B_2^{\dagger}, \ldots, r^{2d} B_r I_2 B_r^{\dagger})
    = \operatorname{diag}(p).
\end{equation}
Substituting this result into the expression for $h^{(d)}(\mathbf{x} + s\mathbf{1})$, we obtain
\begin{equation}
    h^{(d)}(\mathbf{x} + s\mathbf{1})
    = \frac{1}{2} \left\langle \operatorname{VAR}[\hat{\mathbf{z}}_{\mathbf{x}}], \operatorname{diag}(p) \right\rangle = h^{(d)}(\mathbf{x}).
\end{equation}
Thus, we have shown that $h^{(d)}(\mathbf{x})$ is invariant under any shift of $\mathbf{x}$. This completes the proof of \cref{thm-third-veiw}.
\end{proof}

\section{Review of general parameter shift rule}\label{app-psr}

Let us review the parameter shift rule (PSR) \cite{crooks2019gradients,mari2021estimating,wierichs2022general,kyriienko2021generalized,hai2023lagrange,markovich2024parameter,hoch2025variational}, and briefly compare it with the finite difference method.
Consider the cost function $f(x)$ in \cref{eq-trig-poly-real} under equidistant frequency \cref{assm-equi}, i.e.,
\begin{equation}
f(x)=\frac{1}{\sqrt{2}} a_0+\sum_{k=1}^r\left[a_k \cos \left(k x\right)+b_k \sin \left(k x\right)\right].
\end{equation}
\subsection{General parameter shift rule}

The general parameter shift rule aims to compute the derivative of $f(x)$ by only using linear combination of function evaluations. That is given by \cite{wierichs2022general}
\begin{equation}\label{eq-general-psr}
    g(x) := f^{\prime}(x) = \sum_{\mu=1}^{2r} \frac{(-1)^{\mu-1}}{4r \sin^2\left(\frac{1}{2} x_\mu\right)} f\left(x + x_\mu\right),
\end{equation}
where $x_\mu = \frac{\pi}{2r} + (\mu - 1) \frac{\pi}{r}$ for $\mu = 1, 2, \ldots, 2r$. Note that the coefficients preceding this linear combination are only related to $r$ and, are independent of $a_k$ and $b_k$.
For example, when $r = 1$, it reduces to
\begin{align}
    g(x) & = \frac{1}{4 \sin^2\left(\frac{1}{2} \cdot \frac{1}{2} \pi\right)} f\left(x + \frac{1}{2} \pi\right) - \frac{1}{4 \sin^2\left(\frac{1}{2} \cdot \frac{3}{2} \pi\right)} f\left(x + \frac{3}{2} \pi\right) \\
    & = \frac{1}{2} \left(f(x + \frac{\pi}{2}) - f(x - \frac{\pi}{2})\right).
\end{align}
In previous works \cite{wierichs2022general}, one might use the estimator
\begin{equation}
    g_{\text{psr}}(x) := \sum_{\mu=1}^{2r} \frac{(-1)^{\mu-1}}{4r \sin^2\left(\frac{1}{2} x_\mu\right)} \tilde{f}\left(x + x_\mu\right),
\end{equation}
to approximate $g(x)$. In our simulation experiments, the calculation of gradients required for RCD and SGD is based on the PSR presented here.
To the best of our knowledge, for general non-equidistant frequency $\Omega_k$, there is currently no explicit PSR of the same form as \cref{eq-general-psr}.

\subsection{Comparison between the parameter shift rule and finite difference}\label{app-psr-fd}

Now, let us compare PSR and (central) finite difference (FD). PSR leverages the known generator spectrum of gates to derive an exact gradient formula like (e.g. $U (x)=e^{ix H}$ and $H^2=I$)
\begin{equation}
    f^{\prime}(x)=[f (x+s)-f (x-s)]/(2\sin s),
\end{equation}
often with $s=\pi/2$ \cite{mari2021estimating}. FD approximates
\begin{equation}
    f^{\prime}(x) \approx[f (x+h)-f (x-h)]/ (2h),
\end{equation}
and requires $h \to 0$ for accuracy. In practice, when differentiating parameters in a parameterized quantum circuit, PSR is preferred over FD methods for the following reasons: 
\begin{enumerate}
\item  Accuracy: PSR yields the exact analytic derivatives without truncation error. FD incurs a truncation error of order $O (h^2)$ and only converges to the true derivative as $h\to0$.
\item Noise sensitivity: FD's reliance on very small $h$ makes the difference in measured expectations easily swamped by quantum hardware noise and finite shot statistics. PSR uses larger shifts (e.g. $\pi/2$), improving signal to noise in each pair of measurements and offering greater robustness on NISQ devices.
\item Bias: FD provides a biased estimator due to truncation error. PSR yields an unbiased gradient estimate (apart from environment noise).
\end{enumerate}
For further details, see \cite{mari2021estimating}, which offers a comprehensive error analysis comparing the PSR (in the case $H^2 = I$) and FD. The study proves that a properly scaled parameter shift estimator is always optimal, in the sense that it achieves a lower mean squared error. Of course, PSR typically applies only when the generator's spectrum is equally spaced, and the required number of function evaluations increases as the spectrum of $H$ becomes more complex. For generators with non-equally spaced spectrum, the more general PSRs were introduced in \cite{kyriienko2021generalized,markovich2024parameter}, though it is considerably more involved. PSR is generally preferred for computing derivatives in PQCs, because in most cases the spectrum is the singleton ($r=1$), corresponding to the $H^2 = I$ scenario above \cite{qiskit_gradients_framework,pennylane_generalized_parameter_shift}. Unless otherwise specified, all derivatives in this paper are computed using the PSR.

\section{Eigenvalue method of solving subproblem for equidistant frequencies}\label{app-eig-method}

In each iteration of ICD for equidistant frequencies, we need to find the global minimizer of the trigonometric polynomial
\begin{equation}
    f(x) = \frac{1}{\sqrt{2}} a_0 + \sum_{k=1}^r \left[ a_k \cos(kx) + b_k \sin(kx) \right]
\end{equation}
for $x \in [0, 2\pi]$ (here, we omit the hat symbol on the coefficients). To achieve this, we first find all the real roots of its derivative
\begin{equation}\label{eq-951}
    f'(x) = \sum_{k=1}^r \left[ -a_k k \sin(kx) + b_k k \cos(kx) \right] = 0
\end{equation}
within the interval $[0, 2\pi]$. These roots correspond to the stationary points of $f(x)$. By evaluating $f(x)$ at these points, we can determine the global minimizer. Fortunately, \cite[Theorem 2]{boyd2006computing} provides an exact method for transforming the problem of finding the roots of \cref{eq-951} into an eigenvalue problem. We apply this method as follows.

\begin{enumerate}[label=Step \arabic*:, left=0em]  
    \item Note that $f'(x)$ has the same structure as $f(x)$, as it can be written as
\begin{equation}
    f'(x) = \sum_{k=0}^r \tilde{a}_k \cos(kx) + \sum_{k=1}^r \tilde{b}_k \sin(kx)
\end{equation}
where we define
\begin{equation}
    \tilde{a}_k = \begin{cases}
     0, & \text{for } k = 0\\
    b_k k, & \text{for } k = 1, 2, \ldots, r
    \end{cases}
    \quad \text{and} \quad
    \tilde{b}_k = -a_k k, \text{for } k = 1, 2, \ldots, r.
\end{equation}
It is assumed that both $a_r$ and $b_r$ are nonzero; otherwise, they should be removed, and $r$ should be reduced.
\item Define the coefficients $h_j$ as
\begin{equation}
    h_j = \begin{cases}
    \tilde{a}_{r-j} + i \tilde{b}_{r-j}, & j = 0, 1, \ldots, r-1 ,\\
    2 \tilde{a}_0, & j = r, \\
    \tilde{a}_{j-r} - i \tilde{b}_{j-r}, & j = r+1, r+2, \ldots, 2r.
    \end{cases}
\end{equation}
Since $\tilde{a}_0 = 0$, it follows that $h_r = 0$.
\item Next, define a $2r \times 2r$ matrix $\mathbf{B}$ with entries $B_{kj}$ as,
\begin{equation}
    B_{kj} = \begin{cases}
    \delta_{k, j-1}, & \text{for } k = 1, 2, \ldots, 2r-1 ,\\
    -\dfrac{h_{j-1}}{\tilde{a}_r - i \tilde{b}_r}, & \text{for } k = 2r,
    \end{cases}
\end{equation}
where $\delta_{k, j-1}$ is the Kronecker delta function. For example, when $r = 2$, the matrix $\mathbf{B}$ is explicitly
\begin{equation}
    \mathbf{B} = \begin{bmatrix}
    0 & 1 & 0 & 0 \\
    0 & 0 & 1 & 0 \\
    0 & 0 & 0 & 1 \\
    -\frac{\tilde{a}_2 + i \tilde{b}_2}{\tilde{a}_2 - i \tilde{b}_2} & -\frac{\tilde{a}_1 + i \tilde{b}_1}{\tilde{a}_2 - i \tilde{b}_2} & 0 & - \frac{\tilde{a}_1 - i \tilde{b}_1}{\tilde{a}_2 - i \tilde{b}_2}
    \end{bmatrix}.
\end{equation}
Note that $\mathbf{B}$ has a significant sparse structure, with at most $4r - 2$ non-zero elements.
\item Let the eigenvalues of $\mathbf{B}$ be denoted by $z_t \in \mathbb{C}$. \cite[Theorem 2]{boyd2006computing} shows that the roots (which may be complex) of $f'(x) = 0$ are given by $ x_t = -i \log(z_t)$
where the complex logarithm is defined as $\log(z) = \log|z| + i(\arg(z) + 2\pi m), \forall m \in \mathbb{Z}$. Therefore, the final roots are
\begin{equation}
    x_t = (\arg(z_t) + 2\pi m) - i \log|z_t|, \quad t = 1, 2, \ldots, 2r, \quad \forall m \in \mathbb{Z}.
\end{equation}
Since we are only interested in the real roots of $f'(x)$, these real roots correspond to the eigenvalues $z_t$ lying on the unit circle. This simplifies to
\begin{equation}
    x_t = \arg(z_t) + 2\pi m, \quad \text{when} \quad |z_t| = 1.
\end{equation}
By taking $x_k$ modulo $2\pi$, the final real roots can be obtained. 

\item The global minimizer is the value of $x_t$ that yields the smallest $f(x)$ among these points.
\end{enumerate}

This method uses the inherent properties of trigonometric polynomials to transform the problem of finding the global minimizer of $f(x)$ into an equivalent problem of determining all eigenvalues with modulus equal to 1 of a sparse non-Hermitian matrix $\mathbf{B}$.
Compared to directly using global optimization solvers (e.g., differential evolution), which are typically heuristic algorithms, this eigenvalue approach guarantees the identification of the global minimum, thereby avoiding the risk of getting trapped in local minima. Although eigenvalue problems may appear complex, in practical applications, the integer $r$ is usually small, making it feasible to solve the eigenvalues of small matrices both efficiently and accurately.

\section{A concrete example of sparse frequency in TFIM model with HVA circuit}\label{app-sparse-frequency}

We consider a special case of the TFIM model with 3-qubit and a 2-layer HVA circuit (the similar circuit diagram is shown in \cref{fig:qc_tifm_HVA}). Specifically, we aim to minimize the cost function
\begin{equation}
    f (\boldsymbol{\theta})=\langle\psi_0| U (\boldsymbol{\theta})^{\dagger} H_{\mathrm{TFIM}} U (\boldsymbol{\theta})|\psi_0\rangle,
\end{equation}
where $\boldsymbol{\theta} \in \mathbb{R}^4$ and $\left|\psi_0\right\rangle=|+\rangle^{\otimes 3}$ is the uniform superposition state. The Hamiltonian is given by
\begin{equation}
    H_{\mathrm{TFIM}}=H_{z z}+\Delta H_x, \quad H_{z z}=\sum_{i=1}^3 Z_i Z_{i+1}, \quad H_x=\sum_{i=1}^3 X_i, \quad \Delta>0,
\end{equation}
and PQC is
\begin{equation}\label{eq-1009}
    U (\boldsymbol{\theta})=\exp \left (-i \frac{\theta_4}{2} H_{x}\right)\exp \left (-i \frac{\theta_3}{2} H_{z z}\right)\exp \left (-i \frac{\theta_2}{2} H_{x}\right)\exp \left (-i \frac{\theta_1}{2} H_{zz}\right).
\end{equation}
We will show that the function $f (\boldsymbol{\theta})$, with respect to each parameter $\theta_j$, has a singleton frequency $\Omega=\{2\}$. In fact, this property holds for the TFIM model with an arbitrary number of qubits and any number of HVA layers, exhibiting an $\mathcal{O} (1)$ frequency. While our analysis is limited to the 3-qubit, 2-layer setting, the same reasoning extends to general cases. The analysis is based on \cref{app-trig}. Please refer to it and the related discussion in \cref{re-sparse-freq}.

\subsection{Preliminary}

To prepare for the forthcoming proofs, we begin with a few foundational observations. Recall that the Pauli $X$ operator satisfies $X\lvert s\rangle = s \lvert s\rangle$, where $s = \pm1$, and $\lvert+\rangle = \frac{1}{\sqrt2} (\lvert0\rangle + \lvert1\rangle)$, $\lvert-\rangle = \frac{1}{\sqrt2} (\lvert0\rangle - \lvert1\rangle)$. Therefore, on the product state $\lvert s_1 s_2 s_3\rangle = \lvert s_1\rangle \otimes \lvert s_2\rangle \otimes \lvert s_3\rangle$, the operator $H_x = \sum_{i=1}^3 X_i$ acts as
\begin{equation}
    H_x \lvert s_1 s_2 s_3\rangle = (s_1 + s_2 + s_3) \lvert s_1 s_2 s_3\rangle,
\end{equation}
for all $s_i = \pm1$. That is, the eigenvalues and corresponding eigenstates of $H_x$ are
\begin{equation}
    \begin{aligned}
    &\lambda = +3: \quad \lvert+++\rangle, \\
    &\lambda = +1: \quad \lvert-++\rangle, \lvert+-+\rangle, \lvert++-\rangle, \\
    &\lambda = -1: \quad \lvert+--\rangle, \lvert-+-\rangle, \lvert--+\rangle, \\
    &\lambda = -3: \quad \lvert---\rangle.
    \end{aligned}
\end{equation}
Here,
\begin{equation}\label{eq-S3}
    \mathcal{S}_3: =\left\{\left|s_1 s_2 s_3\right\rangle: s_i= \pm 1\right\}
\end{equation}
is an orthonormal basis of $\left (\mathbb{C}^2\right)^{\otimes 3}$, referred to as the $\pm$ eigenbasis. For convenience, we use bold symbol $\boldsymbol{s}=$ $s_1 s_2 s_3$ as an index to label the elements of this basis. This is analogous to a binary representation but using $\pm$ symbols. Here, $\boldsymbol{s}$ also refer to the basis state $|\boldsymbol{s}\rangle=\left|s_1 s_2 s_3\right\rangle$. 
On the other hand, for the operator $\exp\left (-i \tfrac{\theta}{2} H_x \right)$, we have
\begin{equation}
    \lvert s_1 s_2 s_3 \rangle \xrightarrow{\exp (-i \tfrac{\theta}{2} H_x)}
    e^{-i \tfrac{\theta}{2} (s_1 + s_2 + s_3)} \lvert s_1 s_2 s_3 \rangle, \quad \forall \boldsymbol{s} \in \mathcal{S}_3.
\end{equation}
In this case, each basis state picks up only a phase.

Recall also that the Pauli $Z$ operator swaps $\lvert\pm\rangle$ via $Z\lvert\pm\rangle = \lvert\mp\rangle$. Therefore, for example,
\begin{equation}
    \lvert s_1 s_2 s_3\rangle \xrightarrow{Z_1 Z_2} \lvert\, \overline{s_1} \, \overline{s_2} \, s_3\rangle, \quad \forall \boldsymbol{s} \in \mathcal{S}_3,
\end{equation}
where $\overline{s_i} = - s_i$. The same applies to $Z_2 Z_3$ and $Z_3 Z_1$. It is known that if $A$ is a matrix such that $A^2=I$, then $\exp (i A x)=\cos (x) I+i \sin (x) A, $ as shown in \cite[Exercise 4.2]{nielsen2010quantum}. Note that $ (Z_1 Z_2)^2 = I$, so
\begin{equation}
    \exp\left (-i \tfrac{\theta}{2} Z_1 Z_2\right) = \cos\left (\tfrac{\theta}{2}\right) I - i \sin\left (\tfrac{\theta}{2}\right) Z_1 Z_2.
\end{equation}
Similarly, $\exp\left (-i \tfrac{\theta}{2} Z_2 Z_3\right) = \cos\left (\tfrac{\theta}{2}\right) I - i \sin\left (\tfrac{\theta}{2}\right) Z_2 Z_3, $ and $\exp\left (-i \tfrac{\theta}{2} Z_3 Z_1\right) = \cos\left (\tfrac{\theta}{2}\right) I - i \sin\left (\tfrac{\theta}{2}\right) Z_3 Z_1$. Therefore, any state $\boldsymbol{s} \in \mathcal{S}_3$ can evolve as follows:
\begin{equation}\label{eq-1234}
    \lvert s_1 s_2 s_3 \rangle
    \xrightarrow{\exp (-i \frac{\theta}{2} Z_1 Z_2)}
(*) \lvert s_1 s_2 s_3 \rangle + (*) \lvert \overline{s_1}\, \overline{s_2}\, s_3 \rangle,
\end{equation}
where $ (*)$ denotes coefficients that depend on $\theta$ but are not important for the structural argument. Applying the next $\exp (-i \frac{\theta}{2} Z_2 Z_3)$ to the right hand side of \cref{eq-1234}:
\begin{equation}\label{eq-12551}
    \xrightarrow{\exp (-i \frac{\theta}{2} Z_2 Z_3)}
(*) \lvert s_1 s_2 s_3 \rangle + (*) \lvert s_1\, \overline{s_2}\, \overline{s_3} \rangle
    + (*) \lvert \overline{s_1}\, \overline{s_2}\, s_3 \rangle + (*) \lvert \overline{s_1}\, s_2\, \overline{s_3} \rangle,
\end{equation}
and moreover, applying $\exp (-i\frac{\theta}{2}Z_3Z_1)$ to the right hand of above equation:
\begin{align}
& \xrightarrow{\exp (-i\frac{\theta}{2}Z_3Z_1)} 
[ (*)\lvert s_1s_2s_3\rangle + (*) \lvert \overline{s_1}\, s_2\, \overline{s_3}\rangle]
+[ (*) \lvert s_1\, \overline{s_2}\, \overline{s_3}\rangle
+ (*) \lvert \overline{s_1}\, \overline{s_2}\, s_3\rangle]
\notag \\
& \qquad\qquad\qquad 
\quad +[ (*) \lvert \overline{s_1}\, \overline{s_2}\, s_3\rangle
+ (*) \lvert s_1\, \overline{s_2}\, \overline{s_3}\rangle]
+[ (*) \lvert \overline{s_1}\, s_2\, \overline{s_3}\rangle
+ (*) \lvert s_1\, s_2\, s_3\rangle]
\notag \\
&\qquad\qquad\qquad = (*)\lvert s_1s_2s_3\rangle + (*) \lvert s_1\, \overline{s_2}\, \overline{s_3}\rangle
+
(*) \lvert \overline{s_1}\, \overline{s_2}\, s_3\rangle+ (*) \lvert \overline{s_1}\, s_2\, \overline{s_3}\rangle . \label{eq-12552}
\end{align}
Note that the right hand sides of \cref{eq-12551,eq-12552} are both combinations of the same four basis state.
Thus, the operator
\begin{equation}
    \exp\left (-i \tfrac{\theta}{2} H_{zz} \right)=\exp \left (-i \frac{\theta}{2} Z_3 Z_1\right)\exp \left (-i \frac{\theta}{2} Z_2 Z_3\right)\exp \left (-i \frac{\theta}{2} Z_1 Z_2\right)
\end{equation}
acts on the $\lvert s_1 s_2 s_3 \rangle$ by mapping it into a superposition of four basis states. Indeed, the full set $\mathcal{S}_3$ is partitioned into two such closed subsets:
\begin{align}
    \mathcal{S}_3^{ (+)} &=\left\{\lvert+++\rangle, \lvert--+\rangle, \lvert+--\rangle, \lvert-+-\rangle\right\}, \\
 \mathcal{S}_3^{ (-)} &=\left\{\lvert---\rangle, \lvert++-\rangle, \lvert-++\rangle, \lvert+-+\rangle\right\} .
\end{align}
For any fixed $|\boldsymbol{s}\rangle$, applying any $Z_i Z_{i+1}$ operators maps $\boldsymbol{s}$ to another one within the same subset. Similarly, applying any $\exp \left (-i \frac{\theta}{2} H_{z z}\right)$ operators to $\boldsymbol{s}$ results in a linear combination of elements from the corresponding subset. $\mathcal{S}_3^{ (\pm)}$ form the closed subsets under such actions.

\subsection{Frequency for $\theta_4$}

For brevity, we detail only the analysis for $\theta_2$ and $\theta_4$, as the proofs for $\theta_1$ and $\theta_3$ proceed in a similar manner.
We first consider the univariate function $f (\theta_4)$ defined by
\begin{equation}
    \theta_4 \mapsto \langle\psi| \exp \left (-i \frac{\theta_4}{2} H_x\right)^{\dagger} O \exp \left (-i \frac{\theta_4}{2} H_x\right)|\psi\rangle,
\end{equation}
where
\begin{align}
    |\psi\rangle & = \exp \left (-i \frac{\theta_3}{2} H_{z z}\right) \exp \left (-i \frac{\theta_2}{2} H_x\right) \exp \left (-i \frac{\theta_1}{2} H_{z z}\right) \left|\psi_0\right\rangle, \\
    O &= H_{zz} + \Delta H_x.
\end{align}
Next, we determine the matrix representation $[O]$ of the observable $O$ in the eigenbasis $\mathcal{S}_3$ of $H_x$, i.e.,
\begin{equation}
[O]_{\boldsymbol{s^{\prime}}, \boldsymbol{s}}: =\left\langle \boldsymbol{s^{\prime}}\right| O\left|\boldsymbol{s}\right\rangle, \quad \forall \boldsymbol{s^{\prime}}, \boldsymbol{s} \in \mathcal{S}_3.
\end{equation}
Here, we directly use $\boldsymbol{s}$ to label the matrix indices; see \cref{eq-S3}. Specifically, we focus on the upper triangular part of $[O]$ (excluding the diagonal) and aim to identify the nonzero entries (see \cref{app-trig}). From the positions of these nonzero off-diagonal elements, we can extract the corresponding differences in the eigenvalues of $H_x$, which reveal the actual frequency components.

Note that in the eigenbasis $\mathcal{S}_3$, the matrix $\left[\Delta H_x\right]$ becomes diagonal and can thus be ignored for our purposes.
We begin by examining the term $Z_1 Z_2$ in the Hamiltonian $H_{z z}=\sum_{i=1}^3 Z_i Z_{i+1}$. For any $\boldsymbol{s}^{\prime}, \boldsymbol{s} \in \mathcal{S}_3$, we have
\begin{equation}
[Z_1 Z_2]_{\boldsymbol{s'}, \boldsymbol{s}}
    = \langle s'_1 s'_2 s'_3 \lvert Z_1 Z_2 \lvert s_1 s_2 s_3\rangle
    = \langle s'_1 s'_2 s'_3 \lvert\, \overline{s_1}\, \overline{s_2}\, s_3\rangle,
\end{equation}
which is nonzero (equal to 1) only if $s_1^{\prime}=\overline{s_1}, s_2^{\prime}=\overline{s_2}$, and $s_3^{\prime}=s_3$; that is, $s^{\prime}$ and $s$ differ only at positions 1 and 2. In this case, the difference between the corresponding eigenvalues of $H_x$ is
\begin{equation}\label{eq-1005}
    | (s_1 + s_2 + s_3) - (s'_1 + s'_2 + s'_3)| = |2 (s_1 + s_2)| \in \{0, 4\}.
\end{equation}
Thus, an element of $\left[Z_1 Z_2\right]$ is 1 if and only if the indices $\boldsymbol{s}$ and $\boldsymbol{s}^{\prime}$ differ exactly at the positions specified by $Z_1 Z_2$; otherwise, the element is zero. The same pattern applies to $\left[Z_2 Z_3\right]$ and $\left[Z_3 Z_1\right]$ as well. We summarize the results below, with different colors indicating different $Z Z$ terms. Entries below the diagonal are omitted.
\begin{align}
&[O]=\textcolor{red}{[Z_1 Z_2]} +\textcolor{mydarkgreen}{[Z_2 Z_3]} +\textcolor{blue}{[Z_3 Z_1]}+[\Delta H_x] \\
&=     
\bordermatrix{
        & +++ & ++- & +-+ & +-- & -++ & -+- & --+ & --- \cr
+++     &  *  &  0  &  0  &  \ZBC  &  0  &  \ZCA  &  \ZAB  &  0  \cr
++-     &     &  *  &  \ZBC &  0  &  \ZCA  &  0  &  0  &  \ZAB  \cr
+-+     &     &     &  *  &  0  &  \ZAB  &  0  &  0  &  \ZCA  \cr
+--     &     &     &     &  *  &  0  &  \ZAB  &  \ZCA  &  0  \cr
-++     &     &     &     &     &  *  &  0  &  0  &  \ZBC  \cr
-+-     &     &\scalebox{3}{$*$}&     &     &     &  *  &  \ZBC  &  0  \cr
--+     &     &     &     &     &     &     &  *  &  0  \cr
---     &     &     &     &     &     &     &     &  * 
}.
\end{align}

We observe that the three matrices $[Z_1 Z_2]$, $[Z_2 Z_3]$, and $[Z_3 Z_1]$ have no overlapping nonzero entries. According to \cref{eq-1005}, this implies that the set of \textit{actual} frequencies contains only the value 4. Taking into account the scaling factor $1/2$ in \cref{eq-1009}, the resulting frequency is 2. This can also be verified directly from the matrices above. For example, the $ (1, 4)$-th entry is 1, corresponding to two eigenvalues, 3 and $-1$, whose absolute difference is 4.

On the other hand, if we consider only the eigenvalues of $H_x$, which are $\{-3, -1, 1, 3\}$, the possible frequency differences are $\{2, 4, 6\}$. However, the coefficients corresponding to frequencies 2 and 6 vanish, so they do not contribute as effective frequencies. For instance, the $ (1, 2)$-th entry corresponds to an eigenvalue difference of 2, and the $ (1, 8)$-th entry corresponds to 6, but in both cases the matrix elements are zero.

\subsection{Frequency for $\theta_2$}

We next consider the univariate function $f (\theta_2)$ as
\begin{equation}
    \theta_2 \mapsto \langle\psi| \exp \left (-i \frac{\theta_2}{2} H_x\right)^{\dagger} O\exp \left (-i \frac{\theta_2}{2} H_x\right)|\psi\rangle,
\end{equation}
where
\begin{align}
    |\psi\rangle &=  \exp \left (-i \frac{\theta_1}{2} H_{z z}\right) \left|\psi_0\right\rangle, \\
    O &= \left (\exp \left (-i \frac{\theta_4}{2} H_x\right) \exp \left (-i \frac{\theta_3}{2} H_{z z}\right)\right)^\dagger (H_{zz} + \Delta H_x)\left (\exp \left (-i \frac{\theta_4}{2} H_x\right) \exp \left (-i \frac{\theta_3}{2} H_{z z}\right)\right).
\end{align}
Next, we determine the matrix representation $[O]$ in the eigenbasis $\mathcal{S}_3$: $\forall \boldsymbol{s^{\prime}}, \boldsymbol{s} \in \mathcal{S}_3, $
\begin{align}
[O]_{\boldsymbol{s^{\prime}}, \boldsymbol{s}}
& =\left\langle \boldsymbol{s^{\prime}}\right| O\left|\boldsymbol{s}\right\rangle, \\
& =\left (\exp \left (-i \frac{\theta_4}{2} H_x\right) \exp \left (-i \frac{\theta_3}{2} H_{z z}\right)\left|\boldsymbol{s^{\prime}}\right\rangle\right)^\dagger (H_{zz} + \Delta H_x)\left (\exp \left (-i \frac{\theta_4}{2} H_x\right) \exp \left (-i \frac{\theta_3}{2} H_{z z}\right)\right)\left|\boldsymbol{s}\right\rangle
\end{align}
The operator $O$ is more involved in this setting and calls for a more elaborate analysis. If $\boldsymbol{s} \in \mathcal{S}_3^{ (+)}$, then
\begin{equation}
    |\boldsymbol{s} \rangle
    \xrightarrow{\exp \left (-i \frac{\theta_3}{2} H_{z z}\right)} \sum _{\boldsymbol{k} \in \mathcal{S}_3^{ (+)}} (*)| \boldsymbol{k} \rangle
    \xrightarrow{\exp \left (-i \frac{\theta_4}{2} H_x\right)} \sum _{\boldsymbol{k} \in \mathcal{S}_3^{ (+)}} (*)| \boldsymbol{k} \rangle
    \xrightarrow{\Delta H_x} \sum _{\boldsymbol{k} \in \mathcal{S}_3^{ (+)}} (*)| \boldsymbol{k} \rangle
    \xrightarrow{H_{z z}} \sum _{\boldsymbol{k} \in \mathcal{S}_3^{ (+)}} (*)| \boldsymbol{k} \rangle.
\end{equation}
The above conclusions remain valid if $\mathcal{S}_3^{ (+)}$ is replaced by $\mathcal{S}_3^{ (-)}$throughout. Hence, if $\boldsymbol{s}$ and $\boldsymbol{s}^{\prime}$ do not lie in the same subset, the corresponding matrix element vanishes due to orthogonality. This is illustrated in the matrix below.
\begin{align*}
[O]
=
\bordermatrix{
        & +++ & ++- & +-+ & +-- & -++ & -+- & --+ & --- \cr
+++     &  *  &  0  &  0  &  *  &  0  &  *  &  *  &  0  \cr
++-     &  0  &  *  &  *  &  0  &  *  &  0  &  0  &  *  \cr
+-+     &  0  &  *  &  *  &  0  &  *  &  0  &  0  &  *  \cr
+--     &  *  &  0  &  0  &  *  &  0  &  *  &  *  &  0  \cr
-++     &  0  &  *  &  *  &  0  &  *  &  0  &  0  &  *  \cr
-+-     &  *  &  0  &  0  &  *  &  0  &  *  &  *  &  0  \cr
--+     &  *  &  0  &  0  &  *  &  0  &  *  &  *  &  0  \cr
---     &  0  &  *  &  *  &  0  &  *  &  0  &  0  &  *  \cr
}.
\end{align*}
When $\boldsymbol{s}$ and $\boldsymbol{s}^{\prime}$ lie in the same subset, it can be readily verified that the difference of their associated eigenvalues belongs to $\{0, 4\}$. The remaining arguments are the same as previous subsection, and this completes the proof.

\section{Analysis of mean squared error without constant variance assumption}\label{app-mse-nonconst}

In \cref{sec:noise}, we established a mean squared error (MSE) analysis for the estimation of Fourier coefficients and introduced \cref{pro-min-mse} as a criterion for the optimality of interpolation nodes. 
Note that this analysis is predicated on the constant variance \cref{assp-var}. In reality, \cref{assp-var} does not hold, even though it is commonly adopted in the literature \cite{mari2021estimating,wierichs2022general,markovich2024parameter}. On the other hand, optimal interpolation nodes derived under \cref{assp-var} demonstrate good empirical performance in numerical simulations of \cref{sec-experiments}, which seems to be a contradictory phenomenon.
In this appendix, we remove the constant variance \cref{assp-var} and directly analyze the true MSE. We then explore its connection to \cref{pro-min-mse}, providing justification for why the constant variance condition can still be a reasonable approximation in practice.

We adopt the notation in \cref{sec:noise}. Without \cref{assp-var}, each evaluation of the cost function returns a noisy sample represented by the random variable
\begin{equation}
    \tilde{f}(x) = f(x) + \epsilon_x, \quad \epsilon_x \sim \mathcal{N}\left(0, \frac{\sigma^2(x)}{\mathfrak{n}}\right),
\end{equation}
where number of shots $\mathfrak{n}$ is sufficiently large. In what follows, let $\sigma^2(x)$ absorb the constant factor $\frac{1}{\mathfrak{n}}$. Given the noisy data $\{ (x_i, \tilde{f}(x_i)) \}_{i=0}^{2r}$, we solve the perturbed system $A_{\mathbf{x}} \hat{\mathbf{z}}_{\mathbf{x}} = \mathbf{y}_{\mathbf{x}} + \mathbf{e},$ where $\mathbf{e} \sim \mathcal{N}\left(0, \Sigma_{\mathbf{x}} \right)$, with
\begin{equation}
    \Sigma_{\mathbf{x}} = \operatorname{diag}(\sigma^2(x_0), \dots, \sigma^2(x_{2r})).
\end{equation}
The estimator $\hat{\mathbf{z}}_{\mathbf{x}}=\mathbf{z}+A_{\mathbf{x}}^{-1} \mathbf{e}$ is still unbiased, i.e., $\operatorname{E}[\hat{\mathbf{z}}_{\mathbf{x}}] = \mathbf{z},$ and its covariance matrix becomes
\begin{align}
\operatorname{VAR}[\hat{\mathbf{z}}_{\mathbf{x}}]
= A{\mathbf{x}}^{-1} \operatorname{E}[\mathbf{e} \mathbf{e}^{\dagger}] (A_{\mathbf{x}}^{-1})^{\dagger} \
= A_{\mathbf{x}}^{-1} \Sigma_{\mathbf{x}} (A_{\mathbf{x}}^{-1})^{\dagger}.
\end{align}
Accordingly, the true mean squared error (MSE) is
\begin{equation}\label{eq-true-mse}
F(\mathbf{x}) := \operatorname{tr} (\operatorname{VAR}[\hat{\mathbf{z}}_{\mathbf{x}}]) = \sum_{i=0}^{2r} \sigma^2(x_i) \left\| [A_{\mathbf{x}}^{-1}]_{i,:} \right\|^2,
\end{equation}
where $\left[A_{\mathbf{x}}^{-1}\right]_{i,:}$ denotes the $i$-th row of matrix $A_{\mathbf{x}}^{-1}.$ We know that the expression of $\sigma^2(\cdot)$ is given in \cref{eq-4280}, and in principle it can be computed. However, the high computational cost makes minimizing $F(\mathbf{x})$ impractical in practice. For this reason, we aim to ignore the $\sigma^2$-dependent terms. Now define
\begin{equation}
    F_{\text{const}}(\mathbf{x}) =  \operatorname{tr}( A_{\mathbf{x}}^{-1} (A_{\mathbf{x}}^{-1})^{\dagger}) = \sum_{i=0}^{2r} \left\| [A_{\mathbf{x}}^{-1}]_{i,:} \right\|^2,
\end{equation}
which is exactly the function obtained in \cref{pro-min-mse} by factoring out the constant $\sigma^2$. At this point, minimizing $F_{\text {const }}(\mathbf{x})$ becomes straightforward.
For any given $\mathbf{x} \in \mathbb{R}^n$, define
\begin{equation}
    \sigma_{\min}^2(\mathbf{x}) = \min_{i=0,\ldots,2r} \sigma^2(x_i), \quad \sigma_{\max}^2(\mathbf{x}) = \max_{i=0,\ldots,2r} \sigma^2(x_i).
\end{equation}
Then, the true MSE satisfies the bounds
\begin{equation}
    \sigma_{\min}^2(\mathbf{x}) \; F_{\text{const}}(\mathbf{x}) \leq F(\mathbf{x}) \leq \sigma_{\max}^2(\mathbf{x}) \;F_{\text{const}}(\mathbf{x}).
\end{equation}
In other words, for any nodes $\mathbf{x}$, we have the ratio
\begin{equation}
    R(\mathbf{x}) := \frac{F(\mathbf{x})}{F_{\text{const}}(\mathbf{x})} \in \left[ \sigma_{\min}^2(\mathbf{x}), \sigma_{\max}^2(\mathbf{x}) \right].
\end{equation}
Let $\mathbf{x}^* = \arg \min_{\mathbf{x}} F(\mathbf{x}),$ and $\mathbf{x}_{\mathrm{const}}^* = \arg \min_{\mathbf{x}} F_{\text{const}}(\mathbf{x}),$ then we have
\begin{equation}
    \begin{aligned}
    F(\mathbf{x}_{\mathrm{const}}^*) 
    = R(\mathbf{x}_{\mathrm{const}}^*)  F_{\text{const}}(\mathbf{x}_{\mathrm{const}}^*) 
    &\leq R(\mathbf{x}_{\mathrm{const}}^*) F_{\text{const}}(\mathbf{x}^*) \\
    &= \frac{R(\mathbf{x}_{\mathrm{const}}^*)} {R(\mathbf{x}^*)}R(\mathbf{x}^*) F_{\text{const}}(\mathbf{x}^*)\\
    &= \frac{R(\mathbf{x}_{\mathrm{const}}^*)}{R(\mathbf{x}^*)}  F(\mathbf{x}^*).
    \end{aligned}
\end{equation}
Therefore, the worst-case approximation factor is bounded by
\begin{equation}
    \frac{F(\mathbf{x}_{\mathrm{const}}^*)}{F(\mathbf{x}^*)} \leq \frac{R(\mathbf{x}_{\mathrm{const}}^*)}{R(\mathbf{x}^*)} \leq \frac{\sigma_{\max}^2(\mathbf{x}_{\mathrm{const}}^*)}{\sigma_{\min}^2(\mathbf{x}^*)}.
\end{equation}
If we further assume that the variances satisfy $ \sigma^2(x) \in [\,\underline{\sigma}^2, \bar{\sigma}^2] $ for any $x$, which is true for equidistant frequency case, then the bound simplifies to
\begin{equation}
    F(\mathbf{x}_{\mathrm{const}}^*) \leq \frac{\bar{\sigma}^2}{\underline{\sigma}^2} F(\mathbf{x}^*),
\end{equation}
In other words, minimizing $F_{\mathrm{const}}$ yields a constant-factor approximation to the true minimization of $F$: the solution $\mathbf{x}_{\mathrm{const}}^*$ can only be worse than the true minimizer $\mathbf{x}^*$ by at most a factor of $\bar{\sigma}^2 / \underline{\sigma}^2$. The value $\bar{\sigma}^2 / \underline{\sigma}^2$ should be determined based on the specifics of the problem at hand.

Finally, we compare the performance of the ICD algorithm with and without \cref{assp-var} using the MaxCut problem with HEA from \cref{subsec-num-problems}. The results are shown in \cref{fig:mse_no_const_var_2}. We set the number of shots to 1024 and use standard ICD with sequentially update. The ``Actual MSE'' represents the true optimal interpolation node obtained at each iteration by minimizing $F(\mathbf{x})$ in \cref{eq-true-mse}. Since $F(\mathbf{x})$ involves a complicated variance expression $\sigma^2(x)$, an analytical solution is hard to obtain, so we use a differential evolution solver to numerically compute its minimum $F(\mathbf{x}^*)$. The ``Constant MSE'' corresponds to standard ICD using $2\pi/3$ equally spaced interpolation nodes starting from zero (i.e., $\mathbf{x}_{\mathrm{const}}^*$). \cref{fig:mse_no_const_var_2} shows that the performance of the two methods is nearly indistinguishable. This observation is consistent with Result I in \cref{subsec-num-re-1}: when the number of shots is sufficiently large, even suboptimal interpolation nodes can ensure algorithm convergence. 

In \cref{fig:mse_no_const_var_1}, we show the relative MSE error at each iteration, given by
\begin{equation}
    \frac{F\left(\mathbf{x}_{\text{const}}^*\right) - F\left(\mathbf{x}^*\right)}{F\left(\mathbf{x}^*\right)} \leq \frac{\bar{\sigma}^2}{\underline{\sigma}^2} - 1.
\end{equation}
Our numerical results show that $\frac{\bar{\sigma}^2}{\underline{\sigma}^2} - 1$ is always less than or equal to 1. This implies that the true MSE corresponding to the $2\pi/3$ equally spaced nodes under the constant variance assumption is at most twice the true minimum MSE. Therefore, the $2\pi/3$ equally spaced point $\mathbf{x}_{\text{const}}^*$ can be considered near-optimal.


\begin{figure}[htbp]
  \begin{subfigure}{0.45\textwidth}
    \includegraphics[width=0.91\linewidth]{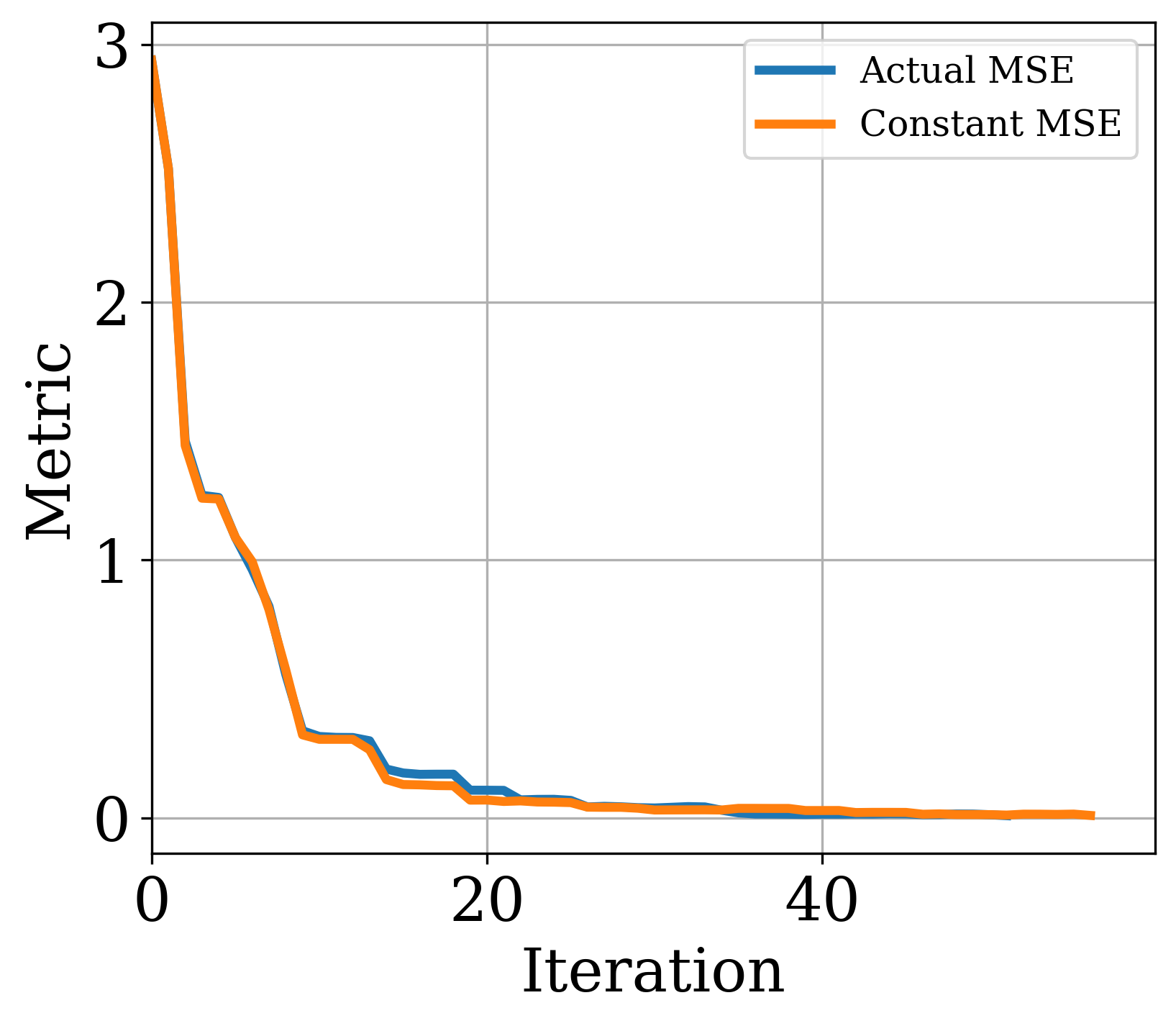}
    \caption{}
    \label{fig:mse_no_const_var_1}
  \end{subfigure}
  \begin{subfigure}{0.45\textwidth}
    \centering
    \includegraphics[width=1\linewidth]{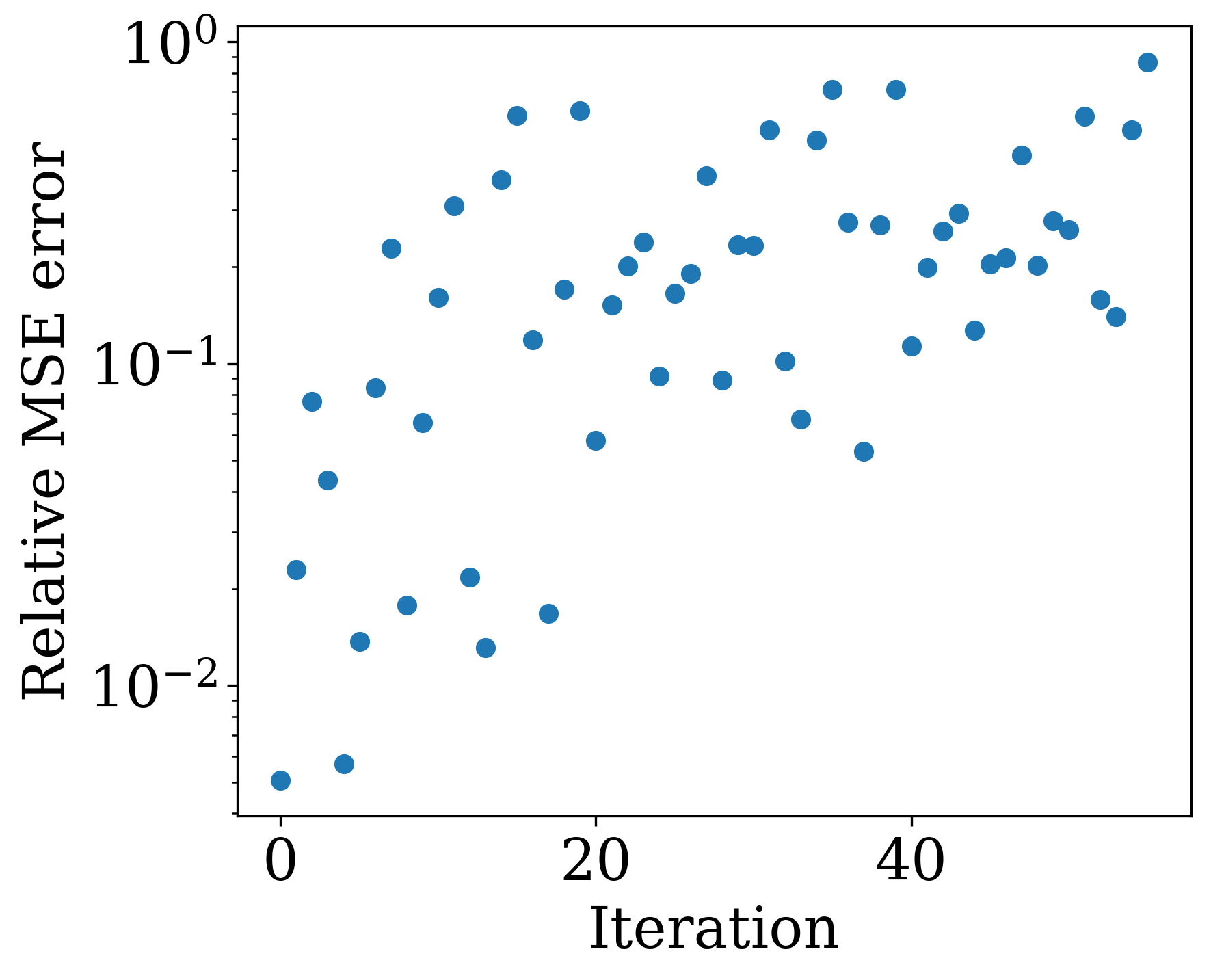} 
    \caption{}
    \label{fig:mse_no_const_var_2}
  \end{subfigure}
  \caption{
  Comparison between standard ICD with and without constant variance \cref{assp-var} on the MaxCut problem. (a) Performance for the true optimal interpolation node (Actual MSE) and the equally spaced point under constant variance assumption (Constant MSE). (b) Relative MSE error $\frac{F\left(\mathbf{x}_{\text{const}}^*\right) - F\left(\mathbf{x}^*\right)}{F\left(\mathbf{x}^*\right)}$ at each iteration, which is theoretically bounded by $\frac{\bar{\sigma}^2}{\underline{\sigma}^2} - 1$.
}
\end{figure}

\section{Quantum circuits in numerical simulation}\label{app_qc}

\begin{figure}[htbp]
    \centering
    \scalebox{0.8}{
    \Qcircuit @C=1em @R=0.8em @!R { \\
    \nghost{{q}_{1} :  } & \lstick{{q}_{1}} &\gate{R_Y(\theta_0)} &\ctrl{1} &\qw &\qw &\gate{R_Y(\theta_4)} &\ctrl{1} &\qw &\qw &\gate{R_Y(\theta_8)} &\ctrl{1} &\qw &\qw &\gate{R_Y(\theta_{12})} &\ctrl{1} &\qw &\qw&\gate{R_Y(\theta_{16})} &\ctrl{1} &\qw &\qw & \qw \\
    \nghost{{q}_{2} :  } & \lstick{{q}_{2}} &\gate{R_Y(\theta_1)} &\control\qw &\ctrl{1} &\qw &\gate{R_Y(\theta_5)} &\control\qw &\ctrl{1} &\qw &\gate{R_Y(\theta_9)} &\control\qw &\ctrl{1} &\qw &\gate{R_Y(\theta_{13})} &\control\qw &\ctrl{1} &\qw &\gate{R_Y(\theta_{17})} &\control\qw &\ctrl{1} &\qw & \qw \\
    \nghost{{q}_{3} :  } & \lstick{{q}_{3}} &\gate{R_Y(\theta_2)} &\qw &\control\qw &\ctrl{1} &\gate{R_Y(\theta_6)} &\qw &\control\qw &\ctrl{1} &\gate{R_Y(\theta_{10})} &\qw &\control\qw &\ctrl{1} &\gate{R_Y(\theta_{14})} &\qw &\control\qw &\ctrl{1}&\gate{R_Y(\theta_{18})} &\qw &\control\qw &\ctrl{1} & \qw \\
    \nghost{{q}_{4} :  } & \lstick{{q}_{4}} &\gate{R_Y(\theta_3)} &\qw &\qw &\control\qw &\gate{R_Y(\theta_7)} \qw &\qw &\qw &\control\qw &\gate{R_Y(\theta_{11})} &\qw &\qw &\control\qw &\gate{R_Y(\theta_{15})} \qw &\qw &\qw &\control\qw&\gate{R_Y(\theta_{19})} \qw &\qw &\qw & \control \qw & \qw \\
    \\ }}
    \caption{The HEA quantum circuit for the MaxCut problem with $N = 4$ and $p = 5$.}
    \label{fig:qc_maxcut}
\end{figure}

\begin{figure}[htbp]
    \centering
    \scalebox{0.8}{
    \Qcircuit @C=1em @R=0.8em @!R { \\
    \nghost{{q}_{1} :  } & \lstick{{q}_{1}} &\gate{H} &\qw &\multigate{1}{R_{ZZ}(\beta)}  &\qw  &\qw &\qw &\sgate{R_{ZZ}(\beta)}{3} &\qw &\qw &\gate{R_X(\gamma)} &\qw \\
    \nghost{{q}_{2} :  } & \lstick{{q}_{2}} &\gate{H} &\qw &\ghost{R_{ZZ}(\beta)} &\multigate{1}{R_{ZZ}(\beta)} &\qw &\qw &\qw &\qw &\qw &\gate{R_X(\gamma)} &\qw \\
    \nghost{{q}_{3} :  } & \lstick{{q}_{3}} &\gate{H} &\qw &\qw &\ghost{R_{ZZ}(\beta)} &\multigate{1}{R_{ZZ}(\beta)}&\qw &\qw &\qw &\qw &\gate{R_X(\gamma)} &\qw \\
    \nghost{{q}_{4} :  } & \lstick{{q}_{4}} &\gate{H} &\qw &\qw &\qw &\ghost{R_{ZZ}(\beta)}&\qw &\gate{R_{ZZ}(\beta)}&\qw &\qw &\gate{R_X(\gamma)} &\qw \gategroup{2}{5}{5}{9}{1.2em}{--}\gategroup{2}{12}{5}{12}{1.2em}{--}\\
    \nghost{{q}_{5} :  } & & & & &\mbox{$\quad\quad\quad\quad\quad \beta$} & & & & & &\mbox{$\gamma$} & &\\
    \\ }}
    \caption{The HVA quantum circuit for the TFIM model with $N = 4$ and $p = 1$.}
    \label{fig:qc_tifm_HVA}
\end{figure}

\begin{figure}[htbp]
    \centering
    \scalebox{0.65}{
\Qcircuit @C=1em @R=0.8em @!R { \\
\nghost{{q}_{1} :  } & \lstick{q_{1}} &\gate{X} &\gate{H} &\ctrl{1} &\qw &\qw &\qw &\sgate{R_{ZZ}(\theta)}{5} &\qw &\qw &\qw &\sgate{R_{YY}(\phi)}{5}&\qw &\sgate{R_{XX}(\phi)}{5} &\qw &\multigate{1}{R_{ZZ}(\beta)} &\qw &\multigate{1}{R_{YY}(\gamma)} &\multigate{1}{R_{XX}(\gamma)} &\qw &\qw \\
\nghost{{q}_{2} :  } & \lstick{q_{2}} &\gate{X} &\qw &\targ &\qw &\qw &\multigate{1}{R_{ZZ}(\theta)} &\qw &\qw &\qw &\multigate{1}{R_{YY}(\phi)} &\qw &\multigate{1}{R_{XX}(\phi)} &\qw &\qw &\ghost{R_{ZZ}(\beta)} &\qw &\ghost{R_{YY}(\gamma)} &\ghost{R_{XX}(\gamma)} &\qw &\qw \\
\nghost{{q}_{3} :  } & \lstick{q_{3}} &\gate{X} &\gate{H} &\ctrl{1} &\qw &\qw &\ghost{R_{ZZ}(\theta)}&\qw &\qw &\qw &\ghost{R_{YY}(\phi)} &\qw &\ghost{R_{XX}(\phi)} &\qw &\qw &\multigate{1}{R_{ZZ}(\beta)} &\qw &\multigate{1}{R_{YY}(\gamma)} &\multigate{1}{R_{XX}(\gamma)} &\qw &\qw \\
\nghost{{q}_{4} :  } & \lstick{q_{4}} &\gate{X} &\qw &\targ &\qw &\qw &\multigate{1}{R_{ZZ}(\theta)} &\qw &\qw &\qw &\multigate{1}{R_{YY}(\phi)}&\qw &\multigate{1}{R_{XX}(\phi)} &\qw &\qw &\ghost{R_{ZZ}(\beta)} &\qw &\ghost{R_{YY}(\gamma)} &\ghost{R_{XX}(\gamma)} &\qw &\qw \\
\nghost{{q}_{5} :  } & \lstick{q_{5}} &\gate{X} &\gate{H} &\ctrl{1} &\qw &\qw &\ghost{R_{ZZ}(\theta)}&\qw &\qw &\qw &\ghost{R_{YY}(\phi)} &\qw &\ghost{R_{XX}(\phi)} &\qw &\qw &\multigate{1}{R_{ZZ}(\beta)}&\qw &\multigate{1}{R_{YY}(\gamma)} &\multigate{1}{R_{XX}(\gamma)} &\qw &\qw \\
\nghost{{q}_{6} :  } & \lstick{q_{6}} &\gate{X} &\qw &\targ &\qw &\qw &\qw &\gate{R_{ZZ}(\theta)}&\qw &\qw &\qw &\gate{R_{YY}(\phi)} &\qw &\gate{R_{XX}(\phi)} &\qw &\ghost{R_{ZZ}(\beta)} &\qw &\ghost{R_{YY}(\gamma)} &\ghost{R_{XX}(\gamma)} &\qw &\qw \gategroup{2}{7}{7}{9}{1.2em}{--}\gategroup{2}{11}{7}{15}{1.2em}{--}\gategroup{2}{17}{7}{17}{1.2em}{--}\gategroup{2}{11}{7}{15}{1.2em}{--}\gategroup{2}{19}{7}{20}{1.2em}{--}\\
\nghost{{q}_{7} :  } & & & & & & &\mbox{$\quad\quad\quad\quad\quad \theta$} & & & & &\mbox{$\quad\quad\quad\quad\quad \phi$} & & & &\mbox{$\beta$} & &\mbox{$\quad\quad\quad\quad\quad \gamma$}&  & & &\\
\\ }
}
    \caption{The HVA quantum circuit for the XXZ model with $N = 6$ and $p = 1$.}
    \label{fig:qc_XXZ_HVA}
\end{figure}

\end{document}